%% file: main.tex
\documentclass[thesis, nofrontmatter]{thesis-umich-paper} 

\input{Packages}
\usepackage{authblk}

\input{Commands_and_Functions}

\renewcommand{\thesection}{\arabic{section}}
\renewcommand{\thesubsection}{\thesection.\arabic{subsection}}
\renewcommand{\thesubsubsection}{\thesubsection.\arabic{subsubsection}}

\newcounter{suppsection}
\usepackage{tocloft}
\begin{document}
\begin{center}
  {\Large\bfseries Multilevel Primary Aim Analyses of Clustered SMARTs:\\
  \Large\bfseries With Applications in Health Policy\\\par}
  \vspace{0.5cm}
  {\large
    Gabriel Durham,$^{1,2}$
    Anil Battalahalli,$^{1}$
    Amy M. Kilbourne,$^{3,4}$
    Andrew Quanbeck,$^{5}$
    Wenchu Pan,$^{2,6}$
    Tim Lycurgus,$^{2}$
    Daniel Almirall$^{1,2}$
    \par}
    
    \vspace{0.5cm}
    
    {\footnotesize
    $^{1}$Department of Statistics, University of Michigan\\
    $^{2}$Survey Research Center, University of Michigan Institute for Social Research\\
    $^{3}$Department of Learning Health Sciences, University of Michigan Medical School\\
    $^{4}$Office of Research and Development, U.S. Department of Veterans Affairs\\
    $^{5}$Department of Family Medicine and Community Health, University of Wisconsin\\
    $^{6}$Department of Biostatistics, University of Michigan
    \par}

\vspace{0.4cm}
\end{center}

\begin{center}
    \begin{minipage}{0.95\textwidth}

        \noindent
Many health policies or programs can be conceptualized as adaptive interventions. An adaptive intervention is a sequence of decision rules that guide the provision of actions (intervention options) at critical decision points based on the evolving need of recipients, including their response to prior actions. 
In many health policy settings, adaptive interventions target a population of clusters (e.g., schools), with the ultimate intent of impacting outcomes at the level of individuals within the clusters (e.g., mental health care providers in the schools). 
Health policy researchers can use clustered, sequential, multiple assignment, randomized trials (SMARTs) to answer important scientific questions concerning clustered adaptive interventions. A common primary aim is to compare the mean of a nested, end-of-study outcome between two clustered adaptive interventions. 
However, existing methods are not suitable when the primary outcome in a clustered SMART is nested and longitudinal (e.g., repeated outcome measures nested within mental healthcare providers, and mental healthcare providers nested within schools). This manuscript proposes a three-level marginal mean modeling and estimation approach for comparing adaptive interventions in a clustered SMART. The proposed method enables policy analysts to answer a wider array of scientific questions in the marginal comparison of clustered adaptive interventions. Further, relative to using existing two-level methods with a nested, but non-longitudinal, end-of-study outcome, the proposed method benefits from improved statistical efficiency. With this approach, we examine longitudinal comparisons of adaptive interventions for improving school-based mental healthcare and contrast its performance with existing approaches for studying static end-of-study outcomes. Methods were motivated by the Adaptive School-Based Implementation of CBT (ASIC) study, a clustered SMART designed to construct an adaptive health policy to improve the adoption of evidence-based CBT by mental healthcare professionals in high schools across Michigan.
    \end{minipage}
\end{center}

\begin{NoHyper}
    \renewcommand{\thefootnote}{}
    \footnote{%
        \scriptsize{
        The authors would like to thank fellow colleagues at d3center and the Institute for Social Research, including Libby Benson, Cherolyn Brooks, Ian Burnette, Mason Ferlic, Gillian Mayman, and Yao Song, for their insight and help during the completion of the paper. We would also like to thank Xiang Li for her invaluable assistance throughout the Balanced Opioid Initiative analysis. 
        
        This work was supported by IES grant R324B220001, and NIH grants R01HD095973, P50DA054039, R01DA039901, R01DA047279, and R01MH114203. The views expressed in this article are those of the authors and do not necessarily represent those of the US Government or other public entities.

        This article was published in the Annals of Applied Statistics (2026). The published version is available at \href{https://doi.org/10.1214/26-AOAS2200}{https://doi.org/10.1214/26-AOAS2200}.
        }
    }
    \addtocounter{footnote}{-1}
\end{NoHyper}
\clearpage

\newpage
\tableofcontents
\clearpage
\newpage

\section{Introduction}\label{c1:sec:intro}
Adaptive interventions, also known as dynamic treatment regimes, are protocols used to guide decision-making at critical decision points during intervention \citep{Murphy2003, Collins2004, Laber2014}. This includes guidance on whether and when to modify (e.g., augment, intensify, or switch) the provision of intervention options, as well as what information should inform such decisions.

In health policy settings, intervention often targets a cluster of individuals. A cluster is defined as an intact group of individuals, often formed through naturally occurring organizational or administrative affiliations. For example, in an attempt to improve the behavior of clinicians (e.g., nurse-practitioners, doctors, or mental health providers), a health policy intervention may target hospitals. Doing so can address widespread barriers to patient wellbeing, as well as promote supportive organizational environments to foster development of medical professionals. We call such interventions (i.e., those that act on clusters of individuals but are designed to improve individual outcomes) \textit{clustered} interventions. This manuscript concerns clustered adaptive interventions ({cAIs}); i.e., adaptive interventions for which the sequence of decision rules guiding intervention delivery is based on the baseline conditions and evolving needs of each pre-determined cluster of individuals \citep{NeCamp2017}. Like clustered interventions at large, a defining feature of cAIs is cluster-level action with the intent of impacting outcomes at the individual-level.

Such interventions have natural applications to public policy, as they leverage existing social structures (e.g., schools, hospitals, communities). Furthermore, cAIs are particularly useful in implementation science, which focuses on improving the adoption and fidelity of evidence-based interventions in real-world settings \citep{Bauer2020}. By leveraging pre-existing administrative clusters, such as schools and hospitals, cAIs can help address systemic barriers to effective implementation and support sustainable practice change \citep{Kilbourne2014, Kilbourne2018, Quanbeck2020}.

Sequential, multiple assignment, randomized trials ({SMARTs}) form a class of experimental designs which act as valuable data collection tools for optimizing the construction of adaptive interventions. Through sequential randomization, SMARTs offer intervention designers an opportunity to analyze a multitude of questions concerning adaptive intervention construction \citep{NahumShani2012}.
Clustered SMARTs ({cSMARTs}) are a class of SMARTs which utilize cluster-level randomization and treatment, but for which the primary outcome is measured at the individual level. Subsequently, researchers can use cSMARTs to address important scientific questions preventing the construction of high quality clustered adaptive interventions. Scientists have employed cSMARTs in a wide variety of health application areas, including school-based healthcare \citep{Kilbourne2018}, mental health \citep{Kilbourne2014}, substance use \citep{Quanbeck2020, Fernandez2020}, and infectious disease prevention \citep{Zhou2020}.

A common primary aim in a SMART is the comparison of two (or more) adaptive interventions on the marginal mean of an end-of-study outcome \citep{NahumShani2012}. The foundation of the statistical approach for this aim is rooted in the work of \citet{Orellana2010-gv, Orellana2010-rc, Robins2008-ua}. \citet{Lu2015} and \citet{Li2016} developed analytic methods addressing this aim in the case of a continuous longitudinal outcome, with \citet{Seewald2019} presenting a sample size formula to power a SMART with such a primary aim. Additionally, \citet{Luers2019, Dziak2019} further explore more general methods for longitudinal outcome analyses in these settings.

SMART design and analyses generally remain active areas of statistical research \citep{Artman2024, Wank2024}. While these methodological advances have materially enhanced the design and analysis of individually randomized SMARTs, the literature on clustered SMARTs is still emerging. \citet{NeCamp2017} extended the standard approach to enable comparison of clustered adaptive interventions via the marginal mean of an end-of-study outcome, and developed a corresponding sample size formula. Additionally, \citet{Ghosh2015} proposed a similar sample size for binary end-of-study outcome comparison and \citet{Xu2019} extended these approaches for complex clustering structures. More recently, \citet{Pan_U} developed finite sample adjustments unique to the analysis of clustered SMARTs. Beyond these contributions, however, research on clustered SMARTs remains sparse.

The primary methodological contribution of this manuscript is an approach for comparing the marginal mean of a longitudinal continuous outcome between two clustered adaptive interventions embedded in a cSMART. The nested structure of repeated observations within individuals within clusters induces multiple \quotes{levels} to consider in the analysis. In the analysis of standard clustered randomized control trials, such three-level analytic methods are commonplace \citep{Teerenstra2010}. However, while methods exist for longitudinal outcome analysis in individually randomized SMARTs and for end-of-study outcome comparisons in clustered SMARTs, new methods are needed for analyzing longitudinal outcomes in clustered SMARTs.

The proposed method combines methods for comparing adaptive interventions on a longitudinal outcome \citep{Lu2015, Li2016, Seewald2019} with methods for comparing clustered adaptive interventions \citep{NeCamp2017, Pan_U}. The method offers two important benefits: First, most importantly, it enables the marginal mean comparison of cAIs on a longitudinal outcome, opening the door to a wider array of causal estimands concerning the dynamical effects of cAIs. Second, we provide empirical evidence that incorporating repeated measurements can help improve statistical efficiency even when the primary outcome of interest is a static (i.e., singly measured) end-of-study measure. Thus, analysts interested in the mean comparison of cAIs on an end-of-study outcome have greater statistical precision under the new approach. In many policy settings, recruiting an additional cluster can be more costly than collecting an additional research outcome for all clusters already participating in the trial, underscoring the practical importance of this advantage \citep{Raudenbush1997, Rutterford2015}.


Methods are illustrated using the Adaptive School-Based Implementation of CBT (ASIC) study, a clustered SMART that aims to improve the adoption of cognitive behavioral therapy ({CBT}), an evidence-based mental health treatment, in Michigan high schools. ASIC collected periodic measures of its primary outcome (quantity of CBT delivery) across 10 months \citep{Kilbourne2018}. We note that the proposed approach applies more broadly than two-stage, prototypical, clustered SMART designs such as ASIC's. We discuss extension to more general settings in Appendix \ref{c1:sec:gen_smart}.

\section{Motivating Example: ASIC}\label{c1:sec:mot_ex}\label{c1:sec:mot_ex:asic}
As discussed in Section \ref{c1:sec:intro}, the \textit{Adaptive School-based Implementation of CBT} ({ASIC}) study, a clustered SMART designed to study school-based mental healthcare, motivates our methods \citep{Kilbourne2018}. 

In the United States, youths are, in general, more likely to receive mental health services from schools than from any other child-serving mental healthcare sector \citep{Duong2020}. Depression and anxiety disorders are the most common mental health disorders among American youths. While evidence-based practices ({EBPs}) such as cognitive behavioral therapy ({CBT}) can improve outcomes among these individuals, less than $20\%$ of young patients have access to EBPs. Furthermore, even when EBPs are offered, treatment fidelity can be weak \citep{Kilbourne2018}.

Recently, researchers at the University of Michigan conducted the ASIC study to inform the design of a two-stage clustered adaptive intervention aimed at addressing systemic barriers to CBT delivery in high schools. 
The scientists sought to create an adaptive intervention that combines three existing implementation strategy components for promoting CBT-uptake among school professionals ({SPs}, i.e., school employees tasked with delivering mental health services to students). The existing strategies were: \textit{(i) Replicating Effective Programs ({REP})}, \textit{(ii) Coaching}, and \textit{(iii) Facilitation}. Based on research conducted prior to the ASIC study, \textit{REP} and \textit{Coaching} were designed for use in Stages 1 and 2 of implementation, whereas \textit{Facilitation} was designed only for use in Stage 2 of implementation. The \textit{REP} component includes a CBT uptake monitoring protocol that guides how implementation support professionals assess implementation progress. This monitoring protocol is used, for example, to determine whether a school is a \quotes{slower-responder} (also referred to as a \quotes{non-responder}) at the end of Stage 1, i.e., eligible for \textit{Facilitation} in Stage 2.

The ASIC study included 169 SPs across 94 Michigan high schools. All participating schools (\textit{i}) had not previously participated in any school-based CBT implementation initiatives, (\textit{ii}) were within a two-hour driving distance of a mental health professional trained to serve as a coach for the study, (\textit{iii}) had at least one eligible SP that agreed to participate in study assessments, and (\textit{iv}) had sufficient resources to allow for delivery of individual and/or group mental health support on school grounds \citep{Smith2022}.

\begin{figure}[H]
    \centering
    \includegraphics[width=12cm]{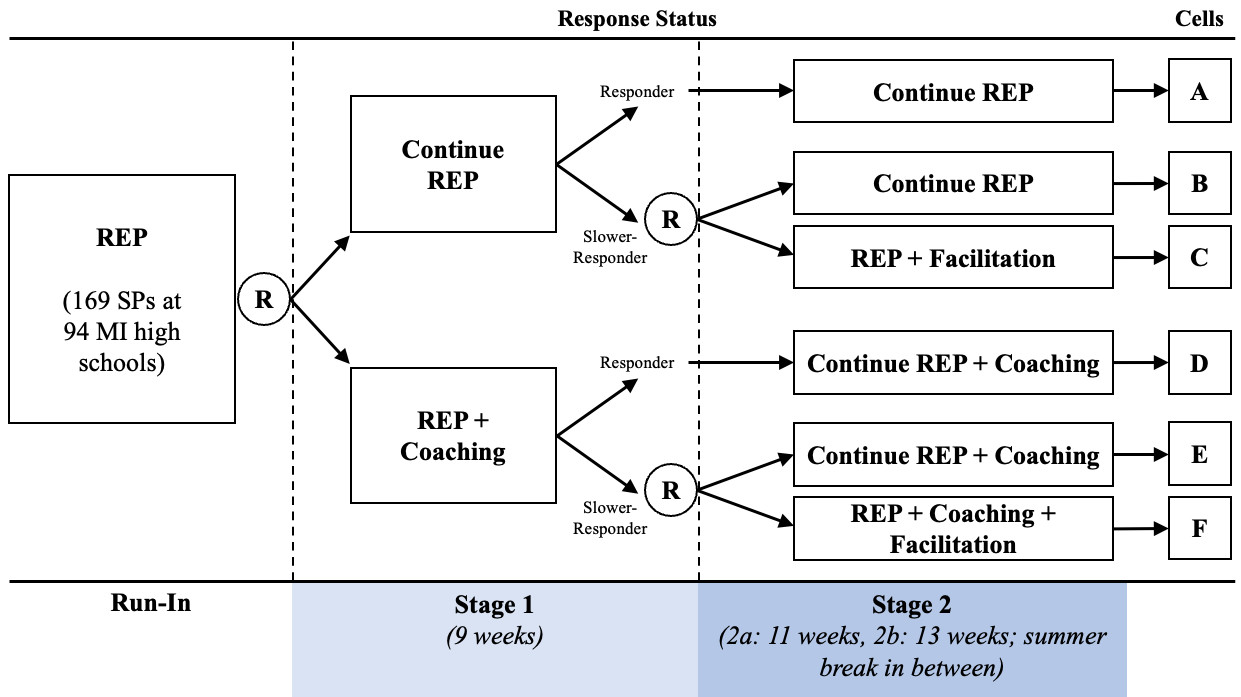}
    \caption{ASIC Structure}
    \label{c1:fig:ASIC_str}
\end{figure}
Figure \ref{c1:fig:ASIC_str} shows ASIC's randomization structure. During a three month prerandomization run-in stage, all 94 schools were offered \textit{REP}. After this phase, schools were randomized to either continue \textit{REP}, or to augment \textit{REP} with \textit{Coaching}. Nine weeks after this initial randomization, schools deemed \quotes{slower-responders} were re-randomized to either augment their first stage intervention with \textit{Facilitation} or continue with their initial treatment. Response status was a function of SP-reported barriers and frequency of CBT delivery, aggregated to the level of the school; see \cite{Smith2022} for a precise definition.


The primary research outcome is the (repeatedly measured) CBT delivery by each SP \citep{Kilbourne2018}. The nesting of repeated outcome measures within each SP and the nesting of multiple SPs within each school induces the three-level clustering structure that is central to the new method developed in this manuscript.

As discussed in greater detail in Section \ref{c1:sec:smart:embcai}, most clustered SMARTs contain a set of \textit{embedded} cAIs. By design, ASIC includes four such cAIs (see Table \ref{c1:tab:cAI_table_ASIC}); and ASIC's primary aim was to study the difference in the marginal expectation of total CBT delivery at the end of implementation under the most versus least intensive of the four cAIs \citep{Kilbourne2018}. As demonstrated in the simulation experiments of Section \ref{c1:sec:sim_study:power}, and illustrated in Section \ref{c1:sec:data_an}, the proposed repeated-measures (longitudinal) data analytic method enhances precision and reliability in addressing such primary aims. Additionally, in the illustrative data analyses, we show how the method can be used to answer new scientific questions using the repeated outcome measurements (e.g., to compare the four embedded cAIs in terms of changes in SP-level CBT delivery trajectories over time) otherwise masked in a static end-of-study analysis.

\section{Clustered SMARTs with Repeated Measures}

SMARTs are a class of multi-stage, factorial randomized trial designs, which leverage sequential randomization to inform the construction of optimal adaptive interventions \citep{SMART_IES}. SMART designs can vary widely; the characterizing feature of a SMART is that at least some units are randomized more than once \citep{Seewald2021}.

As discussed in Section \ref{c1:sec:intro}, clustered SMARTs ({cSMARTs}) can inform the optimal construction of clustered adaptive interventions \citep{NeCamp2017}. 
In a cSMART, clusters of individuals are sequentially randomized, with outcomes primarily measured with respect to individuals. For example, ASIC trial designers randomized entire schools with the intent to study outcomes at the SP level \citep{Kilbourne2018}. While cSMARTs typically randomize groups of humans, this need not be the case in general. \citet{Xu2019} provides an example of a cSMART studying dental procedures in which each human subject represents a cluster, with their teeth representing the individuals.

\subsection{SMART Randomization Structures}

Figure \ref{c1:fig:SMART_strs} displays four common SMART \quotes{design types.} In each of the presented design types, all clusters are randomized to one of two first stage intervention options; however, the design types differ in their re-randomization structure. SMART designs I, II, and III incorporate an embedded binary tailoring variable, \quotes{response.} SMART design IV (often called an \quotes{unrestricted $2\times 2$ SMART}) differs from the other three in this respect, as all clusters that received a given first-stage intervention are re-randomized to one of two second-stage interventions. This re-randomization is restricted to non-responders in SMART design II, with SMART design III further restricting re-randomization to non-responders of a single first-stage treatment arm. SMART design II is possibly the most common SMART design, and is often referred to as a \quotes{prototypical} SMART. 
\begin{figure}[H]
    \centering
    \includegraphics[width=12cm]{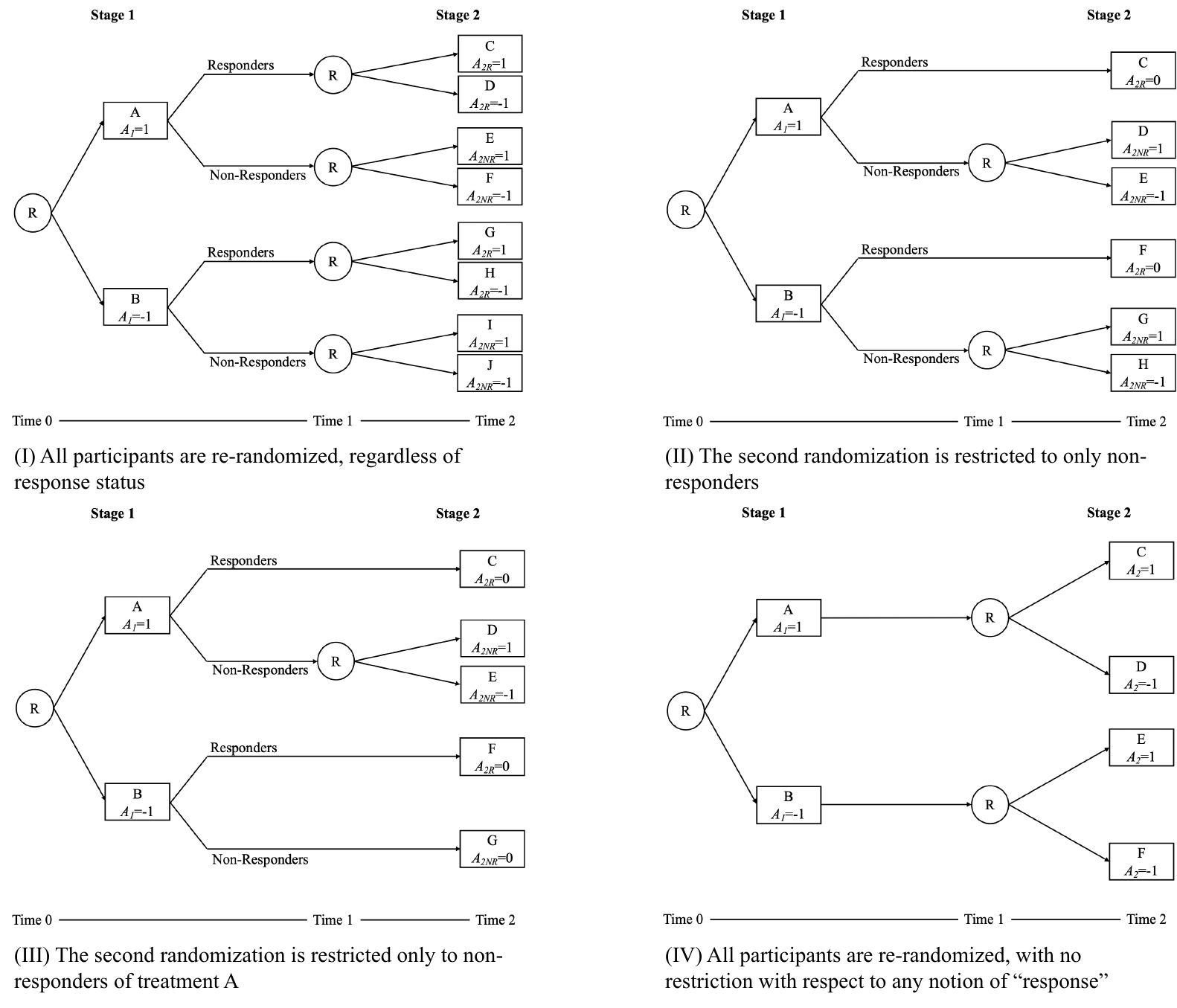}
    \caption{Illustrative SMART Structure Examples}
    \label{c1:fig:SMART_strs}
\end{figure}

SMART design I, while not yet employed in the clustered setting, is popular in individually randomized SMARTs (e.g., \citet{ExTENd}).
\citet{Kilbourne2014} details the first documented clustered SMART, employing SMART design III above. \citet{Quanbeck2020} employed SMART design IV to study implementation strategies to promote concordant opioid prescription. Furthermore, while all four design types above use two stages of binary randomization, this need not be the case for all SMARTs --- see \citet{Xu2019} for a clustered SMART with four-arm re-randomization for non-responding clusters.

As shown in Figure \ref{c1:fig:ASIC_str}, ASIC was a prototypical SMART. As such, we will use this design to illustrate and motivate notation and methods throughout the remainder of the main body of this manuscript. Appendix \ref{c1:sec:gen_smart} discusses generalizations to other SMART structures.

\subsection{Embedded Adaptive Interventions}\label{c1:sec:smart:embcai}

A clustered adaptive intervention is a sequence of decision rules guiding the provision of intervention based on the baseline and changing status of recipient clusters. cSMART designs contain a notion of \textit{embedded} cAIs; i.e., protocolized decision rules embedded in their design. 
Prototypical SMARTs contain four embedded adaptive interventions, indexed by choice of first-stage intervention, second-stage treatment for responders, and second-stage treatment for non-responders. I.e., we can identify a given embedded adaptive intervention for prototypical SMARTs as the triple $(a_1, a_{2R}, a_{2NR})$, where $a_1$ represents choice of first-stage intervention, and $a_{2R}/a_{2NR}$ represent choice of second-stage intervention for responders/non-responders. Given prototypical SMARTs do not re-randomize responders, the choice of first-stage treatment and second-stage treatment for non-responders induces the four embedded adaptive interventions. For brevity, embedded adaptive interventions in prototypical SMARTs are often denoted as a double $(a_1, a_{2NR})$. As such, we can denote the ASIC embedded cAIs as \textit{(i)} (\textit{REP+Coaching}, \textit{Facilitation}), \textit{(ii)} \textit{(REP+Coaching}, \textit{No Facilitation}), \textit{(iii)} (\textit{REP}, \textit{Facilitation}), \textit{(iv)} (\textit{REP}, \textit{No Facilitation}). ASIC's primary aim comparison involved the total CBT delivery under (\textit{REP+Coaching}, \textit{Facilitation}) and (\textit{REP}, \textit{No Facilitation}). Table \ref{c1:tab:cAI_table_ASIC} below illustrates these four embedded cAIs.

By convention, we let $\mathcal{D}$ denote the set of embedded cAIs in a clustered SMART \citep{Seewald2019}. Using standard $\pm1$ notation for treatment indicators, this corresponds to $\mathcal{D}:=\set{(1,1),(1,-1),(-1,1),(-1,-1)}$ for prototypical cSMARTs.

\begin{table}[H]
\centering
\begin{tabular}{c c c c c}
\toprule
cAI$(a_1,a_{2NR})$ & \makecell{First-Stage\\Intervention} & \makecell{School\\Response} &  \makecell{Second-Stage\\Intervention} & \makecell{Cells in \\ Fig. \ref{c1:fig:ASIC_str}} \\ 
\midrule
\multirow{2}{*}{cAI$\paren{1,1}$} & \multirow{2}{*}{\makecell{\textit{REP + Coaching}\\($A_1=1$)}} & \makecell{Responder\\($R=1$)} & 
\makecell{Continue \\ \textit{REP + Coaching}} & D \\ \cline{3-5}
                                      & 
                                      & \makecell{Non-responder\\($R=0$)} & \makecell{Add \textit{Facilitation}\\($A_{2NR}=1$)} & F \\
\midrule
\multirow{2}{*}{cAI$\paren{1,-1}$} & \multirow{2}{*}{\makecell{\textit{REP + Coaching}\\($A_1=1$)}} & \makecell{Responder\\($R=1$)} & 
\makecell{Continue \\ \textit{REP + Coaching}} & D \\ \cline{3-5}
                                     &                                           & \makecell{Non-responder\\($R=0$)} & \makecell{Continue \\ \textit{REP + Coaching}\\($A_{2NR}=-1$)} & E \\
\midrule
\multirow{2}{*}{cAI$\paren{-1,1}$} & \multirow{2}{*}{\makecell{\textit{REP} Only\\($A_1=-1$)}} & \makecell{Responder\\($R=1$)} & Continue \textit{REP} & A \\ \cline{3-5}
                                     &                                      & \makecell{Non-responder\\($R=0$)} & \makecell{Add \textit{Facilitation}\\($A_{2NR}=1$)} & C \\
\midrule
\multirow{2}{*}{cAI$\paren{-1,-1}$} & \multirow{2}{*}{\makecell{\textit{REP} Only\\($A_1=-1$)}} & \makecell{Responder\\($R=1$)} & Continue \textit{REP} & A \\ \cline{3-5}
                                    &                                     & \makecell{Non-responder\\($R=0$)} & \makecell{Continue \textit{REP}\\($A_{2NR}=-1$)} & B \\
\bottomrule
\end{tabular}
\caption{Embedded Clustered Adaptive Interventions in ASIC}
\label{c1:tab:cAI_table_ASIC}
\end{table}

\subsection{Notation}\label{c1:sec:not}

\subsubsection{Observed Data}
Consider data from a prototypical cSMART with $N$ participant clusters ($i=1,\ldots,N$), each with $n_i$ individuals, where data was collected at pre-determined times $t_0\ddd t_T$. Further consider a repeatedly-collected outcome construct, with $\vct{Y}_{i,j}:=\tttvechTinline{Y_{i,j,t_0},}{\dots,}{Y_{i,j,t_T}}$ denoting the vector of outcomes collected at times $t_0\ddd t_T$ for Individual $j$ in Cluster $i$. I.e., $Y_{i,j,t}$ represents the measured value of $Y$ for Individual $j$ in Cluster $i$ measured at time $t$. 

Use the stacked vector 
$$
\vct{\bY}_i:= 
\begin{bmatrix}
    \vct{Y}_{i,1}^{\top}, & \dots, & \vct{Y}^{\top}_{i,n_i}
\end{bmatrix}=\begin{bmatrix}
Y_{i,1,t_0}, & \dots, & Y_{i,1,t_T}, & Y_{i,2,t_0}, & \dots & \dots, & Y_{i,n_i,t_T}
\end{bmatrix}^{\top}
$$
to denote the full vector of observed outcomes for Cluster $i$. Furthermore, we use $\vct{Y}$ and $\vct{\bY}$ to denote the vector of repeated measures for a generic individual and for all individuals in a generic cluster, respectively.

For Cluster $i$, let $A_{1,i}\in\set{-1,1}$ be a random variable denoting first-stage treatment assignment, let $\bX_i$ denote collected baseline information, and use $R_i\in\set{0,1}$ to denote response status. Furthermore, let $t^*\in[t_s, t_{s+1})$ (for some $s=0\ddd {T-1}$) denote the second decision point, with $A_{2NR,i}\in\set{-1,1}$ the randomly assigned, second-stage intervention for Cluster $i$, and use the convention $A_{2NR,i}:= NA$ if Cluster $i$ was not re-randomized by design. E.g., in prototypical cSMARTs the second-stage intervention for responding clusters is pre-determined, and common convention sets $A_{2R,i}:=0$ for all $i$. Lastly, let $A_{2,i}$ denote the observed second-stage intervention for Cluster $i$; i.e., $A_{2,i}=\begin{cases} 
    A_{2R,i} & \text{ if $R_i=1$}\\ 
    A_{2NR,i} & \text{ if $R_i=0$}\end{cases}$.

Given these conventions, we consider observed data collected over the course of the study on Cluster $i$ to have the form
$$O_i = \set{\bX_i, A_{1,i}, \vct{\bY}_{i,t_0},  \vct{\bY}_{i,t_1}\ddd \vct{\bY}_{i,t_{s}}, R_i, A_{2,i}, \vct{\bY}_{i,t_{s+1}}\ddd \vct{\bY}_{i,t_T}}.$$

\subsubsection{Potential Outcomes}\label{c1:sec:not:po}
We employ potential outcomes notation to define our primary aims and discuss causal effects \citep{Rubin2005, Robins2000}.
Let $R_i^{(a_1)}$ denote response status for Cluster $i$ had the cluster been assigned first-stage treatment $a_1$. More generally, for a given embedded cAI $d$ and outcome construct $Y$, let $Y^{(a_1, a_{2NR})}_{i,j,k}$ denote the outcome at time $t_k$ for Individual $j$ in Cluster $i$ had the cluster been assigned cAI $d=(a_1, a_{2NR})$. Appendix \ref{c1:sec:id_ass} discusses several canonical assumptions we place on the potential outcomes in order to identify causal effects.

\section{Primary Aims in a Clustered SMART with Repeated Measures}\label{c1:sec:ce}\label{c1:sec:ce:ce}
For a given outcome construct, $Y$, we consider $\E{\vct{Y}\d}$, 
the marginal mean of $\vct{Y}$ had the entire population of clusters been assigned embedded cAI $d$. As discussed previously, a common primary aim in SMART analyses involves the comparison of functions of marginal means under different embedded AIs \citep{Oetting2010, NahumShani2012}. We present three common such comparisons below.

\subsection{Comparison of Second-Stage Slope}
Researchers interested in the trajectory of outcomes may turn to the average slope of the outcome from the second decision point to end-of-study. Doing so focuses the analysis on the mean trajectory of the outcome during the second-stage of the adaptive intervention \citep{NahumShani2020}. Such an aim can be expressed as:
\begin{equation}\label{c1:eq:ce_s2}
    \Delta^{S2}_{\paren{d,d'}} = \frac{1}{t_T-t^*}\paren{\E{Y_{t_T}\d} - \E{Y_{t^*}\d}} - \frac{1}{t_T-t^*}\paren{\E{Y_{t_T}\dpr} - \E{Y_{t^*}\dpr}}.
\end{equation}

\subsection{Comparison of Average Area Under the Curve}\label{c1:sec:ce:cs}
Area under the curve (AUC) serves as a robust summary measure in longitudinal studies, capturing the evolution of an outcome over time. In SMARTs where the temporal profile of the outcome is of particular import, researchers may turn to the average AUC as an estimand of interest. This can be expressed as \citep{Sun2003}:
\begin{equation}\label{c1:eq:ce_auc}
    \Delta^{AUC}_{\paren{d,d'}} = \frac{1}{t_T-t_0}\int_{t_0}^{t_T}\E{Y_{t}\d}dt - \frac{1}{t_T-t_0}\int_{t_0}^{t_T}\E{Y_{t}\dpr}dt .
\end{equation}

\subsection{Comparison of End-of-Study Outcome}
A classic primary aim in cSMART analyses is the comparison of embedded cAIs with respect to the marginal expectation of an end-of-study outcome. While analyzing this aim does not require collecting repeated measurements, researchers often focus on this aim even when repeated outcome measurements are available, as was the case in ASIC \citep{Kilbourne2018}. As discussed previously, the primary aim for ASIC was to determine whether embedded cAIs (\textit{REP}, \textit{No Facilitation}) and (\textit{REP+Coaching}, \textit{Facilitation}) saw a difference in expected total CBT delivery, measured at end-of-study. This aim induces the estimand below.
\begin{equation}\label{c1:eq:ce_eos}
    \Delta^{ES}_{\paren{d,d'}} = \E{Y_{t_T}\d} - \E{Y_{t_T}\dpr}.
\end{equation}
As discussed above, ASIC's primary aim estimand was
$\E{Y_{t_T}^{\paren{REP+Coaching,\s Facilitation}}}-\E{Y_{t_T}^{\paren{REP,\s No\s Facilitation}}}$, where $Y_t$ denotes the total number of CBT sessions delivered by an SP up to time $t$.

Of course, an analyst may wish to analyze aims other than the three listed above. Broadly, the methods in this manuscript consider estimands of the form $\Delta_{\paren{d,d'}} = f\paren{\E{\vct{Y}\d}} - f\paren{\E{\vct{Y}\dpr}}$, with choice of $f:\mathbb{R}^{T+1}\to \mathbb{R}$ corresponding to choice of estimand. Lastly, we note the use of \quotes{marginal} to signify \quotes{marginal over response status} as well as to highlight the expectation being taken over the joint distribution of $Y_{t_0}\ddd Y_{t_T}$. An analyst may wish to condition on cluster- or individual-level baseline covariates, $X$, to control for finite-sample imbalances resulting from the randomization. In this case, we consider $\Delta_{\paren{d,d'}} := f\paren{\E{\vct{Y}\d\mid X}} - f\paren{\E{\vct{Y}\dpr\mid X}}$.

\section{Modeling}\label{c1:sec:mm}

This section describes marginal mean models 
for $\E{Y_t\d\mid X}$, denoted by $\mu_t\paren{d,X;\theta}$, where $\theta$ is a finite set of unknown parameters to be estimated using the SMART data. We show how the marginal mean models connect with the causal estimands listed above. Note that we focus solely on models that are linear in $\theta$.

As noted previously \citep{Lu2015, NeCamp2017}, the randomization structure of the SMART may induce natural constraints on the form of $\mu_t$. In the following, we provide an example marginal mean model for a prototypical SMART, such as ASIC:
\begin{equation} \label{c1:eq:mmm_ex}
\begin{split}
    \mu_t(a_1, a_{2NR}, X; \theta) & = 
     \gamma_0 + \eta X + \Ind{t\leq t^*}\paren{\gamma_1t + \gamma_2a_1t} \\
     & \hspace{.5cm} + \Ind{t > t^*}\big(\gamma_1t^* + \gamma_2a_1t^* + \gamma_3\paren{t-t^*} + \gamma_4\paren{t-t^*}a_1 \\
     & \hspace{.5cm} + \gamma_5\paren{t-t^*}a_{2NR} + \gamma_6\paren{t-t^*}a_1a_{2NR}\big).
\end{split}
\end{equation}

This example marginal mean model is piecewise linear in time with a knot at the second decision point (i.e., at time $t^*$). All parameters in Equation~\eqref{c1:eq:mmm_ex} have scientific interpretation. $\gamma_0$ and $\eta$ represent baseline mean and effect of baseline covariates, respectively. $\gamma_1$ and $\gamma_2$ encode first-stage treatment slopes for the two first-stage intervention options $a_1=\pm 1$. Lastly, $\gamma_3, \gamma_4, \gamma_5, \gamma_6$ induce the second-stage treatment slopes for the four embedded cAIs in $\mathcal{D}$.

As noted in previous works, the sequential nature of treatment delivery in SMART settings may induce natural constraints on the form of $\vct{\mu}$. E.g., Equation~\eqref{c1:eq:mmm_ex} ensures $\mu_t(a_1, 1, X;\theta)=\mu_t(a_1, -1, X;\theta)$ for any $t\leq t^*$, reflecting the assumption that future treatments should not impact past potential outcomes (discussed further in Appendix \ref{c1:sec:id_ass}). In general, such constraints vary depending on the exact randomization structure of the SMART at hand \citep{Lu2015, NeCamp2017}, and we discuss modeling concerns for alternate SMART design types in Appendix \ref{c1:sec:gen_smart}.

An analyst could employ a different marginal mean model. E.g., incorporating treatment covariate interactions, non-linear temporal trends, and additional knots at times other than at $t^*$ are all potentially prudent modeling decisions to be informed by the subject matter at hand \citep{Li2016}.

As with previous notation, we let $\vbmu(d,\bX;\theta)$ denote a stacked vector of marginal means for a generic cluster.

\subsection{Connection with Target Estimands}\label{c1:sec:mm:lmm:parms}

Section \ref{c1:sec:ce:ce} discussed a number of candidate estimands for analysis. The marginal mean model provides a parametric
representation of the conditional counterfactual mean,
$\E{Y_{t}\d\mid X}$. Estimation of $\theta$ enables estimation of causal contrasts by evaluating the fitted mean model under the corresponding embedded cAIs. For example, to study the causal difference in embedded cAIs $d=\paren{a_1, a_{2NR}}$ and $d'=\paren{a_1', a_{2NR}'}$ with respect to end-of-study outcome, the target estimand is $\E{Y_{t_T}\d\mid X}-\E{Y_{t_T}\dpr \mid X}$. The marginal mean model represents this contrast as $\mu_{t_T}\paren{d,X;\theta}-\mu_{t_T}\paren{d',X;\theta}$ and, using the model in Equation~\eqref{c1:eq:mmm_ex} as an illustrative example, we can write
\begin{align*}
     \mu_{t_T}\paren{d,X;\theta}-\mu_{t_T}\paren{d',X;\theta} 
    & = \paren{a_1-a_1'}t^*\gamma_2 + 
    \paren{a_1-a_1'}\paren{t_T-t^*}\gamma_4 \\
    & \hspace{2cm} +
    \paren{a_{2NR}-a'_{2NR}}\paren{t_T-t^*}\gamma_5  \\
    & \hspace{2cm} +
    \paren{a_1a_{2NR}-a_1'a'_{2NR}}\paren{t_T-t^*}\gamma_6.
\end{align*}

Similarly, to compare the average area under the curve between two embedded AIs, we would calculate
$$\frac{1}{t_T-t_0}\int_{t_0}^{t_T} \mu_t\paren{d,X;\theta} dt - \frac{1}{t_T-t_0}\int_{t_0}^{t_T}\mu_t\paren{d',X;\theta}dt.$$
Using Equation~\eqref{c1:eq:mmm_ex} for a marginal mean model and $d=(1,1), d'=(-1,-1)$ as an example, this corresponds to testing the null 
$$H_0: \frac{1}{t_T-t_0}\paren{\paren{{t^*}^2-t_0^2 + 2t^*\paren{t_T-t^*}}\gamma_2 + \paren{{t_T}^2-{t^*}^2 - 2t^*\paren{t_T-t^*}}\paren{\gamma_4 + \gamma_5}}= 0.$$ 
We discuss parameter estimation and hypothesis testing in Section \ref{c1:sec:est} below.

\subsection{Working Variance Modeling}\label{c1:sec:mm_work_var}
A primary contribution of this paper is the development
of methods that account jointly for within-person and between-person variance structures. This section describes such modeling considerations. In standard randomized trials with repeated measures \citep{Wang2014}, standard cluster-randomized trials \citep{Offorha2023}, and standard three-level randomized trials \citep{Teerenstra2010} it is common for trialists to pose working models for the variance of the residual between the outcome and the marginal mean model. Sections \ref{c1:sec:data_an} and \ref{c1:sec:sim_study:wvm} illustrate two benefits of such an approach: $(i)$ from a scientific perspective, there is often tertiary or exploratory interest in understanding the structure of $\V{\bY\d\mid \bX}$; $(ii)$ from a statistical perspective, employing proper working models of $\V{\bY\d\mid \bX}$ can enhance precision of estimators of the causal mean parameters of interest.

In Section \ref{c1:sec:est} below, we describe an estimator for $\theta$ that allows analysts to pose such working models. As discussed in Section \ref{c1:sec:est:asy_dist}, this estimator is consistent regardless of choice of working variance model. Additionally, Section \ref{c1:sec:sim_study:wvm} empirically shows that the estimator has negligible bias in moderate and large samples, regardless of working variance model. 

Let $V^d\paren{\bX_i;\alpha}$ denote a working model for the variance-covariance matrix 
\\$\V{\vbY_i\d - \vbmu(d,\bX;\theta)\mid \bX=\bX_i}$,
which is indexed by the unknown parameters $\alpha=\begin{bmatrix}
    \sigma\\ \rho
\end{bmatrix}$. In general, $V^d\paren{\bX_i;\alpha}$ takes the form
$$V^d\paren{\bX_i;\alpha} = \bS_i^d(\sigma)\poh \,\,  \boldsymbol\Rho_i^d(\rho) \,\, \bS_i^d(\sigma)\poh.$$

Here, the $n_i(T+1)\times n_i(T+1)$ matrices $\bS_i^d(\sigma)$ and $\boldsymbol\Rho_i^d(\rho)$ correspond to variance and correlation matrices (respectively). Choice of working variance model corresponds to choice of structure for $\bS_i^d(\sigma)$ and $\boldsymbol\Rho_i^d(\rho)$. As a simple example, choosing a homoscedastic-independent working variance model would correspond to setting $\bS_i^d(\sigma) = \sigma^2 I_{n_i(T+1)\times n_i(T+1)}$ and $\boldsymbol\Rho_i^d(\rho) = I_{n_i(T+1)\times n_i(T+1)}$.

Such a homoscedastic-independent working model forgoes modeling correlation between repeated observations in an individual and between individuals in a cluster. Furthermore, adopting such a model involves pooling variance estimates across time and embedded cAI. For an illustrative example better capturing the types of working variance decisions to be made, consider a working correlation model that is exchangeable between-person and within-person and heterogeneous across embedded cAI. Moreover, assume heterogeneous variance across time and embedded cAI. I.e., we model $\V{Y_{i,j,t}\d \mid \bX_i} = \paren{\sigma^d_t}^2$ for any $i,j$; $corr\paren{Y_{i,j,t}\d, Y_{i,j,t'}\d\mid \bX_i}=\rho_w^d$ for any $i,j$ and $t\neq t'$; and $corr\paren{Y_{i,j,t}\d, Y_{i,j',t'}\d\mid \bX_i}=\rho_b^d$ for any $i,j$, $j'\neq j$, $t$, and $t'$. Appendix \ref{c1:sec:var_dis} contains a further discussion on variance modeling and $\bS^d$, $\boldsymbol\Rho^d$ structure.

As with the mean, modeling the marginal variance should reflect the subject matter at hand. Table \ref{c1:tab:wv:mar_var} discusses several models for the marginal variance of a single outcome. The results in Section \ref{c1:sec:sim_study:wvm} show this choice can materially affect estimator performance. As discussed in Section \ref{c1:sec:sim_study:wvm}, we propose choosing more flexible variance models, particularly with respect to time.

\begin{table}[H]
    \centering
    \renewcommand{\arraystretch}{1.5}
    \begin{tabular}{ccc}
    \hline
    \multicolumn{2}{c}{Marginal Variance Structure} & \multirow{2}{*}[-0.5ex]{\centering$\V{Y_{i,j,t}\d\mid \bX_i}$} \\
    \cline{1-2}
    Time & \makecell{Embedded cAI}&  \\
    \hline
    Heteroscedastic  & Heterogeneous &
    $(\sigma_t^{d})^2$\\
    Heteroscedastic & Homogeneous &
    $\sigma_t^2$ \\
    Homoscedastic  &  Heterogeneous &
    $(\sigma^{d})^2$\\
    Homoscedastic & Homogeneous &
    $\sigma^2$ \\
    \hline
    \end{tabular}
    \caption{Marginal Variance Models}
    \label{c1:tab:wv:mar_var}
\end{table}

Additionally, we present examples for working correlation models in Table \ref{c1:tab:wv:corr}. The models presented are heterogeneous with respect to embedded cAI; however, one could choose cAI-homogeneous models for the correlation (analogous to the cAI-homogeneous variance models in Table \ref{c1:tab:wv:mar_var}). 
The separation of within- and between-person correlation models enables us to consider new working variance structures beyond those presented in previous work on multilevel SMART analyses, such as those in \citet{NeCamp2017, Seewald2019}. Appendix \ref{c1:sec:var_est} contains details on variance estimation techniques for the variance/correlation structures introduced in this manuscript.

\begin{table}[H]
    \centering

    \begin{subtable}{0.75\linewidth}
        \centering
        \begin{tabular*}{\linewidth}{@{\extracolsep{\fill}}cc}
            \hline
            \makecell{Within-Person \\ Correlation Structure}
            &
            \makecell{
                $corr\paren{Y_{i,j,t_l}\d,Y_{i,j,t_m}\d\mid\bX_i}$
                \\ $(t_l\neq t_m)$
            } \\
            \hline
            AR(1)        & $(\rho_w^d)^{|l-m|}$ \\
            Exchangeable & $\rho_w^d$ \\
            Unstructured & $\rho^d_{w,t_l,t_m}$ \\
            Independent  & $0$ \\
            \hline
        \end{tabular*}
        \caption{Within-Person Correlation Models}
        \label{c1:tab:wv:wp_corr}

        {\raggedright\scriptsize\textit{
            Within-person correlation is 1 if $t_l=t_m$.
        }\par}
    \end{subtable}

    \par\vspace{0.5cm}

    \begin{subtable}{0.75\linewidth}
        \centering
        \begin{tabular*}{\linewidth}{@{\extracolsep{\fill}}cc}
            \hline
            \makecell{Between-Person \\ Correlation Structure}
            &
            \makecell{
                $corr\paren{Y_{i,j,t_l}\d,Y_{i,k,t_m}\d\mid\bX_i}$
                \\ $(j\neq k)$
            } \\
            \hline
            Exchangeable & $\rho_b^d$ \\
            Unstructured & $\rho^d_{b,t_l,t_m}$ \\
            Independent  & $0$ \\
            \hline
        \end{tabular*}
        \caption{Between-Person Correlation Models}
        \label{c1:tab:wv:bp_corr}
    \end{subtable}

    \caption{Marginal Correlation Models}
    \label{c1:tab:wv:corr}
\end{table}

\section{Estimation}\label{c1:sec:est}

Given models for the marginal mean and variance of $\vct{Y}\d$, we seek to obtain parameter estimates for inference. Similar to \citet{Lu2015, NeCamp2017, Seewald2019}, we employ the following weighted estimating equation:    
\begin{align}
    0  & = \SumiN\Sumd I_i(d) W_i D(d,\bX_i;\theta)^\top V^d(\bX_i;\alpha)\inv(\vbY_i - \vbmu(d,\bX_i;\theta))\label{c1:eq:est_eq}\\
    & =: \SumiN U\paren{A_{1,i}, R_i, A_{2,i}, \vbY_i, \bX_i; \alpha, \theta},\notag
\end{align}
where $I_i(d) = I(d, A_{1,i}, R_i, A_{2,i})$ is an indicator function denoting whether Cluster $i$'s treatment/response history is consistent with cAI $d$ and $D(d,\bX_i;\theta)=\partialD{\mu}{\theta}(d,\bX_i;\theta)$. $D(d,\bX_i;\theta)$ is similar to the \quotes{design matrix} in standard regression. As before, we use $V^d(\bX_i;\alpha)$ to denote a working covariance matrix for the residuals $\paren{\vbY_i - \vbmu(d,\bX_i;\theta)}$.

Lastly, the weights are defined as:
$$W_i = W(a_{1,i}, r_i, a_{2,i})
=\dfrac{1}{\Prb{A_{1}=a_{1,i}}\Prb{A_{2}=a_{2,i} \mid A_{1,i}=a_{1,i},R_i=r_i}}.$$
The summation over $d\in\mathcal{D}$ and the $I_i(d)W_i$ terms represent departures from standard GEE analyses \citep{LIANG1986} that are unique to the analysis of SMART designs. Summing across $d\in\mathcal{D}$ is required because, in a cSMART, any individual cluster may be consistent with multiple embedded cAIs. For example, clusters assigned $A_1=1$ that had positive response to first stage treatment are consistent with both $d_1=(1,1)$ and $d_2=(1,-1)$; i.e., their treatment/response history could have arisen from $d_1$ or $d_2$. The data for such clusters is used to estimate the measured mean for multiple cAIs.

The weights, which are known by design, are constructed to address the imbalance between responders and nonresponders. For example, in prototypical cSMARTs with balanced .5/.5 randomization probabilities (like ASIC), responding clusters have weight 2 and non-responding clusters have weight 4.

On the other hand, the analyst may wish to pre-specify the estimation of these probabilities because doing so can improve statistical efficiency \citep{Robins2000}. See Appendix \ref{c1:sec:wgt_est} for details.

We call the solution to Equation~\eqref{c1:eq:est_eq} $\hat{\theta}$. As discussed above, inference regarding functions of $\hat{\theta}$ can correspond to inference for comparisons between embedded adaptive interventions with respect to a wide variety of outcomes.

\subsection{Estimation Algorithm}


Similar to \citet{NeCamp2017, Seewald2019}, and inspired by classic approaches in fitting related generalized estimating equations (GEEs) \citep{Huang2021}, we employ an iterative estimation procedure that alternates between $\hat\theta$ and $\hat\alpha$, as shown in Algorithm \ref{c1:alg:est} below. By separating the two estimation mechanisms, we can obtain $\theta$ estimates by solving linear equations and use the induced residuals for $\alpha$ estimation.

\allowdisplaybreaks
\begin{breakablealgorithm}
\caption{\texttt{Root Estimation Algorithm}}\label{c1:alg:est}
\begin{algorithmic}[1]
\State \textbf{Initialize} $\hat\theta_0$ by solving Equation~\eqref{c1:eq:est_eq} with $V^d(\bX_i;\alpha) = I_{n_i(T+1)}$ for all $i=1\ddd N$ and $d\in\mathcal{D}$.
\State\label{c1:algstep:start}\textbf{For} each cAI $d\in\mathcal{D}$ and $i=1\ddd N$, obtain the residuals $\vct{\hat\varepsilon}^{d,0}_i := \vbY_i - \vbmu(d,\bX_i;\hat\theta_0)$.
\State\label{c1:algstep:varest}\textbf{Estimate} $\alpha$ using the residuals obtained above via the formulae presented in Appendix \ref{c1:sec:var_est}. Call the resulting estimate $\hat\alpha_0$.
\State\label{c1:algstep:end}\textbf{Solve} Equation~\eqref{c1:eq:est_eq} for $\theta$ using $V^d(\bX_i;\alpha)=V^d(\bX_i;\hat\alpha_0)$ for all $i=1\ddd N$ and $d\in\mathcal{D}$. Call the resulting estimate $\hat\theta_1$.
\State\textbf{Repeat} Steps \ref{c1:algstep:start}-\ref{c1:algstep:end}, using residuals $\vct{\hat\varepsilon}^{d,k}_i:=\vbY_i - \vbmu(d,\bX_i; \hat\theta_{k-1})$ to obtain $\hat\alpha_k$ estimates which can then be used to obtain $\hat\theta_k$ estimates. Continue until a desired convergence criterion is reached (e.g., $\lnorm{\hat\theta_{k}-\hat\theta_{k-1}}{\infty}<\epsilon).$ Denote the resulting estimates $\alphahat$ and $\thetahat$.
\end{algorithmic}
\end{breakablealgorithm}

We note that the exact form of Step \ref{c1:algstep:varest} above depends on the working variance structure chosen. Appendix \ref{c1:sec:var_est} discusses covariance component estimation for a broad class of covariance structures.

For the example working variance structure discussed in Section \ref{c1:sec:mm_work_var} (i.e., heterogeneous with respect to time and cAI and exchangeable both within-person and between-person), we present the following estimators:
$$
    \paren{\sigmahat^{d}_t}^2  = 
    \ddfrac{\Sum{i}{N}W_iI_i(d)\Sum{j}{n_i}\paren{\varepsilon^{d,s}_{i,j,t}}^2}{\Sum{i}{N}W_iI_i(d)n_i},\s
    \rhohat_{w}^d = \ddfrac{\Sum{i}{N}W_iI_i(d)
\displaystyle\sum_{j=1}^{n_i}\displaystyle\sum_{k=0}^{T}\displaystyle\sum_{\substack{k'=0\\k'\neq k}}^T
\dfrac{\varepsilon^{d,s}_{i,j,t_k}\varepsilon^{d,s}_{i,j,t_{k'}}}{\sigmahat^d_{t_k}\sigmahat^{d}_{t_{k'}}}}{\Sum{i}{N}W_iI_i(d)n_i(T+1)T},$$
$$
\rhohat_{b}^d  = \ddfrac{\Sum{i}{N}W_iI_i(d)
\Sum{j}{n_i}\displaystyle\sum_{\substack{j'=1\\j'\neq j}}^{n_i}
\displaystyle\sum_{k=0}^{T}\displaystyle\sum_{k'=0}^T
\dfrac{\varepsilon^{d,s}_{i,j,t_k}\varepsilon^{d,s}_{i,j',t_{k'}}}{\sigmahat^d_{t_k}\sigmahat^{d}_{t_{k'}}}}{\Sum{i}{N}W_iI_i(d)n_i(n_i-1)(T+1)^2}.
$$

\subsection{Asymptotic Distribution}\label{c1:sec:est:asy_dist}

Discussed in more detail in Appendix \ref{c1:sec:asy_thry}, we show that, under certain regularity conditions, 
$$\sqrt{N}\paren{\thetahat-\theta_0}\xrightarrow[N\to\infty]{\mathcal{L}}\N{0}{J\inv Q J\inv},$$
where 
\begin{align*}
    J & = \E{\Sumd I(d) WD(d,X;\theta)^\top V^d(\bX_i; \alpha_+)\inv D(d,X;\theta)},\\
    Q & = \E{UU^\top}.
\end{align*}
We note that, as discussed in the aforementioned appendix, this convergence does not require proper specification of the working variance structure provided the marginal mean model is correctly specified, resembling results from classic GEE theory \citep{LIANG1986}. Appendix \ref{c1:sec:wgt_est} contains the corresponding result for the asymptotic distribution of $\thetahat$ using estimated weights.

As in \citet{Lu2015, NeCamp2017, Seewald2019}, we use plug-in estimators for $J$ and $Q$ to obtain asymptotic variance estimates of $\thetahat$. I.e., we take $\hat{\Sigma}_{\thetahat} = \hat{J}\inv \hat{Q} \hat{J}\inv$ where:
\begin{align*}
    \hat{J} & = \frac{1}{N}\SumiN\Sumd I_i(d) W_i D(d,\bX_i;\thetahat)^\top V^d(\bX_i; \alphahat)\inv D(d,\bX_i;\thetahat),\\
    \hat{Q} & = \frac{1}{N}\SumiN U_i U_i^\top.
\end{align*}

\subsection{Hypothesis Testing for cAI Comparisons}

We can conduct inference on any linear combination of $\theta$ by using the univariate Wald statistic $Z=\sqrt{N}\dfrac{c^\top\thetahat}{\sqrt{c^\top\hat{\Sigma}_{\thetahat}c}}$ to test the null hypothesis $H_0: c^\top\theta=0$. We note that proper choice of $c$ can correspond to inference on a wide variety of primary aim comparisons, as discussed in Section \ref{c1:sec:mm:lmm:parms}. Lastly, we employed the finite-sample adjustments recommended by \citet{Pan_U} for clustered SMART analyses. The simulation code discussed in Appendix~\ref{c1:sec:sim_des} contains details on the exact implementation.

\section{Simulation Study}\label{c1:sec:sim_study}
This section presents three simulation studies to better understand the operating characteristics of the proposed estimator across a variety of settings. Data-generative models were chosen to mimic ASIC, a prototypical clustered SMART, but with three time points: baseline ($t_0=0$), end of first stage of intervention ($t^*=1$), and end of second stage of intervention ($t_T=2$). We explore simulations with more than three time points in Appendix~\protect\ref{c1:sec:add_sim:multi_t}. Without loss of generality, outputs from all simulation studies report metrics for the comparison of embedded cAIs $(1,1)$ and $(-1,-1)$ with respect to end-of-study marginal mean estimates.

The data-generative models were designed to mimic clustered SMARTs in practice, while allowing us to manipulate the $(i)$ true effect size for this comparison (ranging from small to large \citep{Cohen1988}), $(ii)$ sample size (number of clusters ranging from 20 to 700, with cluster sizes 2 and 3), $(iii)$ between-person correlation structure (independent, exchangeable, unstructured), and $(iv)$ within-person correlation structure (independent, exchangeable, unstructured, AR(1)). Unless otherwise specified, all analyses in the simulation experiments were conducted using $20,000$ Monte Carlo iterations and employ the finite-sample adjustments discussed in \cite{Pan_U}. Because a large number of Monte Carlo iterations were run, the resulting confidence intervals are uniformly narrow. To improve visual clarity, we omit them from the main text; full results, including these intervals, are available on \href{https://github.com/GabrielDurham/three_lvl_cSMART_sim_study}{GitHub}. 
Furthermore, all details regarding simulation design and \texttt{R} code to replicate the simulations are provided in Appendix \ref{c1:sec:sim_des} and \href{https://github.com/GabrielDurham/three_lvl_cSMART_sim_study}{GitHub}, respectively. In particular, Appendix \ref{c1:sec:sim_des:parm_sel} outlines the data-generative parameter selection procedure.

\subsection{Estimator Validity}\label{c1:sec:sim_study:validity}

The purpose of the first simulation experiment was to verify the consistency of the estimator and examine empirical performance in small samples. We hypothesized, as supported by the theoretical results stated in Section \ref{c1:sec:est:asy_dist}, that in large samples, the estimator would coalesce around the true value. Furthermore, we hypothesized that the estimator would be more volatile in small samples, but generally centered around the true mean.

Table~\ref{c1:tab:base_sim} below displays the performance of a correctly-specified estimator on simulated data inspired by CBT delivery trends in the ASIC data. As shown below, relative bias ($\frac{\paren{\muhat_2((1,1))-\muhat_2((-1,-1))}-\paren{\mu_2((1,1))-\mu_2((-1,-1))}}{\mu_2((-1,-1))}$) is negligible across sample sizes and the estimator tightens around the true difference as $N$ grows. Furthermore, we achieve near-nominal coverage when applying the finite-sample adjustments discussed above. 

\begin{table}[H]
\centering
\begin{tabular}{ccccc}
  \hline
Number of Clusters & Relative Bias & RMSE & Coverage \\ 
  \hline
N=20 & -0.04 & 2.14 & 0.90 \\ 
  N=30 & -0.03 & 1.77 & 0.91 \\ 
  N=40 & -0.01 & 1.53 & 0.92 \\ 
  N=50 & -0.02 & 1.37 & 0.93 \\ 
  N=75 & -0.01 & 1.11 & 0.94 \\ 
  N=100 &  0.00 & 0.97 & 0.94 \\ 
  N=500 &  0.00 & 0.43 & 0.95 \\ 
   \hline
\end{tabular}
\caption{Estimator Performance by Sample Size}\label{c1:tab:base_sim}
\end{table}

\subsection{Working Variance Modeling}\label{c1:sec:sim_study:wvm}

The purpose of the second simulation experiment was primarily to examine the impacts of working variance modeling choices on estimator performance. We hypothesized that correctly specifying the marginal variance structure would show material efficiency gains. This simulation is also used to examine the asymptotic result that confidence intervals will achieve nominal coverage across different working variance models.

Figure \ref{c1:fig:wvm_miss} shows the comparative performance of estimators with various working variance models. As discussed in Appendix \ref{c1:sec:sim_des:parm_sel}, this comparison is based on a data-generative model with complex correlation (AR(1) within-person, unstructured between-person) and heteroscedastic variance, both homogeneous over embedded cAI.
For each working variance model ($wv$), the plot displays the empirical root relative efficiency of $wv$ compared to a correctly specified working variance model, calculated as $\widehat{\mathrm{SE}}_{wv} / \widehat{\mathrm{SE}}_{spec}$. That is, we report the average ratio of standard errors obtained under $wv$ to those obtained under the correct specification, across a range of sample sizes.

As $N$ grows, the correctly specified working variance model, represented by the dashed line in panel (d) of Figure~\ref{c1:fig:wvm_miss}, outperforms its counterparts. For large $N$, it achieves a $\approx 10\%$ efficiency gain over the homoscedastic-independent approach. Additionally, heteroscedastic working variance models tend to outperform those with temporally homogeneous variance.

While Figure~\ref{c1:fig:wvm_miss} shows that more faithfully modeling the underlying correlation structure can improve efficiency, the results in Section~\ref{c1:sec:est:asy_dist} suggest that choice of working variance estimate should not affect the \textit{validity} of parameter estimates/confidence intervals. To test this intuition, we examined the coverage rates induced by the various choices of working variance. Figure~\ref{c1:fig:wvm_miss_app:cov} in Appendix~\ref{c1:sec:add_figs} shows these results, and suggests that choice of working variance structure does not heavily impact coverage rate.

\begin{figure}[ht]
\centering
\setlength{\unitlength}{1in}
\begin{picture}(5,3)
  \put(.5,0){\includegraphics[width=4in]{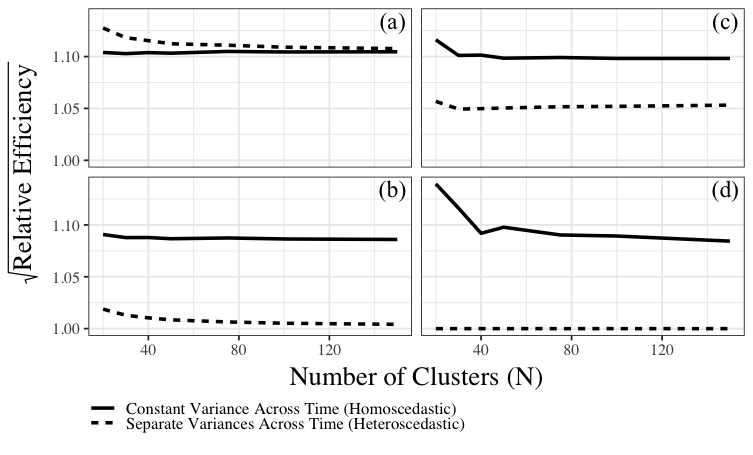}}
    \put(2.65,2.78){\makebox(0,0){\textbf{Working Model: Between-Person Correlation}}}
    \put(1.75,2.52){\makebox(0,0){\small Independent}}
    \put(3.55,2.52){\makebox(0,0){\small Unstructured}}

    \put(0.18,1.55){\makebox(0,0){\rotatebox{90}{\textbf{Working Model: Within-Person Correlation}}}}
    \put(0.40,1.96){\makebox(0,0){\rotatebox{90}{\small Independent}}}
    \put(0.40,1.10){\makebox(0,0){\rotatebox{90}{\small AR(1)}}}
\end{picture}
\captionsetup{belowskip=-7pt}
\caption{Relative efficiency under different working correlation models.}
\caption*{\footnotesize \textit{Note: Data generated under a time-heteroscedastic variance structure that was homogeneous across embedded cAIs, with AR(1) within-person correlation and unstructured between-person correlation. Relative efficiency calculated with respect to a correctly specified working variance model.}}
\label{c1:fig:wvm_miss}
\end{figure}

As a follow-up to this analysis, we investigated the efficiency trade-offs associated with employing complex working variance models when the true data-generating variance structure is simple (i.e., homoscedastic and independent). Figure \ref{c1:fig:wvm_miss_iid} in Appendix \ref{c1:sec:add_figs} presents the corresponding results in this scenario. The figure demonstrates that, in such settings, the choice of working variance model has minimal impact on efficiency. 
These findings suggest that under-specifying the working variance structure can be considerably more consequential than over-specification, which appears to incur little, if any, material efficiency loss in our simulations.

\subsection{Efficiency Comparison with Static Analysis}\label{c1:sec:sim_study:power}

The purpose of the third simulation experiment was to examine the effect of incorporating repeated measurements on power when estimating differences in mean end-of-study outcomes. With this inquiry in mind, we compared the method presented in this report with the static approach outlined in \citet{NeCamp2017} (which models the outcome only at time $t=2$). We hypothesized that, when within-unit correlation was high, modeling the outcome trajectory would present a more powerful approach than solely modeling the outcome at end-of-study. We further hypothesized that the two approaches would perform similarly in outcomes with lower within-person correlation.

NeCamp's earlier work on analyses of cSMARTs introduced a sample size formula for comparing embedded cAIs with respect to a single, end-of-study outcome \citep{NeCamp2017}. We intended for this simulation framework to align with the perspective of the primary aim design. Therefore, using this formula, we estimated the minimum sample size required to achieve $80\%$ power for detecting true effect sizes of $0.2$, $0.5$, $0.8$, and $1.0$ using NeCamp's static approach.

To generate data, we considered two separate approaches, intending to examine high- and low-correlation settings inspired by ASIC data. Both approaches involve an AR(1) correlation structure, homogeneous over embedded cAI ($\rho_{high}\approx 0.58,\s\rho_{low}\approx 0.27$). In Tables \ref{c1:tab:power_comp_lc} and \ref{c1:tab:power_comp}, we do not employ any finite-sample adjustments (as NeCamp's sample size formula does not incorporate such adjustments).

\begin{table}[H]
\centering
\begin{subtable}{\linewidth}
\centering
\begin{tabular}{ccccccc}
\hline
\multirow{2}{*}{Effect Size $(N)$}  & \multicolumn{2}{c}{RMSE} & \multicolumn{2}{c}{Coverage} & \multicolumn{2}{c}{Power} \\\cline{2-7}
     & Static & Long. & Static & Long. & Static & Long.\\
  \hline
$\delta=0.2$ ($N=633$) & 0.33 & 0.33 & 0.95 & 0.95 & 0.81 & 0.82 \\ 
  $\delta=0.5$ ($N=102$) & 0.83 & 0.82 & 0.94 & 0.94 & 0.82 & 0.83 \\ 
  $\delta=0.8$ ($N=40$) & 1.32 & 1.32 & 0.93 & 0.92 & 0.83 & 0.83 \\ 
  $\delta=1.0$ ($N=26$) & 1.65 & 1.65 & 0.91 & 0.91 & 0.84 & 0.85 \\ 
 \hline
\end{tabular}
\caption{Low Correlation}\label{c1:tab:power_comp_lc}
\end{subtable}
\begin{subtable}{\linewidth}
\centering\begin{tabular}{ccccccc}
\hline
\multirow{2}{*}{Effect Size $(N)$}  & \multicolumn{2}{c}{RMSE} & \multicolumn{2}{c}{Coverage} & \multicolumn{2}{c}{Power} \\\cline{2-7}
     & Static & Long. & Static & Long. & Static & Long.\\
  \hline
$\delta=0.2$ ($N=661$) &  7.03 &  6.18 & 0.95 & 0.95 & 0.81 & 0.90 \\ 
  $\delta=0.5$ ($N=106$) & 17.58 & 15.53 & 0.94 & 0.94 & 0.82 & 0.90 \\ 
  $\delta=0.8$ ($N=42$) & 28.15 & 24.77 & 0.93 & 0.93 & 0.83 & 0.91 \\ 
  $\delta=1.0$ ($N=27$) & 35.04 & 30.74 & 0.91 & 0.92 & 0.85 & 0.92 \\  
 \hline
\end{tabular}\caption{High Correlation}\label{c1:tab:power_comp}
\end{subtable}
\caption{Static vs Longitudinal Power Comparison}
\end{table}

Table \ref{c1:tab:power_comp} shows material power gains for the longitudinal estimator over the static estimator, in line with classic results from longitudinal analyses of two-armed randomized trials \protect\citep{Fitzmaurice2011}. In this \quotes{high correlation} setting, outcomes at times $t=0,1$ provide more information about outcomes at time $t=2$ than in the less-correlated setting. Similarly, the longitudinal estimator sees material improvement over its static counterpart in terms of RMSE in this environment as well.

While a longitudinal approach outperforms a static approach in a highly-correlated data-generative setting, the two approaches perform similarly in a low-correlation setting. Therefore, Tables \ref{c1:tab:power_comp_lc} and \ref{c1:tab:power_comp} suggest that, when repeated measures of an outcome are available, an analyst will not sacrifice statistical performance by modeling the outcome longitudinally.

We end this discussion by noting two important caveats to the results in Table \ref{c1:tab:power_comp}. First, the efficiency gains highlighted in this table are not due to modeling assumptions on the temporal trajectory of the outcome $Y$. In this simulation study, we used Equation~\protect\eqref{c1:eq:mmm_ex} as the marginal mean model. When restricted to $Y_{t_0}$, $Y_{t^*}$, and $Y_{t_T}$, this model is correctly specified under the randomization structure of a prototypical cSMART. This suggests that analysts can incorporate repeated measurements in end-of-study outcome comparisons to potentially achieve material precision gains without employing stringent modeling assumptions. Appendix~\ref{c1:sec:add_sim:oversat} discusses this phenomenon in the context of recent work on saturated longitudinal mean models in randomized trials (e.g., \citet{Wang2023}).
It is likely the case that including more time points in an analysis can cause further precision gains (as evidenced in Appendix~\protect\ref{c1:sec:add_sim:multi_t}), although it may expose the analyst to marginal mean model misspecification (as incorporating additional outcome measurements would highlight Equation~\eqref{c1:eq:mmm_ex}'s assumption of linear trends between times $t_0$ and $t^*$, as well as $t^*$ and $t_T$). Second, this pattern may cause a trial designer to reject NeCamp's sample size calculator in the hopes that incorporating repeated measurements can increase power. While these analyses suggest that one would not \textit{lose} power by incorporating repeated measurements, one is not guaranteed to see material power gain. Given the static/longitudinal power alignment in Table \ref{c1:tab:power_comp_lc}, we recommend the use of NeCamp's sample size formula to power a study even when repeated measurements are to be available, following a conservative approach to ensure adequate power under varying conditions.

\section{Data Analysis}\label{c1:sec:data_an}

This section presents an analysis of the ASIC trial, examining CBT delivery using the methods described above. Section \ref{c1:sec:data_an:rev} revisits the trial’s original primary aim, reanalyzing the difference between (\textit{REP+Coaching}, \textit{Facilitation}) and (\textit{REP}, \textit{No Facilitation}) with respect to expected SP-level aggregate CBT delivery, using repeated measurements rather than a single end-of-study measure. Section \ref{c1:sec:data_an:long} explores additional scientific questions that arise from the longitudinal trajectories made analyzable by the availability of repeated outcome measurements.

As in the original primary aims analysis, \citet{Smith2022}, we use multiple imputation with chained equations to address missing data \citep{Azur2011}. As in \citet{Smith2022}, all analyses shown in this section employ Rubin's rules for summarizing analyses across the multiply imputed data sets \citep{Rubin1987}.

\subsection{Original Primary Aim Analysis: Revisited}\label{c1:sec:data_an:rev}

In this section, we revisit the primary aim of the ASIC trial, employing the methods presented in this manuscript. As in \citet{Smith2022}, we condition on the six pre-registered baseline school-level covariates listed in \citet{Kilbourne2018} - $(i)$ whether the school has more or less than 500 students, $(ii)$ whether a majority of students at the school are on free/reduced lunch programs, $(iii)$ urbanicity of school (i.e., rural/urban), $(iv)$ aggregate school professionals' education (prior to randomization), $(v)$ aggregate school professionals' tenure (prior to randomization), $(vi)$ whether school professionals at the school delivered CBT prior to randomization.

While \citet{Smith2022} used the approach outlined in \citet{NeCamp2017} to model expected end-of-study outcome for each embedded cAI, we use the repeated CBT delivery measurements to model full mean trajectories. For these analyses, we employ the marginal mean model discussed in Section \ref{c1:sec:mm}, and a working variance model that is heterogeneous with respect to time and cAI, and has exchangeable between-individual and AR(1) within-person correlation structures. Table \ref{c1:tab:ASIC_parms} below presents results, with $\gamma_./\eta_.$ variables corresponding to the causal/nuisance parameters discussed in Section \ref{c1:sec:mm}. Table \ref{c1:tab:ASIC_corr_parms} in Appendix \ref{c1:sec:add_figs} shows the correlation estimates for each of the four embedded cAIs.

\begin{table}[H]
\centering
\begin{tabular}{ccccc}
  \hline
Parameter & Estimate & SE & T Score & $95\%$ Confidence Interval \\ \hline 
$\gamma_0$ & 0.50 & 0.19 & 2.60 & (0.12, 0.89) \\ 
  $\gamma_1$ & 0.14 & 0.05 & 3.01 & (0.05, 0.24) \\ 
  $\gamma_2$ & -0.04 & 0.05 & -0.84 & (-0.14, 0.06) \\ 
  $\gamma_3$ & 0.05 & 0.02 & 2.73 & (0.01, 0.08) \\ 
  $\gamma_4$ & 0.01 & 0.02 & 0.28 & (-0.03, 0.04) \\ 
  $\gamma_5$ & 0.02 & 0.01 & 1.51 & (-0.01, 0.05) \\ 
  $\gamma_6$ & 0.00 & 0.01 & 0.04 & (-0.03, 0.03) \\ 
  $\eta_{Lunch}$ & 0.06 & 0.32 & 0.17 & (-0.58, 0.69) \\ 
  $\eta_{Delivered}$ & 0.33 & 0.30 & 1.10 & (-0.27, 0.93) \\ 
  $\eta_{Urbanicity}$ & 0.09 & 0.36 & 0.26 & (-0.62, 0.81) \\ 
  $\eta_{Size}$ & -0.05 & 0.41 & -0.12 & (-0.86, 0.77) \\ 
  $\eta_{SP Edu}$ & 0.49 & 0.58 & 0.85 & (-0.65, 1.64) \\ 
  $\eta_{SP Ten}$ & 0.00 & 0.03 & 0.03 & (-0.05, 0.05) \\ 
   \hline
\end{tabular}
\caption{Marginal Mean Parameter Estimates - ASIC}\label{c1:tab:ASIC_parms}
\end{table}

The estimates in Table \ref{c1:tab:ASIC_parms} induce mean trajectory estimates for trends in CBT delivery by embedded cAI, shown in Figure \ref{c1:fig:ASIC_traj}.

\begin{figure}[ht]
    \centering
    \begin{subfigure}[t]{0.39\textwidth}
        \centering
        \includegraphics[width=\textwidth]{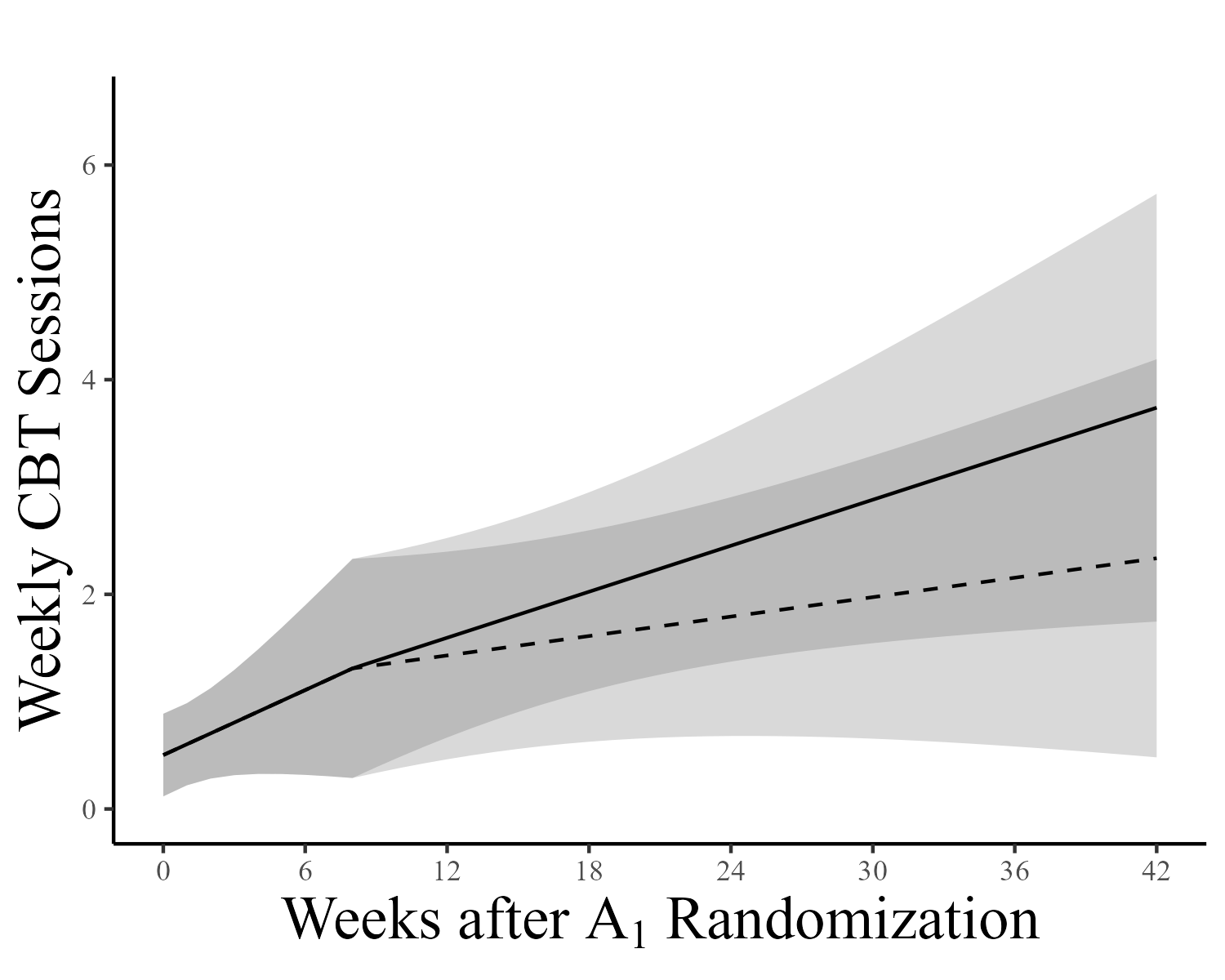}
        \caption{$A_1$ = \textit{REP+Coaching}}
    \end{subfigure}
    \hfill
    \begin{subfigure}[t]{0.39\textwidth}
        \centering
        \includegraphics[width=\textwidth]{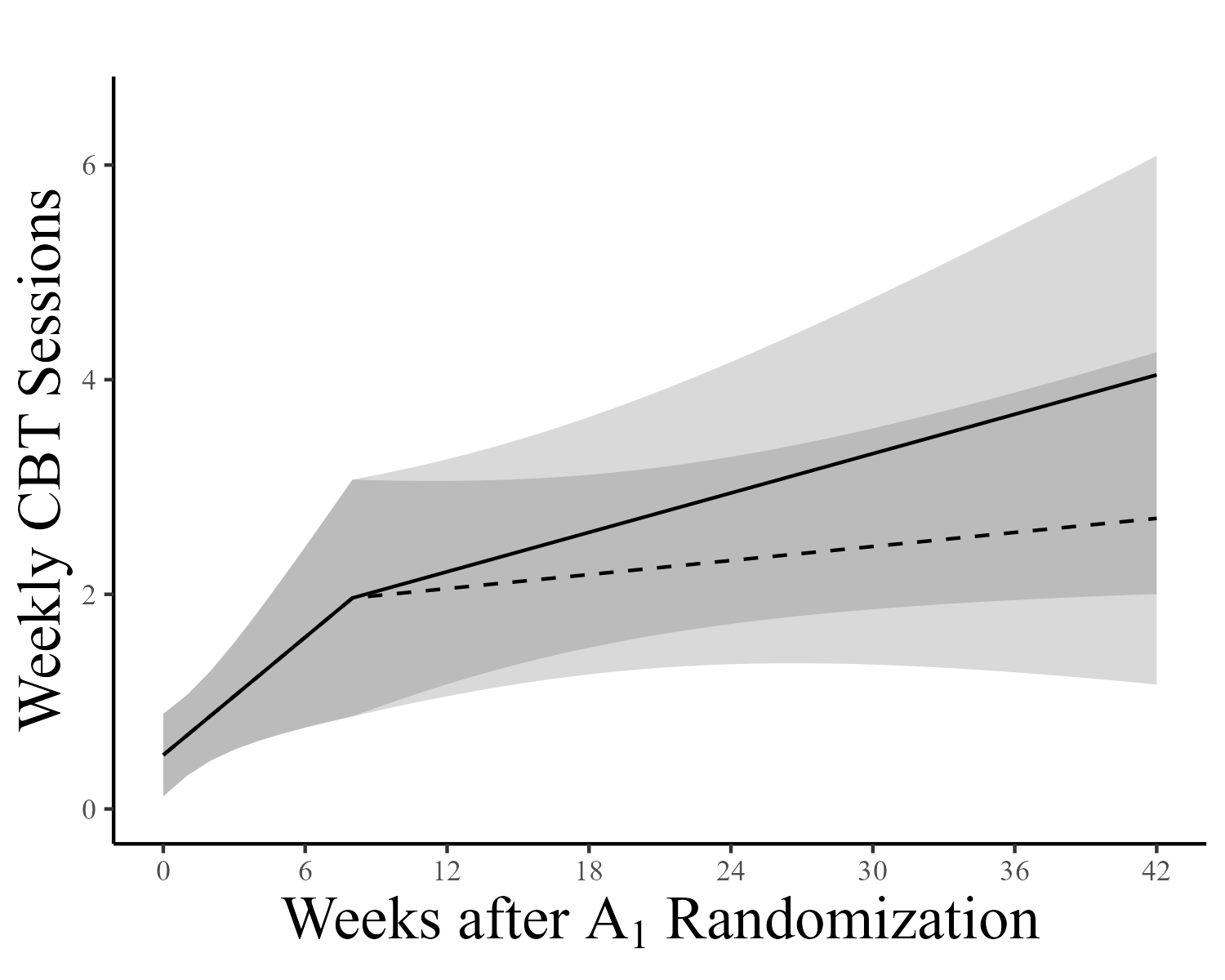}
        \caption{$A_1$ = \textit{REP}}
    \end{subfigure}
    \raisebox{2.75em}{
    \begin{subfigure}[t]{0.2\textwidth}
        \centering
        \includegraphics[width=\textwidth]{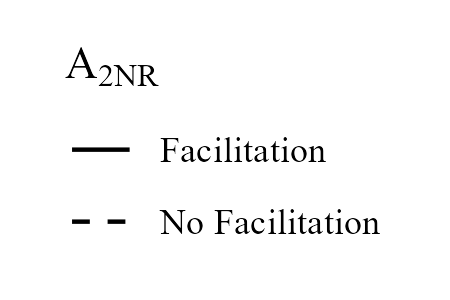}
    \end{subfigure}}
    \captionsetup{belowskip=-7pt}
    \caption{Average Weekly CBT Sessions by Embedded cAI - ASIC}
     \caption*{\raggedright\scriptsize\textit{Shaded regions represent 95\% confidence intervals.}}
    \label{c1:fig:ASIC_traj}
\end{figure}

Figure \ref{c1:fig:ASIC_ci_comp} shows the point estimates and $95\%$ confidence intervals for the six pairwise contrasts in embedded cAIs with respect to the ASIC primary outcome. As shown in the figure, none of the embedded cAIs were significantly different from each other in this respect. In particular, we see a near-zero effect in the primary aim comparison. Subsequently, we would fail to reject the null hypothesis that $\E{\sum_{l=0}^T Y_{t_l}^{(1,1)}}=\E{\sum_{l=0}^T Y_{t_l}^{(-1,-1)}}$. This does not suggest a null effect for either embedded cAI; rather, it merely indicates a lack of evidence for a difference in the two.

Furthermore, as the primary aim in question is a static end-of-study comparison, we can compare our approach with the existing approach to analyze such aims presented in \citet{NeCamp2017}. As shown in Figure \ref{c1:fig:ASIC_ci_comp}, both approaches give similar point estimates for the difference in expected total CBT delivery in each of the six pairwise comparisons between embedded cAIs. However, the estimates using the longitudinal approach had $95\%$ confidence intervals that were $\approx 26\%$ narrower, on average, versus only using an end-of-study outcome. This pattern further supports the efficiency gains described in the simulation results in Section \protect\ref{c1:sec:sim_study:power}.

\begin{figure}[H]
    \centering
    \includegraphics[width=12cm]{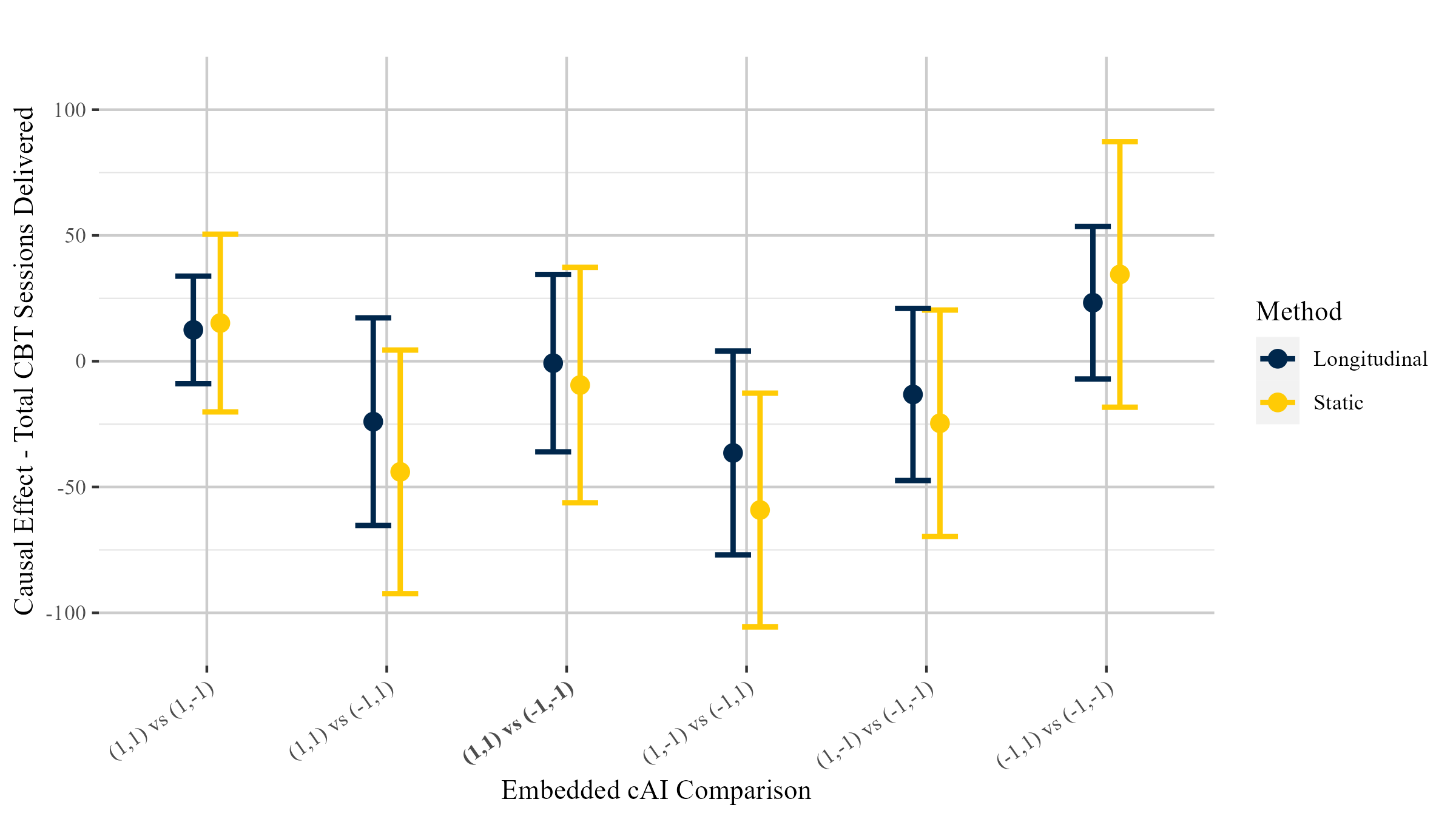}
    \captionsetup{belowskip=-7pt}
    \caption{Longitudinal and Static Estimates of CBT Delivery Differences by cAI - ASIC}
    \caption*{\raggedright\scriptsize\textit{Note: The primary aim embedded cAI comparison was $(1,1)$ vs $(-1,-1)$ (bolded in figure). For this comparison, the longitudinal point estimate was -0.77 (p-value 0.966) and the static point estimate was -9.48 (p-value 0.688).}}
    \label{c1:fig:ASIC_ci_comp}
\end{figure}

\subsection{Longitudinal Follow-Up Analyses}\label{c1:sec:data_an:long}
\subsubsection{Second-Stage Slope}
Examining Figure \ref{c1:fig:ASIC_traj} presents several insights into the trajectory of CBT delivery that would be masked in an end-of-study analysis. As shown in the figure, after the second decision point, the marginal mean trajectories for CBT delivery for (\textit{REP}, \textit{Facilitation}) and (\textit{REP+Coaching}, \textit{Facilitation}) both outpace their \quotes{\textit{No Facilitation}} counterparts.

While the methods presented in this manuscript were introduced in the context of primary aim analyses, they are equally applicable to secondary and exploratory aim analyses. In particular, collecting and analyzing repeated measurements of CBT delivery allows us to conduct inference to investigate the impact of \textit{Facilitation} by comparing second-stage treatment slopes for CBT delivery. Table \ref{c1:tab:ASIC_slope} below shows the results of this comparison for the six pairwise comparisons of the four embedded cAIs in ASIC.

\begin{table}[H]
\centering
\begin{tabular}{ccccc}
  \hline
Pairwise cAI Comparison & Estimate & SE & $95\%$ Confidence Interval \\ 
  \hline
(1,1) vs (1,-1) & 0.041 & 0.04 & (-0.03, 0.12) \\ 
  (1,1) vs (-1,1) & 0.010 & 0.04 & (-0.08, 0.10)  \\ 
  (1,1) vs (-1,-1) & 0.050 & 0.04 & (-0.03, 0.13)  \\ 
  (1,-1) vs (-1,1) & -0.031 & 0.04 & (-0.11, 0.05) \\ 
  (1,-1) vs (-1,-1) & 0.008 & 0.04 & (-0.07, 0.09)  \\ 
  (-1,1) vs (-1,-1) & 0.039 & 0.03 & (-0.03, 0.11)  \\ 
   \hline
\end{tabular}
\caption{Comparison of Second-Stage Slope for Weekly CBT Delivery - ASIC}\label{c1:tab:ASIC_slope}
\end{table}

\subsubsection{Nonlinear Temporal Trend}
Given the time between second-stage treatment decision and end-of-study, a trial designer may prefer a marginal mean model that evolves at a rate slower than $t$ \citep{Borghi2005}. For example,
\begin{equation} \label{c1:eq:mmm_sqrtt}
\begin{split}
    \Tilde{\mu}_t(a_1, a_{2NR}, X; \theta) & = 
     \gamma_0 + \eta X + \Ind{t\leq t^*}\paren{\gamma_1\sqrt{t} + \gamma_2a_1\sqrt{t}} + \Ind{t > t^*}\Bigg(\gamma_1\sqrt{t^*} + \gamma_2a_1\sqrt{t^*} \\
     & \hspace{2cm}  + 
     \paren{\gamma_3 + \gamma_4a_1 + \gamma_5a_{2NR} + \gamma_6a_1a_{2NR}}\paren{\sqrt{t}-\sqrt{t^*}}\Bigg).
\end{split}
\end{equation}

Figure \ref{c1:fig:ASIC_sqrt_t} in Appendix \ref{c1:sec:add_figs} shows the marginal mean trajectories under model $\Tilde{\mu}$. Furthermore, we revisit the analysis of second-stage slope under this alternate model in Table \ref{c1:tab:ASIC_slope_sqrtt} in Appendix \ref{c1:sec:add_figs}. As the figure and table suggest, this approach yielded results that were largely consistent with the original, providing little evidence of meaningful difference.


\subsubsection{ICC Trend}
While not typically a primary analytic focus, examining the intraclass correlation coefficient (ICC) over time can shed light on the evolution of between-cluster variability across intervention stages. To explore this, we rerun the analysis from Section~\ref{c1:sec:data_an:rev} using an unstructured between-person working correlation matrix in place of the exchangeable structure, allowing for time-varying estimates of the ICC.

Table~\ref{c1:tab:ASIC_ICC_trend} reports the ICC at five time points: $t = 0$, 4, 8, 25, and 42 (corresponding to the first decision point, midpoint of the first stage, second decision point, midpoint of the second stage, and end of study, respectively). The table shows a sharp decline in ICC for (\textit{REP+Coaching}, \textit{Facilitation}) following the first randomization, whereas the ICC for (\textit{REP}, \textit{No Facilitation}) remains relatively stable over time. In hindsight, the contrast is reasonable given the structure of the corresponding cAIs. Although \textit{Coaching} and \textit{Facilitation} were delivered at the school level, they primarily involved one-on-one interactions between trial staff and individual school professionals. This may help explain the reduced within-cluster cohesion observed in the (\textit{REP+Coaching}, \textit{Facilitation}) cAI, relative to the more stable ICC pattern seen in (\textit{REP}, \textit{No Facilitation}), which lacked these individual-oriented components.

\begin{table}[ht]
\centering
\begin{tabular}{cccccc}
  \hline
Embedded cAI & $\rho_0$ & $\rho_4$ & $\rho_8$ & $\rho_{25}$ & $\rho_{42}$ \\ 
  \hline
(\textit{REP+Coaching}, \textit{Facilitation}) & 0.42 & 0.10 & 0.18 & 0.05 & 0.14 \\ 
  (\textit{REP}, \textit{No Facilitation}) & 0.42 & 0.45 & 0.44 & 0.42 & 0.28 \\
   \hline
\end{tabular}
\caption{ICC Trajectories for Primary Aim Comparison - ASIC}\label{c1:tab:ASIC_ICC_trend}
\end{table}

More abstractly, examining ICC trends over time may help intervention designers identify intervention components that contribute more strongly to individual-level, as opposed to cluster-level, variation. As this falls outside the scope of traditional primary aims, we do not pursue formal inference here.

\section{Discussion}

This paper details a novel approach for analyzing longitudinal data arising from clustered SMARTs. This method equips domain scientists with a tool to explore temporal patterns in outcomes in their cSMART analyses. In the context of designing adaptive interventions, understanding the trajectory of improvement in a target outcome is often of scientific importance. While end-of-study analyses in cSMARTs can provide valuable insights, they focus on cumulative effects and may not capture the temporal dynamics that are crucial for optimizing adaptation strategies in a cAI \citep{NahumShani2020}. For example, as seen in the Section~\ref{c1:sec:data_an:long} ASIC analyses, incorporating repeated measurements of CBT delivery introduces a number of scientifically meaningful questions about how intervention effects evolve over time. Additionally, Section \ref{c1:sec:data_an} and the power analysis in Section \ref{c1:sec:sim_study:power} demonstrated that incorporating repeated outcome measurements can yield more precise estimates of traditional end-of-study objectives compared to existing methods that rely solely on the final outcome. This raises the potential for future work guiding decisions about the number and timing of outcome measurements in clustered SMARTs \citep{Raudenbush2001}.

Following the introduction of methods for analyzing repeated measures in individually-randomized SMARTs, researchers have been able to address a wider range of substantive questions regarding temporal treatment effects. These more detailed SMART analyses hold significant promise for the design of adaptive interventions, offering researchers deeper insight to inform adaptive intervention development \citep{NahumShani2020}. Similarly, exploring such questions in clustered settings can give domain scientists the tools to construct more effective clustered adaptive interventions.

\newpage

\thispagestyle{empty}  

\phantomsection
\addcontentsline{toc}{chapter}{Technical Appendices}

\null
\vfill

\begin{center}
    \LARGE \textbf{Supplementary Material:}\\
    \textbf{Technical Appendices}
\end{center}

\vfill
\null

\newpage

\setcounter{section}{0}
\setcounter{table}{0}
\setcounter{figure}{0}

\renewcommand{\thesection}{A\arabic{section}}
\renewcommand{\thesubsection}{\thesection.\arabic{subsection}}
\renewcommand{\thesubsubsection}{\thesubsection.\arabic{subsubsection}}

\renewcommand{\thetable}{A\arabic{table}}
\renewcommand{\thefigure}{A\arabic{figure}}

\renewcommand{\theHsection}{appendix.\arabic{section}}
\renewcommand{\theHtable}{appendix.\arabic{table}}
\renewcommand{\theHfigure}{appendix.\arabic{figure}}

\addtocontents{toc}{\protect\setlength{\cftsecnumwidth}{2.5em}}
\addtocontents{toc}{\protect\setlength{\cftsubsecnumwidth}{3.5em}}
\addtocontents{toc}{\protect\setlength{\cftsubsubsecnumwidth}{4.5em}}

\setcounter{equation}{0}
\renewcommand{\theequation}{A.\arabic{equation}}
\renewcommand{\theHequation}{appendix.\arabic{equation}}

\section{Additional Figures and Tables}\label{c1:sec:add_figs}

\begin{table}[ht]
\centering
\begin{tabular}{ccc}
  \hline
Embedded cAI & \makecell{$\rho_w$ \\ (AR(1) Structure)} & \makecell{$\rho_b$ \\(Exchangeable Structure)}\\ 
  \hline
(\textit{REP+Coaching}, \textit{Facilitation}) & 0.09 & 0.14 \\ 
  (\textit{REP+Coaching}, \textit{No Facilitation}) & 0.44 & 0.05 \\ 
  (\textit{REP}, \textit{Facilitation}) & 0.30 & 0.03 \\ 
  (\textit{REP}, \textit{No Facilitation}) & 0.33 & 0.23 \\ 
   \hline
\end{tabular}
\caption{Marginal Correlation Estimates - ASIC}\label{c1:tab:ASIC_corr_parms}
\end{table}

\begin{table}[ht]
\centering
\begin{tabular}{cccc}
  \hline
Comparison & Estimate & SE & CI \\ 
  \hline
(1,1) vs (1,-1) & 0.04 & 0.04 & (-0.04, 0.11)  \\ 
  (1,1) vs (-1,1) & 0.01 & 0.04 & (-0.07, 0.09) \\ 
  (1,1) vs (-1,-1) & 0.05 & 0.04 & (-0.04, 0.13)  \\ 
  (1,-1) vs (-1,1) & -0.03 & 0.04 & (-0.11, 0.05) \\ 
  (1,-1) vs (-1,-1) & 0.01 & 0.04 & (-0.07, 0.09) \\ 
  (-1,1) vs (-1,-1) & 0.04 & 0.03 & (-0.03, 0.10) \\ 
   \hline
\end{tabular}
\caption{ASIC: Second-Stage Slope Comparison for Alternate Marginal Model}
\label{c1:tab:ASIC_slope_sqrtt}
\end{table}

\begin{figure}[H]
    \centering
    \includegraphics[width=10cm]{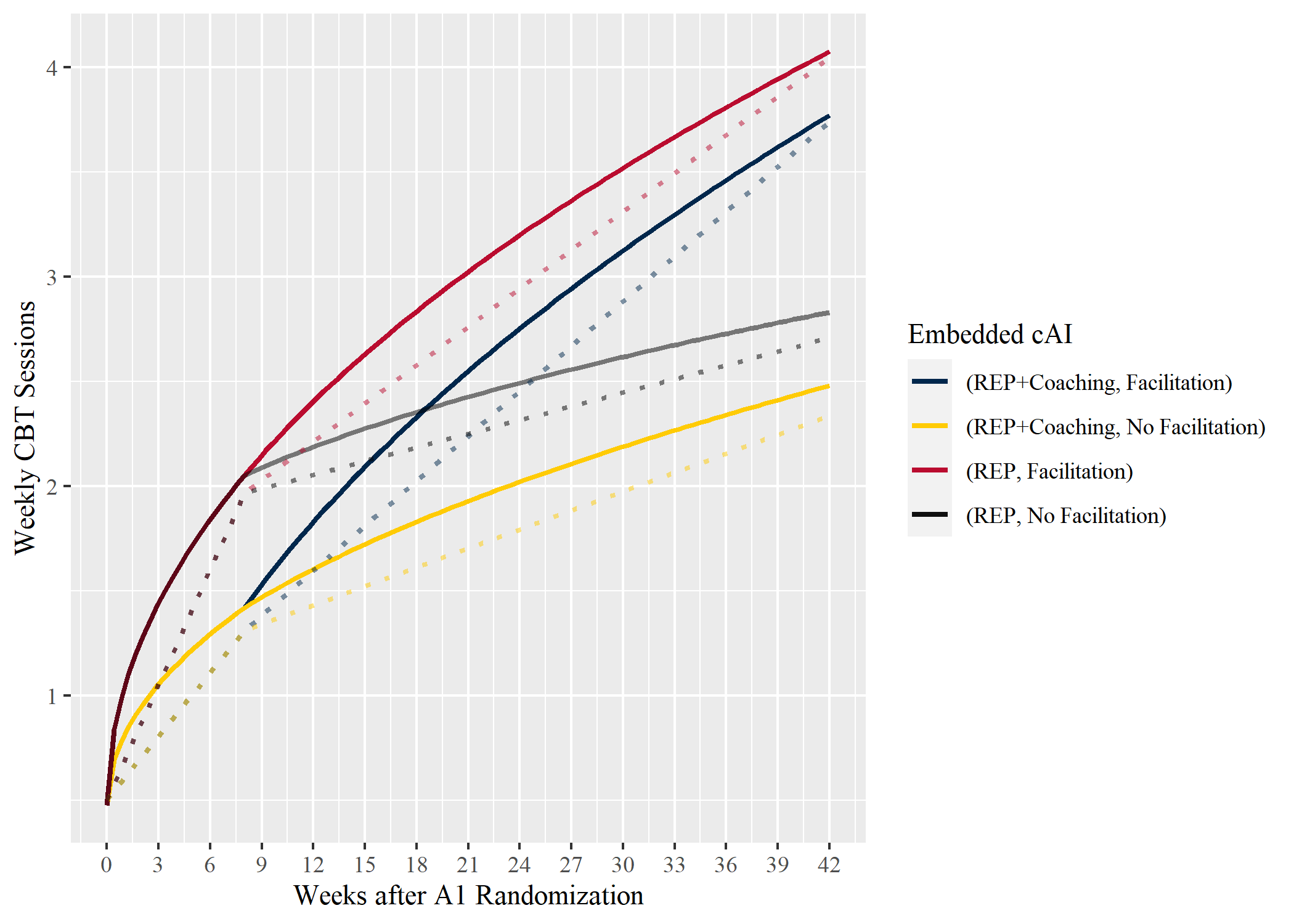}
    \captionsetup{belowskip=-10pt}
\caption{ASIC Trajectories under Alternate Marginal Mean Model}
    \caption*{\raggedright\scriptsize\textit{Dotted lines represent linear trajectories under $\mu$.}}
    \label{c1:fig:ASIC_sqrt_t}
\end{figure}

\begin{figure}[H]
    \centering
    \begin{subfigure}{0.98\linewidth}
        \centering
        \includegraphics[width=\linewidth, keepaspectratio]{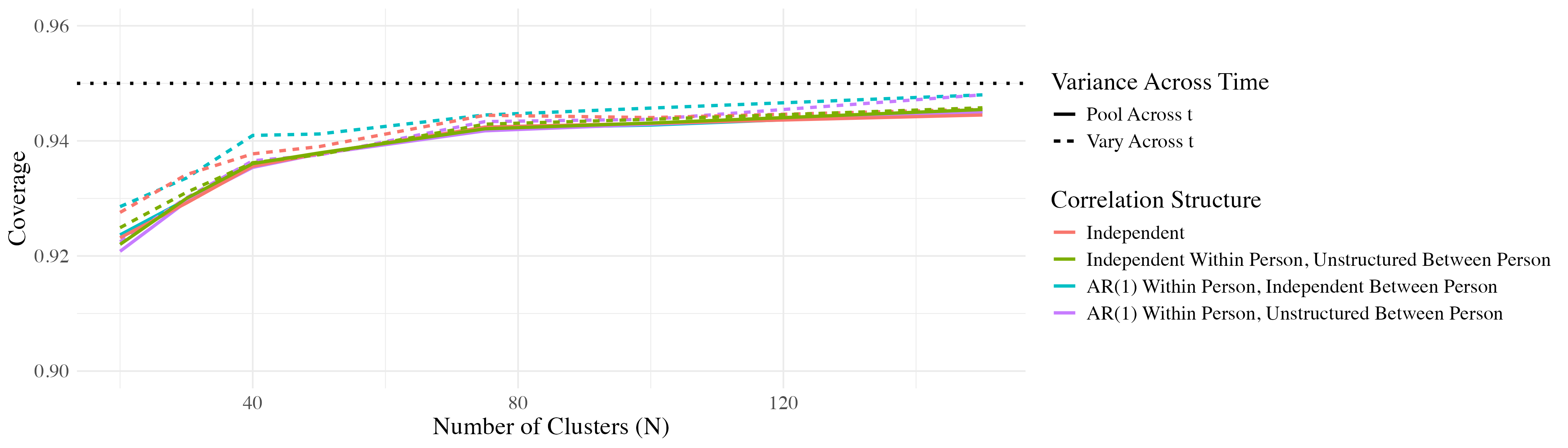}
\caption{Coverage Rates}
        \label{c1:fig:wvm_miss_app:cov}
    \end{subfigure}
    \begin{subfigure}{0.98\linewidth}
        \centering
        \includegraphics[width=\linewidth]{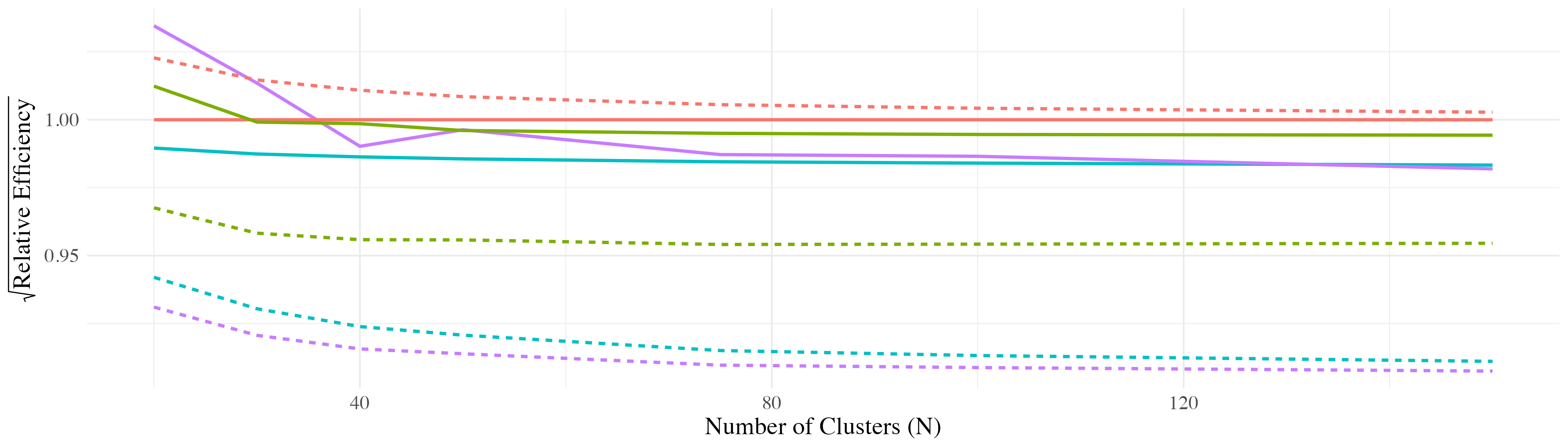}
\caption{Relative Efficiency (Homoscedastic-Independent Baseline)}
        \label{c1:fig:wvm_miss_app:se}
    \end{subfigure}
    \captionsetup{belowskip=-7pt}
\caption{Performance of Different Working Variance Models Across Sample Sizes: Complex Data-Generative Model}
    \label{c1:fig:wvm_miss_app}
    \caption*{\raggedright\scriptsize\textit{Subfigure (b) presents the same information as Figure~\ref{c1:fig:wvm_miss} in the main body, with relative efficiency calculated with respect to an independent time-homoscedastic working variance model. We present it, reformatted, for ease of comparison with Figure~\ref{c1:fig:wvm_miss_iid} below.}}
    \caption*{\raggedright\scriptsize\textit{Data generated under a time-heteroscedastic variance structure that was homogeneous across embedded cAIs, with AR(1) within-person correlation and unstructured between-person correlation.}}
    
\end{figure}

\begin{figure}[H]
    \centering
    \begin{subfigure}{0.98\textwidth}
        \centering
        \includegraphics[width=\linewidth, keepaspectratio]{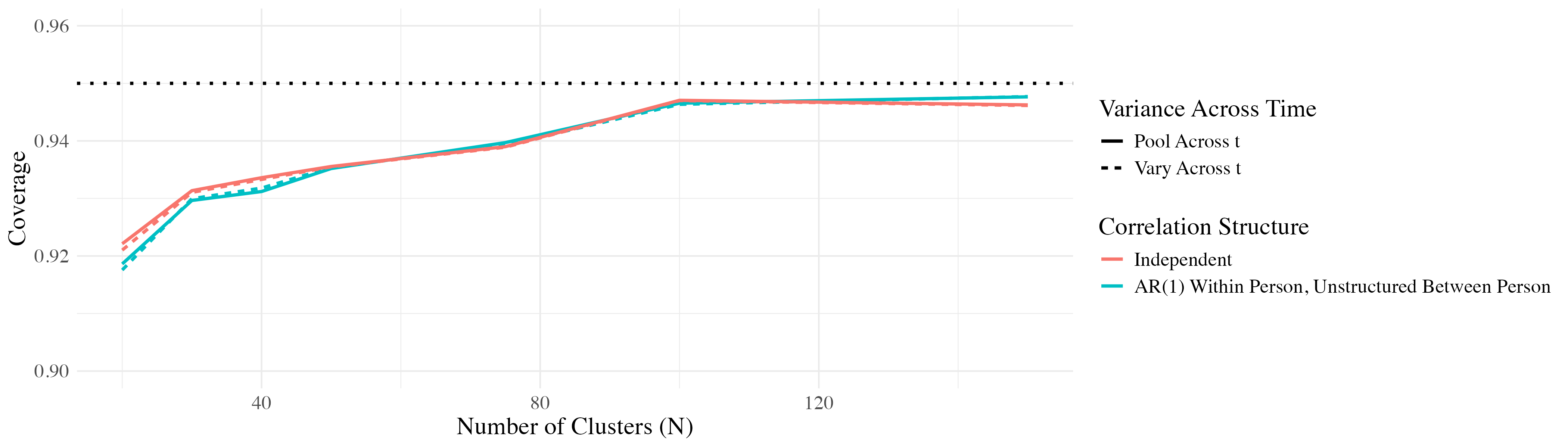}
\caption{Coverage Rates}
        \label{c1:fig:wvm_miss_iid:cov}
    \end{subfigure}
    \begin{subfigure}{0.98\textwidth}
        \centering
        \includegraphics[width=\linewidth]{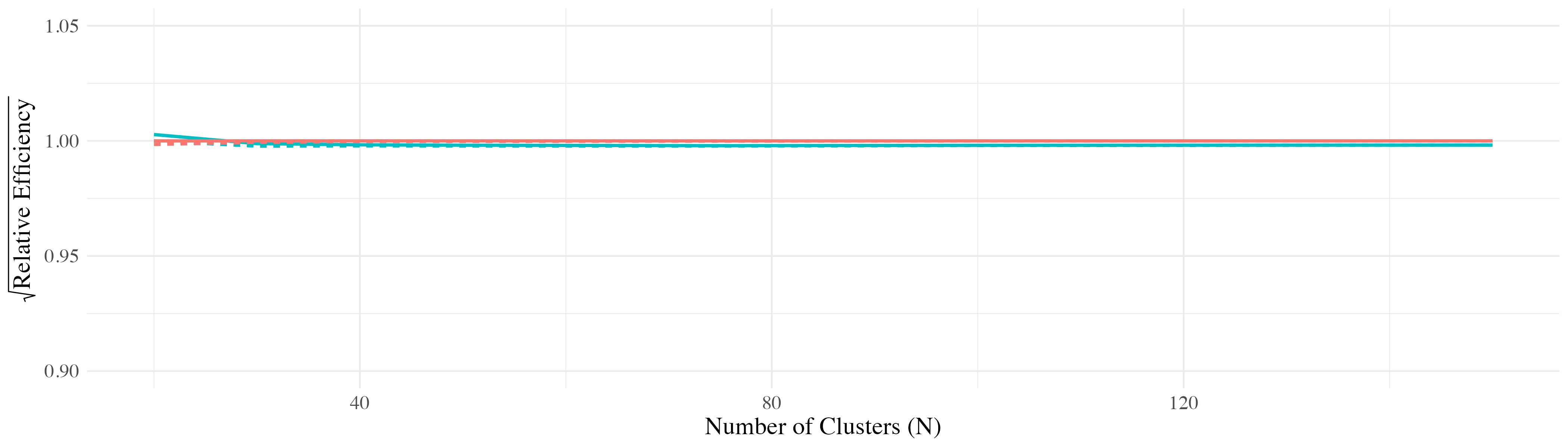}
\caption{Relative Efficiency (Homoscedastic-Independent Baseline)}
        \label{c1:fig:wvm_miss_in:se}
    \end{subfigure}
\caption{Performance of Different Working Variance Models Across Sample Sizes: Simple Data-Generative Model}
    \label{c1:fig:wvm_miss_iid}
    \caption*{\raggedright\scriptsize\textit{Data generated under a time-homoscedastic variance structure that was homogeneous across embedded cAIs, with independent correlation.}}
\end{figure}

\section{Further Discussion of Working Variance}\label{c1:sec:var_dis}
As discussed in Section \ref{c1:sec:mm_work_var} of the main body, we let $V^d\paren{\bX;\alpha}$ denote our working variance model; parameterized by $\alpha=\begin{bmatrix}
    \sigma\\ \rho
\end{bmatrix}$, where $V^d\paren{\bX_i;\alpha}$ takes the general form
$$V^d\paren{\bX_i;\alpha} = \bS_i^d(\sigma)\poh \boldsymbol\Rho_i^d(\rho) \bS_i^d(\sigma)\poh.$$

Here, the $n_i(T+1)\times n_i(T+1)$ matrices $\bS_i^d(\sigma)$ and $\boldsymbol\Rho_i^d(\rho)$ correspond to variance and correlation matrices (respectively), and have the broad structures described below.

$$\bS_i^d(\sigma) = 
\begin{bmatrix}
    S^d(\sigma) & & 0\\
     & \ddots & \\
    0 & & S^d(\sigma)\\
\end{bmatrix}\text{, } 
\boldsymbol\Rho_i^d(\rho) = 
\begin{bmatrix}
    W^d(\rho) & B^d(\rho) & \dots & B^d(\rho)\\
    B^d(\rho) & W^d(\rho)  & & \\
    \vdots &  & \ddots & \\
    B^d(\rho) & & & W^d(\rho)\\
\end{bmatrix},$$

where $S^d(\sigma)$ is a $(T+1)\times (T+1)$ diagonal matrix with entries corresponding to $\V{Y_{t_0}\d \mid \bX} \ddd \V{Y_{t_T}\d\mid \bX}$ and $W^d(\rho)$ and $B^d(\rho)$ are within- and between-person $(T+1)\times (T+1)$ correlation matrices representing $corr\paren{\vct{Y}\d_{i,j,.}, \vct{Y}\d_{i,j,.}\mid \bX}$ and $corr\paren{\vct{Y}\d_{i,j,.}, \vct{Y}\d_{i,j',.}\mid \bX}$ (respectively).

The exact forms of $S^d(\sigma)$, $W^d(\rho)$, $B^d(\rho)$ depend on the variance structure the analyst chooses to model. Using our illustrative model from Section \ref{c1:sec:mm_work_var} of the main body, we can consider a working variance model that is exchangeable within-person and between-person. Furthermore, we will assume heterogeneous variance across time and adaptive intervention. This choice of variance structure induces the following forms of $S^d(\sigma)$, $W^d(\rho)$, and $B^d(\rho)$
$$
\begin{bmatrix}
    \paren{\sigma^{d}_{t_0}}^2 & & 0\\
    & \ddots & \\
    0 & & \paren{\sigma^{d}_{t_T}}^2
\end{bmatrix},\hspace{.5cm}
\begin{bmatrix}
    1 & & \rho^{d}_w\\
    & \ddots & \\
    \rho^{d}_w & & 1
\end{bmatrix},\text{ and}\hspace{.5cm}
\begin{bmatrix}
    \rho^{d}_b & & \rho^{d}_b\\
    & \ddots & \\
    \rho^{d}_b & & \rho^{d}_b
\end{bmatrix},
$$
respectively, where $\sigma = \begin{bmatrix} \sigma^{d_1}_{t_0}, & \sigma^{d_1}_{t_1}, & \dots & \sigma^{d_1}_{t_T}, & \sigma^{d_2}_{t_0}, & \dots  & \sigma^{d_D}_{t_T}
\end{bmatrix}^\top$, and \\
$\rho = \begin{bmatrix}
    \rho^{d_1}_w, & \rho^{d_1}_b, & \rho^{d_2}_w, & \rho^{d_2}_b, & \dots & \rho^{d_D}_w, & \rho^{d_D}_b
\end{bmatrix}^\top$.

\section{Variance Estimation}\label{c1:sec:var_est}
Section \ref{c1:sec:est} of the main body discusses how to construct variance estimates under an example working variance for each step in the model fitting procedure presented in the main body of the report. This appendix discusses how to adapt this algorithm for alternative working variance structures.

Recall the general form of the working variance, discussed in detail in Appendix \ref{c1:sec:var_dis}:
$$V^d\paren{\bX_i;\alpha} = \bS_i^d(\sigma)\poh \boldsymbol\Rho_i^d(\rho) \bS_i^d(\sigma)\poh.$$

This appendix concerns the estimation of the $\alpha$ parameters which define $V^d\paren{\bX_i;\alpha}$. In the case of repeated measures in clustered SMARTs, this involves three steps: Estimation of the marginal variance components, estimation of the within-person correlation components, and estimation of the between-person correlation components.

Recall from Section \ref{c1:sec:est} that variance estimation at each step ($s$) in the iterative model fitting algorithm relies on the model residuals from the previous step: 
$\hat{\vct{\varepsilon}}^{d,s}_i=\vbY_i - \vbmu(d,\bX_i; \thetahat_{s-1})$.

\subsection{Marginal Variance Components}
We first seek an estimate of $\V{Y_{i,j,t}\d\mid \bX_i}$. We consider two separate modeling decisions:
\begin{enumerate}
    \item $S^d(\sigma) = 
    \begin{bmatrix}
        \paren{\sigma^{d}_{t_0}}^2 & & 0\\
        & \ddots & \\
        0 & & \paren{\sigma^{d}_{t_T}}^2
    \end{bmatrix}$, and
    \item $S^d(\sigma) = 
    \begin{bmatrix}
        \paren{\sigma_{t_0}}^2 & & 0\\
        & \ddots & \\
        0 & & \paren{\sigma_{t_T}}^2
    \end{bmatrix}.$
\end{enumerate}
The first case describes a working variance model in which the marginal variances can differ across embedded AI, whereas the second describes a marginal variance model that is homogeneous with respect to embedded AI. Table \ref{c1:tab:var_est:mar_var} below presents estimators for these quantities.

\begin{table}[ht]
    \centering
    \begin{tabular}{ccc}
    \hline
    \makecell{Marginal \\Variance \\Structure} & $\V{Y_{i,j,t}\d\mid \bX_i}$ & Estimator \\
    \hline
    \makecell{Heterogeneous \\with respect to \\embedded cAI}  &   
    $\paren{\sigma_t^{d}}^2$
    & $\paren{\sigmahat^{d}_t}^2  = 
    \ddfrac{\Sum{i}{N}I_i(d)W_i\Sum{j}{n_i}\paren{\hat{\varepsilon}^{d,s}_{i,j,t}}^2}{\Sum{i}{N}I_i(d)W_in_i}$\\
    \makecell{Homogeneous \\with respect to \\embedded cAI} & 
    $\sigma_t^2$ &  $\paren{\sigmahat_t}^2  = 
    \ddfrac{\Sum{i}{N}\Sumd I_i(d)W_i\Sum{j}{n_i}\paren{\hat{\varepsilon}^{d,s}_{i,j,t}}^2}{\Sum{i}{N}\Sumd I_i(d)W_in_i}$\\
    \hline
    \end{tabular}
\caption{Marginal Variance Estimators}
    \label{c1:tab:var_est:mar_var}
\end{table}

In addition to pooling over embedded AIs, a researcher could also pool over time, taking either
$$
    \left(\sigmahat^d\right)^2  = \frac{1}{T+1}\sum_{k=0}^{T} \left(\sigmahat^d_{t_k}\right)^2\text{ or }
    \left(\sigmahat\right)^2 = \frac{1}{T+1}\sum_{k=0}^{T} \left(\sigmahat_{t_k}\right)^2,
$$
depending on whether they wanted to pool over embedded cAI as well (where $\sigmahat_t$ and $\sigmahat_t^d$ are as obtained above).

\subsection{Within-Person Correlation}
After obtaining marginal variance estimates, we must estimate within-person correlation components. This corresponds to estimating the $(T+1)\times (T+1)$ diagonal blocks of $\boldsymbol\Rho_i$.

We will consider four different structures for these blocks. Here, we present estimators that are heterogeneous with respect to embedded AI; however, one could homogenize over adaptive interventions just as in Table \ref{c1:tab:var_est:mar_var}.\footnote{That is, by summing over $d\in\mathcal{D}$ in both the numerator and denominator of the estimator.} 
We note that, while an analyst may wish to pool across time and/or embedded cAI to estimate marginal variance components, the $\sigmahat^d_{t_l}$ terms in Tables \ref{c1:tab:var_est:wp_corr} and \ref{c1:tab:var_est:bp_corr} should not be pooled over time or embedded cAI to ensure stable correlation estimates.

\subsection{Between-Person Correlation}
Finally, we must estimate within-cluster (between-person) correlation components. This corresponds to estimating the $(T+1)\times (T+1)$ off-diagonal blocks of $\boldsymbol\Rho_i$.

We consider three different structures for these blocks. As before, we present estimators that are heterogeneous with respect to embedded AI; however, one could homogenize over adaptive interventions just as in Table \ref{c1:tab:var_est:mar_var}.

\begin{landscape}

\begin{table}[ht]
    \centering
    \begin{tabular}{ccc}
    \hline
    \makecell{Within-Person\\ Correlation\\ Structure} & $corr\paren{Y_{i,j,t_l}\d, Y_{i,j,t_m}\d\mid\bX_i} $ & Estimator\\
    \hline
    AR(1)  &  $\begin{cases}
        1 & \text{ if $t_l=t_m$}\\
        \paren{\rho_w^d}^{|l-m|} & \text{ otherwise}
    \end{cases}$ 
    & $\rhohat^d_{w, AR(1)} = 
    \ddfrac{\Sum{i}{N}I_i(d)W_i\Sum{j}{n_i}\Sumz{l}{T-1}
\paren{\frac{\hat{\varepsilon}^{d,s}_{i,j,t_l}}{\sigmahat^d_{t_l}}}
\paren{\frac{\hat{\varepsilon}^{d,s}_{i,j,t_{l+1}}}{\sigmahat^d_{t_{l+1}}}}}
{\Sum{i}{N}I_i(d)W_in_iT}$\\
    Exchangeable & 
    $\begin{cases}
        1 & \text{ if } t_l=t_m \\
        \rho_w^d & \text{ otherwise}
    \end{cases}$ &  
    $\rhohat^d_{w, Ex} = \ddfrac{\Sum{i}{N}I_i(d)W_i
    \Sum{j}{n_i}\Sumz{l}{T}\displaystyle\sum_{\substack{m=0\\m\neq l}}^T
    \paren{\frac{\hat{\varepsilon}^{d,s}_{i,j,t_{l}}}{\sigmahat^d_{t_{l}}}}
    \paren{\frac{\hat{\varepsilon}^{d,s}_{i,j,t_m}}{\sigmahat^d_{t_m}}}}
    {\Sum{i}{N}I_i(d)W_in_i(T+1)T}$\\
    Unstructured & 
    $\begin{cases}
        1 & \text{ if } t_l=t_m \\
        \rho^d_{w, t_l, t_m} & \text{ otherwise}
    \end{cases}$ &  
    $\rhohat^d_{w, Un, t_l, t_m} = \ddfrac{\Sum{i}{N}I_i(d)W_i\Sum{j}{n_i}
    \paren{\frac{\hat{\varepsilon}^{d,s}_{i,j,t_{l}}}{\sigmahat^d_{t_{l}}}}
    \paren{\frac{\hat{\varepsilon}^{d,s}_{i,j,t_m}}{\sigmahat^d_{t_m}}}}
    {\Sum{i}{N}I_i(d)W_in_i}$\\
    Independent & 
    $\begin{cases}
        1 & \text{ if } t_l=t_m \\
        0 & \text{ otherwise}
    \end{cases}$ &  
    $\rhohat^d_{w, In} = 0$\\
    \hline
    \end{tabular}
\caption{Within-Person Correlation Estimators}
    \label{c1:tab:var_est:wp_corr}
\end{table}

\begin{table}[ht]
    \centering
    \begin{tabular}{ccc}
    \hline
    \makecell{Between-Person\\ Correlation\\ Structure} & 
    \makecell{$corr\paren{Y_{i,j,t_l}\d, Y_{i,k,t_m}\d\mid\bX_i} $ \\ $\paren{j\neq k}$} &
     Estimator\\
    \hline
    Exchangeable & $\rho_b^d$  &  
    $\rhohat^d_{b, Ex} = \ddfrac{\Sum{i}{N}I_i(d)W_i
    \Sum{j}{n_i}\displaystyle\sum_{\substack{k=1\\k\neq j}}^{n_i}
    \Sumz{l}{T}\Sumz{m}{T}
    \paren{\frac{\hat{\varepsilon}^{d,s}_{i,j,t_{l}}}{\sigmahat^d_{t_{l}}}}
    \paren{\frac{\hat{\varepsilon}^{d,s}_{i,k,t_m}}{\sigmahat^d_{t_m}}}}
    {\Sum{i}{N}I_i(d)W_in_i\paren{n_i-1}(T+1)^2}$\\
    Unstructured & $\rho^d_{b, t_l, t_m}$ &  
    $\rhohat^d_{b, Un, t_l, t_m} = \ddfrac{\Sum{i}{N}I_i(d)W_i
    \Sum{j}{n_i}\displaystyle\sum_{\substack{k=1\\k\neq j}}^{n_i}
    \paren{\frac{\hat{\varepsilon}^{d,s}_{i,j,t_{l}}}{\sigmahat^d_{t_{l}}}}
    \paren{\frac{\hat{\varepsilon}^{d,s}_{i,k,t_m}}{\sigmahat^d_{t_m}}}}
    {\Sum{i}{N}I_i(d)W_in_i\paren{n_i-1}}$\\
    Independent & 
    0 &  
    $\rhohat^d_{b, In} = 0$\\
    \hline
    \end{tabular}
\caption{Between-Person Correlation Estimators}
    \label{c1:tab:var_est:bp_corr}
\end{table}

\end{landscape}

\section{Weight Estimation}\label{c1:sec:wgt_est}

As discussed in Section \ref{c1:sec:est} of the main body, in a SMART, the weights $W_i$ are known by design. However, they can also be estimated, which is known to offer improvements in statistical efficiency \citep{Robins2000}. This appendix concerns adjustments to the formulae for the standard errors to account for the sampling variance in the estimated weights.

For an analyst wishing to construct weights to adjust for imbalances in select baseline information, $\bX^w$, and outcome/covariate information collected between $t_0$ and $t^*$, $\bS^{w}_{1}$, weights for prototypical SMARTs take the form
\begin{align*}
    W_i & = W(\bx^w_i, a_{1,i}, \bs^{w}_{1,i}, r_i, a_{2NR,i})\\
& =\dfrac{1}{\Prb{A_{1}=a_{1,i}\mid \bX^w=\bx^w_{i}}\Prb{A_{2}=a_{2,i} \mid \bX^w=\bx^w_{i}, A_{1,i}={a_{1,i}}, \bS^{w}_{1}=\bs^{w}_{1,i}, R_i=r_i}}.
\end{align*}
This approach requires estimating 
$p_1\paren{a_1, \bx^w;\pihat} = \hat{\mathbb{P}}\left[A_{1}=a_1\mid \bX^w=\bx^{w}\right]$ 
and \\
$p_2\paren{a_1, r, a_2, \bx^w, \bs^w_1;\pihat} = \hat{\mathbb{P}}\left[A_{2}=a_{2}\mid \bX^w=\bx^{w},  A_1=a_1,\bS^w_1=\bs^w_{1}, R=r\right]$ 
for each $a_1$, $r$, and $a_{2}$, using, for example, logistic regression or any other maximum likelihood estimator suitable for binary outcomes. With slight abuse of notation, we let $S_\pi$ denote the score function of the weights. 
The estimated weights then have the form
$$W_i(\pihat) = W(\bX^w_{i}, a_{1,i}, \bs^w_{1,i}, r_i, a_{2,i}; \pihat)
=\dfrac{1}{p_1\paren{a_{1,i}, \bx^w_{i};\pihat}p_2\paren{a_{1,i}, r_i, a_{2,i}, \bx^w_{i}, \bs^w_{1,i};\pihat}}.$$

As we show in Appendix \ref{c1:sec:asy_thry}, using estimated weights induces the following asymptotic behavior of the estimator:

$$\sqrt{N}\paren{\thetahat-\theta_0}\xrightarrow[N\to\infty]{\mathcal{L}}\N{0}{J\inv Q J\inv},$$
where 
\allowdisplaybreaks
\begin{align*}
    J & = \E{\Sumd I(d) W(\pi_0)D(d,\bX;\theta)^\top \Sigma^d(\alpha_+)\inv D(d,\bX;\theta)},\\
    Q & = \E{UU^\top} - \E{US_{\pi_0}^\top}\E{S_{\pi_0}S_{\pi_0}^\top}\inv\E{S_{\pi_0}U^\top}.
\end{align*}

Comparing the above distribution with that presented in Section \ref{c1:sec:est} of the main body (which assumes the use of known weights) reveals that the only difference introduced by weight estimation is the addition of the term $- \E{US_{\pi_0}^\top}\E{S_{\pi_0}S_{\pi_0}^\top}\inv\E{S_{\pi_0}U^\top}$ in the $Q$ matrix (i.e., the \quotes{meat} component). As noted by \citet{Robins2000}, such adjustments often reduce the variance.

\section{Causal Identifiability Assumptions}\label{c1:sec:id_ass}
As discussed in Section \ref{c1:sec:not:po} of the main body, we use potential outcomes concepts to discuss causal effects. Furthermore, we adopt canonical assumptions to enable identifying causal estimands using observed SMART data \citep{Robins2000, Rubin2005}. In particular, following \citet{Robins2008-ua, Orellana2010-gv, Orellana2010-rc}, we assume positivity, consistency, and sequential ignorability. In this section, we provide the formal definitions for the assumptions specifically in the prototypical SMART setting and discuss these assumptions more generally in Appendix \ref{c1:sec:gen_smart:id_ass}.

\begin{enumerate}
    \item \label{c1:ass:id_1} \textbf{Positivity}
    \subitem $\Prb{A_1=1}\in(0,1)$ and $\Prb{A_2=1\mid A_1, R=0}\in(0,1)$. That is, all valid treatment pathways in the SMART have non-zero probability.
    \item \label{c1:ass:id_2} \textbf{Consistency}
    \subitem We consider consistency with respect to response and outcome:
    \begin{enumerate}
        \item $R_i=\Ind{A_{1,i}=1}R_i^{(A_1=1)} + \Ind{A_{1,i}=-1}R_i^{(A_1=-1)}$.
        \item $\vct{Y}_{i,.,t}^{(a_1, a_{2NR})} = \begin{cases}
        \vct{Y}_{i,.,t} & \text{ if $t=t_0$}\\
        \vct{Y}_{i,.,t} & \text{ if $a_{1,i}=a_1$ and $t\in(t_0,t^*]$}\\
        \vct{Y}_{i,.,t} & \begin{aligned}[t]
            &\text{if $a_{1,i}=a_1$, $t\in(t^*,t_T]$}\\
            &\quad\quad\text{and either $a_{2NR,i}=a_{2NR}$ or $R_i=1$.}
          \end{aligned}
        \end{cases}$
    \end{enumerate}
    \subitem As in \citet{Chakraborty2014}, we note that this assumption subsumes Rubin's traditional SUTVA assumption \citep{Rubin1980}. We also note that the latter consistency equality simply asserts that if, at time $t$, Cluster $i$'s observed treatment history is consistent with an embedded adaptive intervention $d$, then $Y_{i,.,t}=Y_{i,.,t}\d$. 
    \item \label{c1:ass:id_3} \textbf{Sequential Ignorability}
    \subitem Given \textit{any} set of baseline covariates $\bX$ (where $\bX$ could be empty), and any $(a_1,a_{2NR})\in\mathcal{D}$, we have
    \begin{enumerate}
        \item $\set{\vct{\bY}_i^{(a_1, a_{2NR})}, R_i^{(a_1)}}\indep A_{1}\big|_{\bX}$.
        \item $\vct{\bY}_i^{(a_1, a_{2NR})}\indep A_{2}\big|_{A_{1,i}=a_1, R_i^{(a_1)}, \bX}$.
    \end{enumerate}
    \subitem That is, we have both marginal sequential ignorability and sequential ignorability conditional on any set of baseline covariates.
\end{enumerate}

\section{General SMART Structures}\label{c1:sec:gen_smart}
This appendix discusses generalizations of the notation and methodology presented in the main report to general SMART structures. We restrict consideration to discrete treatment and embedded tailoring variables.

\subsection{Embedded Adaptive Interventions}
As discussed in Section \ref{c1:sec:smart:embcai} of the main body, there are adaptive interventions embedded in the design of SMARTs. The number and nature of these embedded adaptive interventions depend on the SMART structure at hand, with prototypical SMARTs containing four embedded adaptive interventions indexed by the first stage treatment and second stage treatment decision for non-responding clusters. Similarly, SMART design IV also contains four embedded adaptive interventions. While this design does not contain a notion of embedded tailoring (i.e., there is no determination of \quotes{response}), the two choices of first-stage intervention and two choices for second-stage intervention induce four embedded decision protocols. Such embedded adaptive interventions can be indexed by these choices $(a_1, a_2)$.

SMART design I contains eight embedded adaptive interventions. As in prototypical SMARTs, the presence of a binary embedded tailoring variable (\quotes{response}) induces three choices - choice of first-stage treatment, choice of second-stage treatment for responding clusters, and choice of second-stage treatment for non-responding clusters. Therefore, embedded adaptive interventions for SMART design I can be denoted as $(a_1, a_{2R}, a_{2NR})$. However, unlike prototypical SMARTs, there is more than one option for the choice of second-stage intervention for responding clusters. This additional choice induces eight embedded AIs, rather than the four in prototypical SMARTs.

\subsection{Notation}
We consider a SMART with $K$ stages of treatment. We consider \quotes{valid} levels of treatment and response to be those identified in the SMART design. For example, in ASIC, the valid first-stage treatments are \quotes{Add Coaching} and \quotes{Do Not Add Coaching.} The valid levels of response for either first-stage treatment are \quotes{Responder} and \quotes{Non-Responder.} By design, the only valid second-stage intervention option for all responders is \quotes{Continue,} whereas the valid second-stage intervention options for all non-responders are \quotes{Add Facilitation} and \quotes{Do Not Add Facilitation.}

Given this notion of \quotes{valid} response and intervention options, we adopt the formal notation laid out below. For any $k\in\set{1\ddd K}$, let $\Bar{a}_k:=\paren{a_1\ddd a_k}$ denote past treatments and let $\Bar{r}_k:=\paren{r_1\ddd r_k}$ denote response history, where $r_k$ is response status used to tailor the $k+1$-stage intervention. For any $k>1$ let $\mathcal{A}_{k|\Bar{A}_{k-1},\Bar{R}_{k-1}}$ denote the valid treatment options for a cluster at time $k$ given past interventions $\Bar{A}_{k-1}$ and response history $\Bar{R}_{k-1}$ and let $\mathcal{R}_{k|\Bar{A}_{k},\Bar{R}_{k-1}}$ denote the valid levels of response given intervention/response history $\paren{\Bar{A}_{k},\Bar{R}_{k-1}}$. Finally, we let $\paren{\Bar{\mathcal{A}},\Bar{\mathcal{R}}}_k$ denote the set of valid treatment-response pathways up to stage $k$; i.e., $\paren{\Bar{a}_{k},\Bar{r}_{k}}\in\paren{\Bar{\mathcal{A}},\Bar{\mathcal{R}}}_k$ if and only if $a_1\in\mathcal{A}_1$, $r_1\in\mathcal{R}_{1\mid a_1}$, and, for all $j=2\ddd k$, $a_{j}\in\mathcal{A}_{j\mid \Bar{a}_{j-1}, \Bar{r}_{j-1}}$ and 
$r_{j}\in\mathcal{R}_{j\mid \Bar{a}_{j}, \Bar{r}_{j-1}}$.

We also note that embedded adaptive interventions take the form $d=(f_1,\ldots,f_K)$, where, for each $k=1,\ldots,K$, $f_k$ is a deterministic function mapping the available treatment and tailoring history prior to the $k^{\text{th}}$ treatment decision to a valid $k^{\text{th}}$-stage treatment. That is,
\begin{align*}
f_k:\bar{\mathcal A}_{k-1}\times \bar{\mathcal R}_{k-1}\to \mathcal A_k,
\end{align*}
with the image restricted to valid treatment options under the corresponding history.

For a given embedded adaptive intervention $d=(f_1,\ldots,f_K)$, define the treatment history induced by $d$ recursively. Let
$\Tilde{f}_1\paren{R_0}:=f_1\paren{R_0}$ and 
$$\Tilde{f}_k\paren{\Bar{R}_{k-1}}:=\paren{\Tilde{f}_{k-1}\paren{\Bar{R}_{k-2}}, f_k\paren{\Tilde{f}_{k-1}\paren{\Bar{R}_{k-2}}, \Bar{R}_{k-1}}}\in\Bar{\mathcal{A}}_k.$$

Thus, $\Tilde{f}_k(\bar R_{k-1})$ is the treatment history through stage $k$ that would be assigned by $d$ under tailoring-variable history $\bar R_{k-1}$.

For notational brevity, consider $A_0$ to be a degenerate treatment assignment (i.e., constant for all clusters). Similarly, consider a pre-trial baseline embedded tailoring variable $R_0$. In all four SMART designs considered in Figure \ref{c1:fig:SMART_strs} of the main body, $R_0$ is constant/trivial. However, it would be possible to conduct a trial in which first-stage treatment options depended on baseline information.

\subsection{Marginal Mean Model}
The exact form of marginal mean model should reflect the randomization structure of the SMART at hand. For SMART design II, Equation~(\ref{c1:eq:mmm_ex}) in the main body gave a piecewise linear model with a knot at the second decision point, allowing for separate first and second stage treatment slopes. The equations below represent analogues of this model for SMART designs I, III, and IV, respectively.
\begin{equation} \label{c1:eq:mmm_I}
\begin{split}
    \mu^{I}_t(a_1, a_{2R}, a_{2NR}, X; \vct{\theta}) & = 
    \eta X + \gamma_0 + \Ind{t\leq t^*}\paren{\gamma_1t + \gamma_2a_1t} \\
     & \hspace{.5cm} + \Ind{t > t^*}\big(\gamma_1t^* + \gamma_2a_1t^* + \gamma_3\paren{t-t^*} + \gamma_4\paren{t-t^*}a_1 \\
     & \hspace{.5cm} + \gamma_5\paren{t-t^*}a_{2R} + \gamma_6\paren{t-t^*}a_{2NR}\\
     & \hspace{.5cm} + \gamma_7\paren{t-t^*}a_1a_{2R} + \gamma_8\paren{t-t^*}a_1a_{2NR}
     \big)
\end{split}
\end{equation}
\begin{equation} \label{c1:eq:mmm_III}
\begin{split}
    \mu^{III}_t(a_1, a_{2NR}, X; \vct{\theta}) & = 
    \eta X + \gamma_0 + \Ind{t\leq t^*}\paren{\gamma_1t + \gamma_2a_1t} \\
     & \hspace{.5cm} + \Ind{t > t^*}\big(\gamma_1t^* + \gamma_2a_1t^* + \gamma_3\paren{t-t^*} + \gamma_4\paren{t-t^*}a_1 \\
     & \hspace{.5cm} + \gamma_5\paren{t-t^*}a_{2NR}\Ind{a_1=1}\big)
\end{split}
\end{equation}\begin{equation} \label{c1:eq:mmm_IV}
\begin{split}
    \mu^{IV}_t(a_1, a_{2}, X; \vct{\theta}) & = 
    \eta X + \gamma_0 + \Ind{t\leq t^*}\paren{\gamma_1t + \gamma_2a_1t} \\
     & \hspace{.5cm} + \Ind{t > t^*}\big(\gamma_1t^* + \gamma_2a_1t^* + \gamma_3\paren{t-t^*} + \gamma_4\paren{t-t^*}a_1 \\
     & \hspace{.5cm} + \gamma_5\paren{t-t^*}a_{2} + \gamma_6\paren{t-t^*}a_1a_{2}\big)
\end{split}
\end{equation}

While the SMART structure in question may heavily influence the construction of the marginal mean model, this is not necessarily the case for the working variance model, which may be more heavily influenced by the outcome modeled and nature of clustering.

\subsection{Estimation}\label{c1:sec:gen_smart:est}
As discussed in Section \ref{c1:sec:est} of the main body, the analyst must employ weights to avoid over-representation of certain clusters in their estimating equation. The exact form of the weights depends on the SMART structure at hand. Generally, the weights take the form:
$$W_i=W\paren{\Bar{A}_{K,i},\Bar{R}_{K-1,i}} = 
\dfrac{1}{
\displaystyle\prod_{k=1}^K
\Prb{A_{k}=A_{k,i}\mid \Bar{A}_{k-1,i}, \Bar{R}_{k-1,i}}}.$$
As in the prototypical setting, the analyst can either use the known randomization probabilities to calculate these weights or estimate the probabilities themselves to improve precision. For weight estimation in the general setting, the weights take the form

$$W_i =
W(\Bar{a}_{K, i}, \Bar{r}_{K-1, i}, \bh_{K,i})
=\dfrac{1}{\displaystyle\prod_{k=1}^K p_{k;\pihat}\paren{a_{k,i}; r_{k-1,i}, \bh_{k,i}}},$$
where $\bh_{k,i}$ represents information on Cluster $i$ available up to the $k^{th}$ decision point, to be used for weight estimation. Note that $\Bar{a}_{k-1,i}\in\bh_{k,i}$ and $\bh_{1,i}\subset \bX_i$.

In this setting, we let 
$$p_{\pi}\paren{\Bar{a}_K; \Bar{r}_{K-1}, \bh_{K}}:= 
\prod_{k=1}^K p_{k;\pi}\paren{a_{k}; r_{k-1}, \bh_{k}}.$$

Furthermore, we consider the indicator function
\begin{equation*}
    I_i(d) = 
    \begin{cases}
    1 & \text{ if }
    \Bar{A}_{K,i} = \Tilde{f}_K\paren{\Bar{R}_{K-1, i}}\\
    0 & \text{ otherwise}
    \end{cases},
\end{equation*}
where $d\in\mathcal{D}$ denotes the embedded adaptive intervention $(f_1,\ldots,f_K)$.

After constructing a marginal mean and working variance model, as well as calculating weights, the analyst can follow the same estimation procedure laid out in Section \ref{c1:sec:est} of the main body for estimation/inference of marginal parameters in the general SMART setting.

\subsection{Causal Identifiability Assumptions}\label{c1:sec:gen_smart:id_ass}
Appendix \ref{c1:sec:id_ass} presented the assumptions necessary to identify causal effects in the prototypical SMART setting. Here, we present these assumptions in the general SMART setting.
\begin{enumerate}
    \item \label{c1:ass_gen:id_1} \textbf{Positivity}
    \subitem 
    For any $k=1\ddd K$, $\Prb{A_k=a_k\mid \Bar{A}_{k-1}=\Bar{a}_{k-1}, \Bar{R}_{k-1}=\Bar{r}_{k-1}}\in(0,1]$ for any $a_k\in\mathcal{A}_{k\mid \Bar{a}_{k-1}, \Bar{r}_{k-1}}$ and $\paren{\Bar{a}_{k-1}, \Bar{r}_{k-1}}\in\paren{\Bar{\mathcal{A}}, \Bar{\mathcal{R}}}_{k-1}.$
    \item \label{c1:ass_gen:id_2} \textbf{Consistency}
    \subitem We consider consistency with respect to response and outcome:
    \begin{enumerate}
        \item For any $k=1\ddd K-1$, 
        \begin{align*}
            R_{k,i}  = & 
        \displaystyle\sum_{\paren{\Bar{a}_{k-1}, \Bar{r}_{k-1}}\in\paren{\Bar{\mathcal{A}}, \Bar{\mathcal{R}}}_{k-1}}
        \Ind{\Bar{A}_{k-1,i}=\Bar{a}_{k-1}, \Bar{R}_{k-1,i}=\Bar{r}_{k-1}}\\
        & \hspace{3cm}\times \paren{
        \displaystyle\sum_{a_k\in\mathcal{A}_{k\mid \Bar{a}_{k-1}, \Bar{r}_{k-1}}} \Ind{A_{k,i}=a_k}R_{k,i}^{\paren{\Bar{A}_{k,i}=\Bar{a}_{k}, \Bar{R}_{k-1,i}=\Bar{r}_{k-1}}}}.
        \end{align*}
        \item Let $d$ denote the embedded adaptive intervention $\paren{f_1\ddd f_K}$. For any $k=1\ddd K$, if $t\in(\Tilde{t}_{k}, \Tilde{t}_{k+1}]$ and $\Bar{A}_{k,i}=\Tilde{f}_k\paren{\Bar{R}_{k-1, i}}$, then $\vct{Y}_{i,.,t}\d=\vct{Y}_{i,.,t}$.\footnote{Where $\Tilde{t}_{K+1}$ is defined as the end of follow-up measurement.}
    \end{enumerate}
    \item \label{c1:ass_gen:id_3} \textbf{Sequential Ignorability}
    \subitem Given \textit{any} set of baseline covariates $\bX$ (where $\bX$ could be empty), we have, for any $k=1\ddd K-1$,
    \begin{enumerate}
        \item $\set{\vct{\bY}_i^{(d)}, 
        R_{k, i}^{\paren{\Bar{a}_k, \Bar{r}_{k-1}}} \ddd 
        R_{K-1, i}^{\paren{\Bar{a}_{K-1}, \Bar{r}_{K-2}}}}
        \indep 
        A_{k}\big|_{\Bar{A}_{k-1}=\Tilde{f}_{k-1}\paren{\Bar{R}_{k-2}}, \Bar{R}_{k-1}, \bX}$,
        \item $\vct{\bY}_i^{(d)}\indep A_{K}\big|_{\Bar{A}_{K-1}=\Tilde{f}_{K-1}\paren{\Bar{R}_{K-2}}, \Bar{R}_{K-1}, \bX}$.
    \end{enumerate}
\end{enumerate}

\section{Data-Generative Models for Simulation Experiments}\label{c1:sec:sim_des}
In designing the simulation study for this report, we sought to simulate outcomes sequentially, believing this to be more realistic than other data-generative approaches. Simulating outcomes conditional on past outcomes (and response status) necessitated specifying\footnote{Implicitly so, via the choice of several conditional parameters, as discussed below.} the conditional distribution of the outcomes. Given that our proposed approach fundamentally concerns marginal parameter estimation, we sought to construct a data-generative mechanism consistent with the marginal structural model discussed in Section \ref{c1:sec:mm} of the main body. In an attempt to limit the complexity of the simulation model, we limited our simulations to the case of three time points ($t=0,1,2$), where response is determined at time $t^*=1$. 

Unlike the main body, in this appendix we use $\YY{0}$, $\YY{1}$, and $\YY{2}$ to denote the vector of outcomes for a given cluster at times $t=0,1,2$, respectively. E.g., $Y_{i,t}=\begin{bmatrix}
    Y_{i,1,t}, & \dots & Y_{i,n_i,t}
\end{bmatrix}^\top$. We use this notation for improved clarity, given the conditional data-generative model described in this appendix.

Implementation code for this data-generative model can be found at \url{https://github.com/GabrielDurham/three_lvl_cSMART_sim_study}.

\subsection{Target Structural Mean Model}\label{c1:sec:sim_des:tgt_smm}
As discussed above, we sought a conditional data-generative framework; however, as discussed in Section~\ref{c1:sec:mm}, we propose a marginal modeling framework. Consequently, studying the properties of our estimator requires proper control over the marginal distribution of the generated data.

In the section below, we seek to generate data that mimics the marginal mean model for a prototypical SMART proposed in Section \ref{c1:sec:mm}:
\begin{align*}
    \E{\vct{Y}_{i,t}^{(a_1, a_{2NR})}\mid \bX_i} & = \eta \bX_i + \gamma^M_0\oni 
    + \Ind{t\leq 1}\paren{\gamma^M_1t + \gamma^M_2a_1t}\oni \\
    & \hspace{1cm}+ \Ind{t> 1}\Big(\gamma^M_1 + \gamma^M_2a_1 + \gamma^M_3\paren{t-1} + \gamma^M_4\paren{t-1}a_1 \\
    & \hspace{3cm} + \gamma^M_5\paren{t-1}a_{2NR} + \gamma^M_6\paren{t-1}a_1a_{2NR}
    \Big)\oni,
\end{align*}
where we employ the superscript $M$ to denote marginal parameters. We note that the model above is fully saturated in that it allows the specification of all potential expectations: $\E{\vct{Y}^{(.,.)}_{i,0}}$, $\E{\vct{Y}^{(1,.)}_{i,1}}$, $\E{\vct{Y}^{(-1,.)}_{i,1}}$, $\E{\vct{Y}^{(1,1)}_{i,2}}$, $\E{\vct{Y}^{(1,-1)}_{i,2}}$, $\E{\vct{Y}^{(-1,1)}_{i,2}}$, and $\E{\vct{Y}^{(-1,-1)}_{i,2}}$. For brevity, we use $\mu_{t,{d}}$ to denote $\E{\paren{\vct{Y}^{(d)}_{i,t}}_j}$. That is, choice of the aforementioned values induces $\gamma^M$ parameters via:

$$
\begin{bmatrix}
    \gamma_0^M \\ \gamma_1^M \\ \gamma_2^M \\ \gamma_3^M \\ \gamma_4^M \\ \gamma_5^M \\ \gamma_6^M
\end{bmatrix} = 
\begin{bmatrix}
    1 & 0 & 0 & 0 & 0 & 0 & 0 \\
    1 & 1 & 1 & 0 & 0 & 0 & 0 \\
    1 & 1 & -1 & 0 & 0 & 0 & 0 \\
    1 & 1 & 1 & 1 & 1 & 1 & 1 \\
    1 & 1 & 1 & 1 & 1 & -1 & -1 \\
    1 & 1 & -1 & 1 & -1 & 1 & -1 \\
    1 & 1 & -1 & 1 & -1 & -1 & 1 \\
\end{bmatrix}\inv
\begin{bmatrix}
    \mu_{0, {(.,.)}} \\ 
    \mu_{1, {(1,.)}} \\ 
    \mu_{1, {(-1,.)}} \\ 
    \mu_{2, {(1,1)}} \\ 
    \mu_{2, {(1,-1)}} \\ 
    \mu_{2, {(-1,1)}} \\ 
    \mu_{2, {(-1,-1)}} \\ 
\end{bmatrix}.
$$

Furthermore, we wish to control the corresponding variance $\V{\vct{Y}_{i,t}^{(a_1, a_{2NR})}\mid \bX_i}$. We wish to generate a variance structure that is exchangeable in individuals and unstructured in time. That is,
\begin{itemize}
    \item $\V{\YY{i,t}\d\mid \bX=\bX_i} = \paren{\sigma^2_{t, d} - \rho_{t, d}} I_{n_i} + \rho_{t, d}\oni\oni^\top$ and
    \item $\Cov{\YY{i,t_1}\d, \YY{i,t_2}\d\mid \bX=\bX_i} = \paren{\phi_{t_1,t_2, d} - \rho_{t_1,t_2, d}} I_{n_i} + \rho_{t_1,t_2, d}\oni\oni^\top$ (for $t_1\neq t_2$).
\end{itemize}

Specifying such a variance structure requires selecting the following values for distributions of $Y^{(d)}$ (conditioned on $\bX$). 
\begin{align*}
    \sigma^2_{0, d} = \sigma_0^2 & - \text{The variance of the outcome at baseline for an individual.}\\
    \sigma^2_{1, d} = \sigma^2_{1,a_1} & - \text{The variance of the outcome at time $t=1$ under $A_1=a_1$}\\
        &\hspace{1cm} \text{for an individual.}\\
    \sigma^2_{2, d} = \sigma^2_{2,a_1,a_{2NR}} & - \text{The variance of the outcome at time $t=2$ under $A_1=a_1$ and}\\
        &\hspace{1cm} \text{ $A_{2NR}=a_{2NR}$ for an individual.}\\
    \rho_{0, d} = \rho_0 & - \text{The intra-cluster covariance of the outcome at baseline.}\\
    \rho_{1, d} = \rho_{1,a_1} & - \text{The intra-cluster covariance of the outcome at time $t=1$ under}\\
        &\hspace{1cm} \text{ $A_1=a_1$.}\\
    \rho_{2, d} = \rho_{2,a_1,a_{2NR}} & - \text{The intra-cluster covariance of the outcome at time $t=2$ under}\\
        &\hspace{1cm} \text{ $A_1=a_1$ and $A_{2NR}=a_{2NR}$.}\\
    \phi_{0, 1, d} = \phi_{0,1,a_1} & - \text{The intra-person covariance of the outcome, under $A_1=a_1$,}\\
        &\hspace{1cm} \text{ between times $t_1=0$ and $t_2=1$.}\\
    \phi_{t_1, 2, d} = \phi_{t_1,2,a_1,a_{2NR}} & - \text{The intra-person covariance of the outcome, under $A_1=a_1$ and}\\
        & \hspace{1cm} \text{$A_{2NR}=a_{2NR}$, between times $t_1\in\set{0,1}$ and $t_2=2$.}\\
    \rho_{0, 1, d} = \rho_{0,1,a_1} & - \text{The inter-person covariance of the outcome, under $A_1=a_1$,}\\
        & \hspace{1cm} \text{ between times $t_1=0$ and $t_2=1$.}\\
    \rho_{t_1, 2, d} = \rho_{t_1,2,a_1,a_{2NR}} & - \text{The inter-person covariance of the outcome, under $A_1=a_1$ and}\\
    & \hspace{1cm} \text{$A_{2NR}=a_{2NR}$, between times $t_1\in\set{0,1}$ and $t_2=2$.}
\end{align*}

Note that we require pooling across DTR for certain variance parameters, as implied by the notation above. E.g., $\sigma^2_{0, d}=\sigma_0^2$ implies a shared baseline variance of the outcome across all embedded DTRs (e.g., $\sigma^2_{0, d}=\sigma^2_{0, d'}$ for all pairs $d,d'$). This pooling is merely a consequence of the causal identifiability assumptions discussed in Appendix \ref{c1:sec:id_ass}.

Furthermore, for a distribution parameter $\alpha$, we use the notation $\alpha\sr$ to denote the corresponding variance, conditioned on $R=r$. E.g., ${\sigma\sRt_{2, d}}$ represents the variance of the outcome $Y$ for an individual (conditioned on $\bX$) at time $2$, given the unit was in a responding cluster. Similarly, ${\sigma\sNRt_{2, d}}$ represents this variance conditioned on the unit belonging to a non-responding cluster. We also use $\mu_{t,d}\sr$ to denote conditional mean parameters.

\subsection{Data-Generative Model}\label{c1:sec:sim_des:dgm}
In the subsections below, we describe a data-generative model designed to produce data under a specified marginal distribution. Here, we focus on simulating data for a prototypical SMART; however, this model can be slightly modified to produce data under alternate SMART structures.

\subsubsection{Arguments to Select}
In designing a data-generative mechanism for a prototypical SMART consistent with the marginal structural model discussed above, we wanted to retain control over the following:
\begin{enumerate}
    \item\label{c1:ArgSel_MSM} A target marginal mean and variance structure for the outcome $Y\d$ for all $d\in\mathcal{D}$, as well as corresponding mean/variance structures conditioned on response status (as discussed above). This requires specifying:
    \begin{enumerate}
        \item Pre-response marginal mean and variance structures. That is, $\mu_{0}$, $\mu_{1,a_1}$, $\sigma^2_0$, $\sigma^2_{1,a_1}$, $\rho_0$, $\rho_{1,a_1}$, $\phi_{0,1,a_1}$, and $\rho_{0,1,a_1}$ for $a_1=\pm 1$. As discussed below, these choices will induce corresponding parameters conditional on response (which we derive via Monte-Carlo, as in \citet{Seewald2019}).
        \item Post-response marginal mean and variance structures. That is, $\mu_{2,d}$, $\sigma^2_{2,d}$, $\rho_{2,d}$, $\phi_{0,2,d}$, $\phi_{1,2,d}$, $\rho_{0,2,d}$, and $\rho_{1,2,d}$ for $d\in\mathcal{D}$.
        \item Post-response conditional mean and variance structures. That is, $\mu\sr_{2,d}$, $\sigma\srt_{2,d}$, $\rho\sr_{2,d}$, $\phi\sr_{0,2,d}$, $\phi\sr_{1,2,d}$, $\rho\sr_{0,2,d}$, and $\rho\sr_{1,2,d}$ for $d\in\mathcal{D}$ and $r\in\set{0,1}$. As discussed below, these choices must be consistent with the choice of marginal post-response variance.
    \end{enumerate}
    \item\label{c1:ArgSel_p_R} The probability of response under each first-stage treatment assignment. That is, $\Prb{R(a_1)=1\mid \bX}=:p_{r,a_1}$ for $a_1=\pm1$.
    \item\label{c1:ArgSel_p_RA} The treatment assignment probabilities (e.g., $\Prb{A_1=1}$ and $\Prb{A_{2NR}=1\big|A_1}$).\footnote{These are usually taken to be 0.5; however, an analyst may want to simulate a SMART under imbalanced randomization probabilities.}
    \item\label{c1:ArgSel_X_Dist} The distribution of covariates (provided the distribution has mean zero) and their influence on the outcome ($\eta$). 
    \item\label{c1:ArgSel_K} The number of clusters, $K$.\footnote{For notational convenience in this simulation section, which pertains to simulating data from a prototypical two-stage SMART, $K$ denotes the number of clusters. This use is local to the simulation design and is distinct from the use of $K$ elsewhere for the number of stages in a general SMART.}
    \item\label{c1:ArgSel_NiDist} The distribution of $N_i$, local sample size (provided the distribution is over $\mathbb{N}$ and independent of outcomes).
\end{enumerate}

Therefore, the analyst must specify the above quantities, ideally to approximate a real-world scenario relevant to the analysis in question. As discussed below, conditional variance components must be chosen to be consistent with the selected marginal variance components. Additionally, we require that variance parameters are chosen such that all resulting covariance matrices are positive definite.

\subsubsection{Marginal and Conditional Specification Consistency}\label{c1:sec:sim_des:dgm:MC_consistent}
As discussed above, we wish to simulate data conditionally (i.e., conditioned on response and past outcomes). Doing so more closely resembles real-world processes, and allows the user to specify stark differences between responders and non-responders in order to test estimator performance. 

In Step \ref{c1:ArgSel_MSM} above, the user must select the marginal and conditional distribution of $\vct{Y}\d$ for all $d$, as well as the probability of response under $A_1=1$ and $A_1=-1$. However, for a given DTR $d$ (with probability of response $p_r$), we note that, for any $j=1\ddd N_i$ and $t_1,t_2\in\set{0,1,2}$,
$$\E{\paren{\vct{Y}_{i,t}\d}_j\mid \bX_i}  = p_r\E{\paren{\vct{Y}_{i,t}\d}_j\mid \bX_i, R_i=1} + (1-p_r)\E{\paren{\vct{Y}_{i,t}\d}_j\mid \bX_i, R_i=0}$$
and $\Cov{\YY[j]{i,t_1}\d,\YY[k]{i,t_2}\d\mid \bX_i}$
\begin{align*}
     \lalign{= \E{\Cov{\paren{\vct{Y}_{i,t_1}\d}_j,\paren{\vct{Y}_{i,t_2}\d}_k\mid \bX_i, R_i}}} \\
    & \ralign{
    + \Cov{\E{\paren{\vct{Y}_{i,t_1}\d}_j\big|R_i=r},\E{\paren{\vct{Y}_{i,t_2}\d}_k\mid \bX_i, R_i}}}\\
    & = \E{\Ind{R_i=1}\Cov{\paren{\vct{Y}_{i,t_1}\d}_j,\paren{\vct{Y}_{i,t_2}\d}_k\mid \bX_i, R_i=1}} \\
    & \ralign{
    + \E{\Ind{R_i=0}\Cov{\paren{\vct{Y}_{i,t_1}\d}_j,\paren{\vct{Y}_{i,t_2}\d}_k\mid \bX_i, R_i=0}}}\\
    & \ralign{
    + \Cov{\Ind{R=0}\mu_{t_1, d}\sNR + \paren{1-\Ind{R=0}}\mu_{t_1, d}\sR,\s \Ind{R=0}\mu_{t_2,d}\sNR + \paren{1-\Ind{R=0}}\mu_{t_2,d}\sR}}\\
    & = p_r\Cov{\paren{\vct{Y}_{i,t_1}\d}_j,
    \paren{\vct{Y}_{i,t_2}\d}_k\mid \bX_i, R_i=1} \\
    & \ralign{
    + \paren{1-p_r}\Cov{\paren{\vct{Y}_{i,t_1}\d}_j,
    \paren{\vct{Y}_{i,t_2}\d}_k \mid \bX_i, R_i=0}}\\
    & \ralign{
    + \paren{\mu_{t_1, d}\sNR-\mu_{t_1,d}\sR}\paren{\mu_{t_2, d}\sNR-\mu_{t_2,d}\sR}\paren{p_r\paren{1-p_r}}.}
\end{align*}

Subsequently, the distributions of $\vct{Y}\d\big|_{\bX}$, $\vct{Y}\d\big|_{\bX, R=1}$, and $\vct{Y}\d\big|_{\bX, R=0}$ cannot vary freely, as choosing two implies the third. We specified the marginal distribution as well as the distribution conditional on positive response.

Lastly, as discussed below, specifying a marginal pre-response distribution (i.e., the distribution of $\ttvec{Y_0\d}{Y_1\d}\big|_\bX$) and a response-generation mechanism induces the conditional distribution of $\ttvec{Y_0\d}{Y_1\d}\big|_{\bX, R}$. Analytically deriving this distribution for non-trivial models of response proved intractable, and subsequently we used Monte-Carlo methods to estimate necessary conditional mean and variance parameters, as in \citet{Seewald2019}.\footnote{We did so by simulating $500,000$ simulations for each pre-response variance structure/cluster size pair and taking the empirical response-conditional distribution among the simulations.}

\subsubsection{Data-Generative Process}\label{c1:sec:sim_des:dgm:dgp}
Given the chosen parameters discussed above, we propose the algorithm to produce data which targets the desired marginal structural model. We note the choices above require specification of all $\mu_.$, $\mu_.\sr$, $\sigma_.$, $\sigma_.\sr$, $\rho_.$, $\rho_.\sr$, $\phi_.,$ and $ \phi_.\sr$ terms, in addition to all response probabilities and treatment assignment probabilities. Choice of these terms defines the simulation parameters used below (as we make explicit later in this section). In the model described below, we use $[r]$ notation to denote simulation parameters that depend on response status (and note that these only apply to post-response parameters). Recall the use of $<r>$ notation to denote parameters regarding the distribution of $\vct{Y}\d\big|_{R=r}$.

We carry out the algorithm below for $i=1\ddd K$ (taking independent draws for each cluster): 

\begin{algorithm}[H]
\caption{\texttt{Data-Generative Algorithm}}
\begin{algorithmic}[1]
\State Simulate $d_i=\paren{A_{1,i},A_{2NR,i}}$ using the treatment assignment probabilities specified in Argument \ref{c1:ArgSel_p_RA}. We use $\paren{a_{1,i},a_{2NR,i}}$ to denote simulated treatment assignments. 
\State Simulate $N_i=n_i$, drawing it from the distribution specified in Argument \ref{c1:ArgSel_NiDist}. 
\State Simulate $\bX_i$, drawing it from the distribution specified in Argument \ref{c1:ArgSel_X_Dist}.
\State Simulate $\ttvec{\vct{\varepsilon}_{i,0}}{\vct{\varepsilon}^{(a_{1,i})}_{i,1}}  \sim
\N{\ttvec{\vct{0}_{n_i}}{\vct{0}_{n_i}}}{
\left[\begin{array}{@{}c|c@{}}
  \Sigma_{0}
  & \paren{\phi_{0,1} - P_{1,a_{1,i}}\sigma_0^2}I_{n_i} \\
\hline
  \paren{\phi_{0,1} - P_{1,a_{1,i}}\sigma_0^2}I_{n_i} &
  \Sigma_{1,a_{1,i}, n_i}
\end{array}\right]
}$.
\State Take $\YY{i,0}\di\big|_{\bX_i} = \eta \bX_i + \gamma_0\oni + \vct{\varepsilon}_{i,0}$.
\State Take 
$$\YY{i,1}\di\big|_{\bX_i, \YY{i,0}\di} = \paren{1-P_{1,a_{1,i}}}\paren{\eta \bX_i + \gamma_0\oni} + \gamma_1\oni + P_{1,a_{1,i}}\YY{i,0}\di + \gamma_2a_{1,i}\oni + \vct{\varepsilon}^{(a_{1,i})}_{i,1}.$$
\State Simulate $R_i^{(a_{1,i})}$ as a draw from a $B\paren{g\paren{\ttvec{\YY{i,0}\di}{\YY{i,1}\di}}}$ distribution, taking $g$ to be the response propensity mechanism discussed below. We denote this realized value of response as $r_i$.
\State Simulate $\vct{\varepsilon}_{i,2}^{(d_i,r_i)} \sim \N{\vct{0}_{n_i}}{\Sigma^{[r_i]}_{2,d_i,n_i}}$, with $\vct{\varepsilon}_{i,2}^{(d_i, r_i)}\indep\ttvec{\ee{i,0}}{\ee{i,1}^{(a_{1,i})}}$.
\State Take 
\begin{align*}
\YY{i,2}\di\big|_{\bX_i,\YY{i,0}\di,\YY{i,1}\di,R_i^{(a_{1,i})}=r_i} & = \paren{1-\px[r_i]{2,d_i,n_i}-\px[r_i]{4,d_i,n_i}}\paren{\eta \bX_i + \gamma_0\oni}\\
    & \s\s\s\s\s + \paren{1-\px[r_i]{4,d_i,n_i}}\paren{\gamma_1 + \gamma_2 a_{1,i}}\oni + \paren{\gamma_3 + \gamma_4a_{1,i}}\oni\\
    &  \s\s\s\s\s + \px[r_i]{2,d_i,n_i}\YY{i,0}  + \px[r_i]{4,d_i,n_i}\YY{i,1} \\
    &  \s\s\s\s\s + \paren{\px[r_i]{3,d_i,n_i}\frac{1}{n_i}\Sum{j}{n_i}\ee[j]{i,0}}\oni \\
    &  \s\s\s\s\s + \paren{\px[r_i]{5,d_i,n_i}\frac{1}{n_i}\Sum{j}{n_i}\ee[j]{i,1}}\oni\\
    &  \s\s\s\s\s + \paren{\paren{\gamma_5 + \gamma_6a_{1,i}}a_{2NR, i}}\paren{\dfrac{1-r_i}{1-p_{r_{a_{1,i}}}}}\oni\\
    &  \s\s\s\s\s + \paren{\lambda_1 + \lambda_2a_{1,i}}\paren{\paren{r_i-p_{r_{a_{1,i}}}}\oni}\\
    &  \s\s\s\s\s + \vct{\varepsilon}_{i,2}^{(d_i, r_i)}.
\end{align*}
\end{algorithmic}
\end{algorithm}

We consider the following model of response:
\begin{align*}
    g\paren{\tttvec{\bX_i}{\YY{i,0}}{\YY{i,1}}} & = F_{p_{r_{a_1}}}^{\leftarrow}\paren{\mathbf{\Phi}\paren{\dfrac{\paren{\frac{1}{n}\Sum{j}{n}\YY[j]{i,1}-
    \paren{\eta \bX_i + \gamma_0 + \gamma_1 + \gamma_2 a_1}}}{\sqrt{\frac{1}{n}\paren{\sigma^2_1 + \paren{n-1}\rho_{1,a_1}}}}}}
\end{align*}
where $F_{p_{r_{a_1}}}^{\leftarrow}$ is the inverse CDF of the Beta$\paren{\dfrac{p_{r_{a_1}}}{1-p_{r_{a_1}}},1}$ distribution. Note that we force response to be uncorrelated with baseline covariates. 

We note that, by construction of $g$, response propensity is solely a function of $\ttvec{\vct{\varepsilon}_{0}}{\vct{\varepsilon}^{(a_{1})}_{1}}$. Therefore, the distribution of $\ttvec{\vct{Y}_0\d}{\vct{Y}_1\d}\big|_R$ can be directly derived via the distribution of $\ttvec{\vct{\varepsilon}_{0}}{\vct{\varepsilon}^{(a_{1})}_{1}}\big|_R$. As discussed above, we used Monte-Carlo techniques to derive this distribution under various settings. Consequently, we use the following notation:
\begin{align*}
    \epsilon\sr_{0,n} & := \E{\paren{\vct{\varepsilon}_{0}}_j\mid R=r, N=n} \\
    \epsilon\sr_{1,n} & := \E{\paren{\vct{\varepsilon}_{1}}_j\mid R=r, N=n} \\
    s\srt_{0,n} & := \V{\paren{\vct{\varepsilon}_{0}}_j\mid R=r, N=n} \\
    s\srt_{1,n} & := \V{\paren{\vct{\varepsilon}_{1}}_j\mid R=r, N=n} \\
    c\sr_{0,n} & := \Cov{\paren{\vct{\varepsilon}_{0}}_j, \paren{\vct{\varepsilon}_{0}}_k \mid R=r, N=n} \\
    c\sr_{1,n} & := \Cov{\paren{\vct{\varepsilon}_{1}}_j, \paren{\vct{\varepsilon}_{1}}_k \mid R=r, N=n} \\
    s\sr_{0,1,n} & := \Cov{\paren{\vct{\varepsilon}_{0}}_j, \paren{\vct{\varepsilon}_{1}}_j \mid R=r, N=n} \\
    c\sr_{0,1,n} & := \Cov{\paren{\vct{\varepsilon}_{0}}_j, \paren{\vct{\varepsilon}_{1}}_k \mid R=r, N=n},
\end{align*}
where $j,k=1\ddd n$ and $j\neq k$. We note that the exchangeability of the distributions of $\vct{\varepsilon}_0$ and $\vct{\varepsilon}_1^{(a_{1})}$ ensure well-posedness of the variables described above (i.e., invariance with respect to choice of $j,k$).

Here, we take 
\begin{align*}
    P_{1,a_1} & := \begin{cases}
        \dfrac{\rho_{0,1,a_1}}{\rho_0} & \text { if } \rho_0\neq 0\\
        0 & \text { if } \rho_0 = 0,
            \end{cases}\\
    \Sigma_{0} &:= \paren{\sigma_0^2 - \rho_0} I_{n} + \rho_0\oni\oni^\top,\\
    \Sigma_{1,a_1, n} & := \paren{\paren{\sigma^2_{1,a_1} - \rho_{1,a_1}} I_{n} + \rho_{1,a_1}\oni\oni^\top} - P_{1,a_1}^2\Sigma_0 - 2P_{1,a_1}\paren{\phi_{0,1,a_1} - P_{1,a_1}\sigma_0^2}I_{n},\\
    \Sigma_{2,d,n}^{[r]} &:= \paren{\paren{\sigma\srt_{2,d}-\zeta_{d,n}} - 
    \paren{\rho\sr_{2, d}-\zeta'_{d,n}}}I_{n} + 
    \paren{\rho\sr_{2, d}-\zeta'_{d,n}}\oni\oni^\top,
    \shortintertext{and}
    \ttttvec{\px[r]{2,d,n}}{\px[r]{3,d,n}}{\px[r]{4,d,n}}{\px[r]{5,d,n}} & := \Upsilon\inv_{a_1,n}\ttttvec{\phi\sr_{0,2,d}}{\phi\sr_{1,2,d}}{\rho\sr_{0,2,d}}{\rho\sr_{1,2, d}},
\end{align*}
where $\zeta_{d,n},\zeta'_{d,n},$ and $\Upsilon_{a_1,n}$ are defined in Appendix \ref{c1:sec:sim_des:supp_defs} below. Note that the construction of $\Sigma_0$ ensures that all distributions of $\YY{0}\d$ are identical for all $d\in\mathcal{D}$.

We also note that the $\gamma$ terms above are related, but not identical, to the marginal $\gamma^M$ terms in the target structural mean model. We let
\begin{align*}
\lalign{\vct{\boldsymbol\gamma}  :=
\begin{bmatrix}
    \gamma_0 \\ \gamma_1 \\ \gamma_2 \\ \gamma_3 \\ \gamma_4 \\ \gamma_5 \\ \gamma_6 \\ \lambda_1 \\ \lambda_2
\end{bmatrix} = 
\begin{bmatrix}
    1 & 0 & 0 & 0 & 0 & 0 & 0 & 0 & 0 \\
    1 & 1 & 1 & 0 & 0 & 0 & 0 & 0 & 0  \\
    1 & 1 & -1 & 0 & 0 & 0 & 0 & 0 & 0  \\
    1 & 1 & 1 & 1 & 1 & 0 & 0 & (1-p_{r,1}) & (1-p_{r,1}) \\
    1 & 1 & 1 & 1 & 1 & \frac{1}{1-p_{r,1}} & \frac{1}{1-p_{r,1}} & -p_{r,1} & -p_{r,1}  \\
    1 & 1 & 1 & 1 & 1 & -\frac{1}{1-p_{r,1}} & -\frac{1}{1-p_{r,1}} & -p_{r,1} & -p_{r,1}  \\
    1 & 1 & -1 & 1 & -1 & 0 & 0 & (1-p_{r,-1}) & -(1-p_{r,-1}) \\
    1 & 1 & -1 & 1 & -1 & \frac{1}{1-p_{r,-1}} & -\frac{1}{1-p_{r,-1}} & -p_{r,-1} & p_{r,-1}  \\
    1 & 1 & -1 & 1 & -1 & -\frac{1}{1-p_{r,-1}} & \frac{1}{1-p_{r,-1}} & -p_{r,-1} & p_{r,-1}  \\
\end{bmatrix}\inv}[0cm]\\
& \ralign{\times
\begin{bmatrix}
    \mu_{0, {(.,.)}} \\ 
    \mu_{1, {(1,.)}} \\ 
    \mu_{1, {(-1,.)}} \\ 
    \mu^{<1>}_{2, {(1,1)}} - v_{(1,1), 1, n}\\ 
    \mu^{<0>}_{2, {(1,1)}} - v_{(1,1), 0, n}\\ 
    \mu^{<0>}_{2, {(1,-1)}} - v_{(1,-1), 0, n}\\ 
    \mu^{<1>}_{2, {(-1,1)}} - v_{(-1,1), 1, n}\\ 
    \mu^{<0>}_{2, {(-1,1)}} - v_{(-1,1), 0, n}\\ 
    \mu^{<0>}_{2, {(-1,-1)}} - v_{(-1,-1), 0, n}\\ 
\end{bmatrix},}
\end{align*}
where $v_{d, r, n}:=\px[r]{2,d,n}\epsilonx{0,n} + \px[r]{3,d,n}\epsilonx{0,n} + \px[r]{4,d,n}\paren{P_{1,a_1}\epsilonx{0,n}+\epsilonx{1,n}} + \px[r]{5,d,n}\epsilonx{1,n}.$ As shown below, the $\gamma$ parameters are constructed to ensure that responders and non-responders have their specified conditional means, with the $v$ terms representing corrections to adjust for the fact that responders and non-responders will have different pre-response means.

We briefly note our introduction of cross-temporal cluster spillover into $\YY{i,2}\di$ via the inclusion of the $\px[r_i]{3,d,n}\frac{1}{n}\Sum{j}{n}\ee[j]{i,0}$ and $\px[r_i]{5,d,n}\frac{1}{n}\Sum{j}{n}\ee[j]{i,1}$ terms. We do so in order to control cross-temporal intra-cluster covariances. Simulations under which $\px[r_i]{3,d,n}=0$ and/or $\px[r_i]{5,d,n}=0$ correspond to situations in which there are no direct cross-temporal intra-cluster spillover effects.

\subsection{Marginal Consistency}\label{c1:sec:sim_des:proof}
We will consider the simulation of the outcome trajectory of a cluster of size $N=\sss$, with covariate data $\bX$, under a given DTR $d=\paren{A_1=a_1,A_{2NR}=a_{2NR}}$, where $N$, $\bX$, and $\paren{A_1,A_{2NR}}$ are simulated according to the data-generative algorithm laid out above. In this subsection, we prove that the distribution of such an outcome trajectory has the distribution specified above. This subsection will be laid out as follows:
\begin{itemize}
    \item Remarks on notational conventions henceforth adopted to promote conciseness.
    \item Proof that the pre-response marginal distribution of the data-generative model matches that of the target distribution.
    \item Proof that the response probability of the data-generative model matches that of the target model.
    \item Proof that the post-response conditional distributions of the data-generative model match those of the target distribution.
\end{itemize}
We note that the post-response marginal distribution of the data-generative model is implied by the response probability and post-response conditional distributions of the data-generative model, as discussed above. Therefore, proving the three characteristics listed above will guarantee that the marginal distribution of the data-generative model matches the specified target distribution.

\subsubsection{Notational Remarks}
The following arguments employ heavy use of sub/superscripts. Given that we are restricting consideration to one sample size ($n$) and one embedded DTR $d=(a_1, a_{2NR})$, in order to promote readability in this section we will henceforth omit subscripts identifying parameters that denote dependence on sample size and DTR assignment. 

That is, we are removing subscripts signifying the dependence of various parameters on treatment assignment and local sample size. This is purely for the purpose of conciseness, and the reader should note that these parameters are still indexed by treatment assignment and local sample size. We also note that the independence of $N$ from other random values in our simulation allows us to discuss the results below without explicitly conditioning on $N$. For example, we will use $\phi_{0,2}$ rather than $\phi_{0,2,d}$ to denote the user-specified value for $\Cov{\paren{\vct{Y}_0\d}_j,\paren{\vct{Y}_2\d}_j\mid \bX}$. Similarly, we use $\px[r]{5}$ as a shorthand for $\px[r]{5,d,n}$. We will also note that we are simulating the potential outcomes under the DTR $d$, but will use $\YY{.}$ and $R$ in lieu of $\YY{.}\d,\s R_i^{(a_{1,i})}$.

\subsubsection{Marginal Distribution of Chosen Data-Generative Model - Pre-Response Outcomes}
We recall that we began by simulating $\YY{0}$ and $\YY{1}|_{\YY{0}}$ as follows:
\begin{align*}
    \YY{0}\big|_\bX & = \eta \bX + \gamma_0\on + \ee{0}\\
    \YY{1}\big|_{\bX, \YY{0}} & = \paren{1-P_1}\paren{\eta \bX + \gamma_0\on} + \gamma_1\on + P_1\YY{0} + \gamma_2a_1\on + \ee{1},
    \shortintertext{where}
    P_1 & = \begin{cases}
        \dfrac{\rho_{0,1}}{\rho_0} & \text { if } \rho_0\neq 0\\
        0 & \text { if } \rho_0 = 0.
    \end{cases}
\end{align*}
We generated $\ee{0}$ and $\ee{1}$ jointly, with:
\begin{align*}
    \ttvec{\ee{0}}{\ee{1}} & \sim
\N{\ttvec{\vct{0}_{n}}{\vct{0}_{n}}}{
\left[\begin{array}{@{}c|c@{}}
  \Sigma_{0}
  & \paren{\phi_{0,1} - P_1\sigma_0^2}I_{n} \\
\hline
  \paren{\phi_{0,1} - P_1\sigma_0^2}I_{n} &
  \Sigma_{1}
\end{array}\right]
},\\
\shortintertext{where}
    \Sigma_{0} & = \paren{\sigma_0^2 - \rho_0} I_{n} + \rho_0\on\on^\top \text{ and }\\
    \Sigma_{1} & = \paren{\paren{\sigma_1^2 - \rho_1} I_{n} + \rho_1\on\on^\top} - P_1^2\Sigma_0 - 2P_1\paren{\phi_{0,1} - P_1\sigma_0^2}I_{n}.
\end{align*}

We immediately observe that $\E{\YY{0}\mid \bX} = \eta \bX + \gamma_0\on$ (and thus $\E{\YY{0}}=\gamma_0\on=\mu_{0,(.,.)}\on$) and $\V{\YY{0} \mid \bX} = \paren{\sigma_0^2 - \rho_0} I_{n} + \rho_0\on\on^\top$, satisfying the requirements for the marginal distribution of $\YY{0}$. Additionally, 
\begin{align*}
    \E{\YY{1}\mid \bX} & = \paren{1-P_1}\paren{\eta \bX + \gamma_0\on} + \gamma_1\on + P_1\E{\YY{0}\mid \bX} + \gamma_2a_1\on\\
    & = \eta \bX + \paren{\gamma_0 + \gamma_1 + \gamma_2a_1}\on,
    \shortintertext{so}
    \E{\YY{1}} & = \paren{\gamma_0 + \gamma_1 + \gamma_2a_1}\on\\
    & = \mu_{1, (a_1,.)}\on \text{ (by construction of $\vct{\boldsymbol\gamma}$).}
    \shortintertext{Furthermore,}
    \V{\YY{1}\mid \bX} & = \V{P_1\ee{0} + \ee{1}\mid \bX}\\
    & = P_1^2\V{\YY{0}\mid \bX} + 2P_1\Cov{\ee{0},\ee{1}} + \V{\ee{1}\mid \bX}\\
    & = P_1^2\Sigma_0 + 2P_1\paren{\phi_{0,1} - P_1\sigma_0^2}I_{n} + \Sigma_1\\
    & = \paren{\sigma_1^2 - \rho_1} I_{n} + \rho_1\on\on^\top,
\end{align*}
satisfying the requirements for the marginal distribution of $\YY{1}$.
Finally, 
\begin{align*}
    \Cov{\YY{0}, \YY{1}\mid \bX} & = \Cov{\ee{0}, P_1\ee{0} + \ee{1}}\\
    & = P_1\Sigma_0 + \paren{\phi_{0,1} - P_1\sigma_0^2}I_{n}\\
    & = P_1\paren{\paren{\sigma_0^2 - \rho_0} I_{n} + \rho_0\on\on^\top} + \paren{\phi_{0,1} - P_1\sigma_0^2}I_{n}\\
    & = \paren{\phi_{0,1} - P_1\rho_0}I_{n} + P_1\rho_0\on\on^\top.
\end{align*}
Therefore, by construction of $P_1$, $\Cov{\YY{0}, \YY{1}\mid \bX} = \paren{\phi_{0,1} - \rho_{0,1}}I_{n} + \rho_{0,1}\on\on^\top$, satisfying our requirements for the temporal covariance between $\YY{0}$ and $\YY{1}$.

\subsubsection{Marginal Distribution of Chosen Data-Generative Model - Response}

In response model $g_1$ above, we centered and standardized $\Bar{\vct{Y}_{1}}$ with respect to the (known) true mean and standard deviation. We then applied an inverse-CDF transform to simulate response likelihood via a Beta$\paren{\dfrac{p_{r}}{1-p_{r}},1}$ distribution, guaranteeing a marginal likelihood of response of $p_r$.\footnote{And ensuring that $g_1\paren{\ttvec{\YY{0}}{\YY{1}}}\in(0,1)$.} Thus, we have satisfied $\E{R}=p_r$.

\subsubsection{Conditional Distributions of Chosen Data-Generative Model - Post-Response Outcome Simulation}

Recall that we generated $\YY{2}$ according to the following data-generative procedure:\\
$\YY{2}\big|_{\bX, \YY{0},\YY{1},R=r}$
\begin{align*}
    \lalign{ = \paren{1-\px{2}-\px{4}}\paren{\eta \bX + \gamma_0\on} + \paren{1-\px{4}}\paren{\gamma_1 + \gamma_2 a_1}\on + \paren{\gamma_3 + \gamma_4a_1}\on}[0cm]\\
    &  \hspace{1cm} + \px{2}\YY{0} + \px{3}\frac{1}{n}\Sum{j}{n}\ee[j]{0}\on + \px{4}\YY{1} + \px{5}\frac{1}{n}\Sum{j}{n}\ee[j]{1}\on\\
    &  \hspace{1cm} + \paren{\paren{\gamma_5 + \gamma_6a_1}a_{2NR}}\paren{\dfrac{1-r}{1-p_r}}\on  + \paren{\lambda_1 + \lambda_2a_1}\paren{\paren{r-p_r}\on} + \eex{2},
\end{align*}
where $\eex{2}\indep \ttvec{\ee{0}}{\ee{1}}\big|_{R=r}$ and 
$\eex{2}  \sim \N{\vct{0}_{n}}{\Sigma_{2}^{[r]}}$.

This setup ensures
\begin{align*}
    \E{\YY{2}\mid \bX, R=r}
    & = \eta \bX + \paren{\gamma_0 + \gamma_1 + \gamma_2 a_1 + \gamma_3 + \gamma_4a_1}\on\\
    &  \hspace{1cm} + \px{2}\E{\ee{0}\mid R=r} + \px{3}\E{\ee{0}\mid R=r}\\
    &  \hspace{1cm} + \px{4}\E{P_1\ee{0} + \ee{1}\mid R=r} + \px{5}\E{\ee{1}\mid R=r}\\
    &  \hspace{1cm} + \paren{\paren{\gamma_5 + \gamma_6a_1}a_{2NR}}\paren{\dfrac{1-r}{1-p_r}}\on\\
    &  \hspace{1cm} + \paren{\lambda_1 + \lambda_2a_1}\paren{\paren{r-p_r}\on}\\
    &  = \eta \bX + \paren{\gamma_0 + \gamma_1 + \gamma_2 a_1 + \gamma_3 + \gamma_4a_1}\on\\
    &  \hspace{1cm} + \paren{\paren{\gamma_5 + \gamma_6a_1}a_{2NR}\paren{\dfrac{1-r}{1-p_r}} +  \paren{\lambda_1 + \lambda_2a_1}\paren{r-p_r}}\on\\
    &  \hspace{1cm} + \paren{\px[r]{2}\epsilonx{0} + \px[r]{3}\epsilonx{0} + \px[r]{4}\paren{P_{1}\epsilonx{0}+\epsilonx{1}} + \px[r]{5}\epsilonx{1}}.
\end{align*}
By construction of $\vct{\boldsymbol\gamma}$, the above resolves to
\begin{align*}
    \E{\YY{2}\mid \bX, R=r} & = \eta \bX + \paren{\mu\sr_{2,d} - v_{r}}\on\\
    &  \hspace{1cm} + \underbrace{\paren{\px[r]{2}\epsilonx{0} + \px[r]{3}\epsilonx{0} + \px[r]{4}\paren{P_{1}\epsilonx{0}+\epsilonx{1}} + \px[r]{5}\epsilonx{1}}}_{=:v_{r}}\on.
    \shortintertext{Therefore,}
    \E{\YY{2}\mid \bX, R=r} & = \eta \bX + \mu\sr_{2,d}\on\\
    \shortintertext{and}
    \E{\YY{2}\mid R=r} & = \mu\sr_{2,d}\on.
\end{align*}
The above result guarantees that the response-conditional means of the generated $\vct{Y}_2$ match those specified in the target distribution. This result, combined with the consistency of the data-generative and target response probabilities, guarantees that the data-generative model produces trajectories with the same marginal mean as the specified target distribution.

We now turn to the response-conditional variance components. We first examine the cross-response conditional variance components (i.e., $\rho\sr_{0,2}$, $\rho\sr_{1,2}$, $\phi\sr_{0,2}$, and $\phi\sr_{1,2}$). Recall that we obtained the following parameters via Monte-Carlo simulation:\\
$\V{\ttvec{\ee{0}}{\ee{1}}\mid \bX, R=r} = $
$$
\left[\begin{array}{@{}c|c@{}}
              \paren{{s\srt_{0}} - c\sr_{0}}I_{n}
              + c\sr_{0}\on\on^\top
              & \paren{{s\sr_{0,1}} - c\sr_{0,1}}I_{n}
              + c\sr_{0,1}\on\on^\top  \\
            \hline
              \paren{{s\sr_{0,1}} - c\sr_{0,1}}I_{n}
              + c\sr_{0,1}\on\on^\top &
              \paren{{s\srt_{1}} - c\sr_{1}}I_{n}
              + c\sr_{1}\on\on^\top
            \end{array}\right].
$$

We observe that for $\Cov{\YY[j]{t},\YY[k]{2}\mid \bX, R=r}$ ($j\neq k$) in the data-generative model to match the user specified value $\rho_{t,2}\sr$, we must have\\
$\rho_{0,2}\sr$
\begin{align*}
    \lalign{ = \text{Cov}\Bigg(\ee[j]{0},\px{2}\ee[k]{0} + \px{3}\frac{1}{n}\Sum{m}{n}\ee[m]{0}}[0cm]\\
     & \ralign{ + \px{4}\paren{P_1\ee[k]{0} + \ee[k]{1}} + \px{5}\frac{1}{n}\Sum{m'}{n}\ee[m']{1}\mid R=r\Bigg)}\\
    & = \px{2}\paren{c\sr_{0}}
    + \px{3}\paren{\frac{{s\srt_{0}}}{n} + \frac{n-1}{n}c\sr_{0}}
    + \px{4}\paren{P_1 c\sr_{0} + c\sr_{0,1}}\\
    & \ralign{
    + \px{5}\paren{\frac{{s\sr_{0,1}}}{n} + \frac{n-1}{n}c\sr_{0,1}},}
    \shortintertext{and $\rho_{1,2}\sr$}
     & = \text{Cov}\Bigg(P_1\ee[j]{0} + \ee[j]{1},\px{2}\ee[k]{0} + \px{3}\frac{1}{n}\Sum{m}{n}\ee[m]{0} \\
     & \ralign{ + \px{4}\paren{P_1\ee[k]{0} + \ee[k]{1}} + \px{5}\frac{1}{n}\Sum{m'}{n}\ee[m']{1}\mid R=r\Bigg)}\\
    & = P_1\Bigg(\px{2}\paren{c\sr_{0}}
    + \px{3}\paren{\paren{\frac{{s\srt_{0}}}{n} + \frac{n-1}{n}c\sr_{0}}} 
    + \px{4}\paren{P_1 c\sr_{0} + c\sr_{0,1}}\\
    & \ralign{
    + \px{5}\paren{\frac{{s\sr_{0,1}}}{n} + \frac{n-1}{n}c\sr_{0,1}}\Bigg)}\\
    & \ralign{ + \px{2}\paren{c\sr_{0,1}}
    + \px{3}\paren{\frac{{s\sr_{0,1}}}{n} + \frac{n-1}{n}c\sr_{0,1}} 
    + \px{4}\paren{P_1 c\sr_{0,1} + c\sr_{1}}}\\
    & \ralign{
    + \px{5}\paren{\paren{\frac{{s\srt_{1}}}{n} + \frac{n-1}{n}c\sr_{1}}}}\\
    & = \px{2}\paren{P_1c\sr_{0} + c\sr_{0,1}}
    + \px{3}\paren{\frac{P_1{s\srt_{0}}+{s\sr_{0,1}}}{n} + \frac{n-1}{n}\paren{P_1c\sr_{0}+c\sr_{0,1}}}\\
    & \ralign{
    + \px{4}\paren{P_1^2 c\sr_{0} + 2P_1c\sr_{0,1} + c\sr_{1}}}\\
    & \ralign{
    + \px{5}\paren{\frac{P_1{s\sr_{0,1}}+{s\srt_{1}}}{n} + \frac{n-1}{n}\paren{P_1c\sr_{0,1}+c\sr_{1}}}.}
\end{align*}
Similarly, for $\Cov{\YY[j]{t},\YY[j]{2}\mid \bX, R=r}$ in the data-generative model to match the target value $\phi_{t,2}\sr$, we must have
\begin{align*}
    \phi_{0,2}\sr & = \px{2}\paren{s\srt_{0}}
    + \px{3}\paren{\frac{{s\srt_{0}}}{n} + \frac{n-1}{n}c\sr_{0}}
    + \px{4}\paren{P_1 s\srt_{0} + s\sr_{0,1}}\\
    & \hspace{2cm}
    + \px{5}\paren{\frac{{s\sr_{0,1}}}{n} + \frac{n-1}{n}c\sr_{0,1}}
    \shortintertext{and}
    \phi_{1,2}\sr & = \px{2}\paren{P_1s\srt_{0} + s\sr_{0,1}}
    + \px{3}\paren{\frac{P_1{s\srt_{0}}+{s\sr_{0,1}}}{n} + \frac{n-1}{n}\paren{P_1c\sr_{0}+c\sr_{0,1}}}\\
    & \hspace{2cm}
    + \px{4}\paren{P_1^2 s\srt_{0} + 2P_1s\sr_{0,1} + s\srt_{1}}\\
    & \hspace{2cm}
    + \px{5}\paren{\frac{P_1{s\sr_{0,1}}+{s\srt_{1}}}{n} + \frac{n-1}{n}\paren{P_1c\sr_{0,1}+c\sr_{1}}}.
\end{align*}

Therefore, using $\Upsilon$ (as defined in Appendix \ref{c1:sec:sim_des:supp_defs}),
we see the following is a necessary and sufficient condition for our pre-response/post-response conditional variance requirements to hold:

$$\ttttvec{\phi_{0,2}\sr}{\phi_{1,2}\sr}{\rho_{0,2}\sr}{\rho_{1,2}\sr} = \Upsilon\ttttvec{\px{2}}{\px{3}}{\px{4}}{\px{5}}.$$
Recalling that $\ttttvec{\px{2}}{\px{3}}{\px{4}}{\px{5}}=\Upsilon\inv\ttttvec{\phi_{0,2}\sr}{\phi_{1,2}\sr}{\rho_{0,2}\sr}{\rho_{1,2}\sr},$ we observe that our data-generative model satisfies our cross-response conditional covariance requirements.

Finally, we check our requirements on $\V{\YY{2}\mid \bX, R=r}$. Recalling that we constructed $\eex{2}\indep \ttvec{\ee{0}}{\ee{1}}\big|_{R=r}$, we see that the following are necessary and sufficient conditions for $\Cov{\YY[j]{2},\YY[k]{2}\mid \bX, R=r}$ and $\V{\YY[j]{2}\mid \bX, R=r}$ to match their respective target values, $\rho\sr_2$ and $\sigma\srt_2$ (for any $j,k=1\ddd n,\s j\neq k$):\\
$\rho\sr_2$
\begin{align*}
     \lalign{= \text{Cov}\Bigg(\px{2}\ee[j]{0} + \px{3}\frac{1}{n}\Sum{m}{n}\ee[m]{0} + \px{4}\paren{P_1\ee[j]{0} + \ee[j]{1}} + \px{5}\frac{1}{n}\Sum{m'}{n}\ee[m']{1}}[0cm]\\
     & \ralign{ + \eex[j]{2},\s
     \px{2}\ee[k]{0} + \px{3}\frac{1}{n}\Sum{m}{n}\ee[m]{0} + \px{4}\paren{P_1\ee[k]{0} + \ee[k]{1}}
     + \px{5}\frac{1}{n}\Sum{m'}{n}\ee[m']{1}}\\
    & \ralign{ + \eex[k]{2}\mid R=r\Bigg)}\\
    & = \text{Cov}\Bigg(\paren{\px{2} + \frac{\px{3}}{n} + \px{4}P_1}\ee[j]{0,i} + \frac{\px{3}}{n}\ee[k]{0,i} + \frac{\px{3}}{n}\displaystyle\sum_{m\neq j,k}\ee[m]{0,i}  \\
    & \ralign{ + \paren{\px{4} + \frac{\px{5}}{n}}\ee[j]{1,i} + \frac{\px{5}}{n}\ee[k]{1,i} + \frac{\px{5}}{n}\displaystyle\sum_{m\neq j,k}\ee[m]{1,i} + \eex[j]{2},}\\
    & \ralign{ \paren{\px{2} + \frac{\px{3}}{n} + \px{4}P_1}\ee[k]{0,i} + \frac{\px{3}}{n}\ee[j]{0,i} + \frac{\px{3}}{n}\displaystyle\sum_{m\neq j,k}\ee[m]{0,i}}\\
    & \ralign{ + \paren{\px{4} + \frac{\px{5}}{n}}\ee[k]{1,i} + \frac{\px{5}}{n}\ee[j]{1,i} + \frac{\px{5}}{n}\displaystyle\sum_{m\neq j,k}\ee[m]{1,i} + \eex[k]{2}\mid R=r\Bigg)}\\
    & = \paren{\px{2} + \frac{\px{3}}{n} + \px{4}P_1}\paren{\paren{\px{2} + \frac{\px{3}}{n}\paren{n-1} + \px{4}P_1}\cx{0} + \frac{\px{3}}{n}\ssx{0}}\\
    & \ralign{
    + \paren{\px{2} + \frac{\px{3}}{n} + \px{4}P_1}\paren{\paren{\px{4} + \frac{\px{5}}{n}\paren{n-1}}\cx{0,1} + \frac{\px{5}}{n}\sx{0,1}}}\\
    & \ralign{+ \frac{\px{3}}{n}\paren{\px{2} + \frac{\px{3}}{n} + \px{4}P_1}\ssx{0} + \paren{\frac{\px{3}}{n}}^2\paren{n-1}\cx{0}}\\
    & \ralign{
    + \frac{\px{3}}{n}\paren{\px{4} + \frac{\px{5}}{n}}\sx{0,1} + \frac{\px{3}\px{5}}{n^2}\paren{n-1}\cx{0,1}}\\
    & \ralign{
    + \frac{\px{3}}{n}\paren{n-2}\paren{\paren{\px{2} + 2\frac{\px{3}}{n} + \px{4}P_1}\cx{0} + \frac{\px{3}}{n}\left[\ssx{0}+\paren{n-3}\cx{0}\right]}}\\
    & \ralign{
    + \frac{\px{3}}{n}\paren{n-2}\paren{\px{4} + 2\frac{\px{5}}{n}}\cx{0,1} + \frac{\px{3}\px{5}}{n^2}\paren{n-2}\left[\sx{0,1}+\paren{n-3}\cx{0,1}\right]}\\
    & \ralign{+ \paren{\px{4} + \frac{\px{5}}{n}}\paren{\paren{\px{2} + \frac{\px{3}}{n}\paren{n-1} + \px{4}P_1}\cx{0,1} + \frac{\px{3}}{n}\sx{0,1}}}\\
    & \ralign{
    + \paren{\px{4} + \frac{\px{5}}{n}}\paren{\paren{\px{4} + \frac{\px{5}}{n}\paren{n-1}}\cx{1} + \frac{\px{5}}{n}\ssx{1}}}\\
    &\ralign{ + \frac{\px{5}}{n}\paren{\px{2} + \frac{\px{3}}{n} + \px{4}P_1}\sx{0,1} + \frac{\px{3}\px{5}}{n^2}\paren{n-1}\cx{0,1}}\\
    & \ralign{
    + \frac{\px{5}}{n}\paren{\px{4} + \frac{\px{5}}{n}}\ssx{1} + \paren{\frac{\px{5}}{n}}^2\paren{n-1}\cx{1}}\\
    & \ralign{
    + 
    \frac{\px{5}}{n}\paren{n-2}\paren{\paren{\px{2} + 2\frac{\px{3}}{n} + \px{4}P_1}\cx{0,1} + \frac{\px{3}}{n}\left[\sx{0,1}+\paren{n-3}\cx{0,1}\right]}}\\
    & \ralign{
    + \frac{\px{5}}{n}\paren{n-2}\paren{\px{4} + 2\frac{\px{5}}{n}}\cx{1} + \paren{\frac{\px{5}}{n}}^2\paren{n-2}\left[\ssx{1}+\paren{n-3}\cx{1}\right]}\\
    & \ralign{+ \paren{\Sigma_{2}^{[r]}}_{j,k}}\\
    & =: \zeta' + \paren{\Sigma_{2}^{[r]}}_{j,k},
\end{align*}
and $\sigma\srt_2$
\begin{align*}
     \lalign{ = \mathbb{V}\Bigg[\paren{\px{2}+\frac{\px{3}}{n}+\px{4}P_1}\ee[j]{0} + \frac{\px{3}}{n}\displaystyle\sum_{k\neq j}\ee[k]{0,i} }[0cm]\\
     & \ralign{ + \paren{\px{4} + \frac{\px{5}}{n}}\ee[j]{1,i}  
     + \frac{\px{5}}{n}\displaystyle\sum_{k\neq j}\ee[k]{1,i} + \eex[j]{2}\mid R=r\Bigg]}\\
    & = \paren{\px{2}+\frac{\px{3}}{n}+\px{4}P_1}^2\ssx{0}
    + \paren{\frac{\px{3}}{n}}^2\paren{n-1}\left[\ssx{0}+\paren{n-2}\cx{0}\right]\\
    & \ralign{
    + \paren{\px{4} + \frac{\px{5}}{n}}^2\ssx{1} + 
    \paren{\frac{\px{5}}{n}}^2\paren{n-1}\left[\ssx{1}+\paren{n-2}\cx{1}\right] + \paren{\Sigma_{2}^{[r]}}_{j,j}}\\
    & \ralign{+ 2\paren{\px{2}+\frac{\px{3}}{n}+\px{4}P_1}\Bigg(\hspace{8.5cm}}\\
    & \ralign{\times\paren{\frac{\px{3}}{n}\paren{n-1}\cx{0} + \paren{\px{4} + \frac{\px{5}}{n}}\sx{0,1} + \frac{\px{5}}{n}\paren{n-1}\cx{0,1}}\Bigg)}\\
    & \ralign{
    + 2\frac{\px{3}}{n}\paren{n-1}\paren{\paren{\px{4} + \frac{\px{5}}{n}}\cx{0,1} + \frac{\px{5}}{n}\left[\sx{0,1}+\paren{n-2}\cx{0,1}\right]}}\\
    & \ralign{+ 2\paren{\px{4} + \frac{\px{5}}{n}}\paren{\frac{\px{5}}{n}}\paren{n-1}\cx{1}}\\
    & =: \zeta + \paren{\Sigma_{2}^{[r]}}_{j,j}.
\end{align*}

Given that we generated $\ee{2}\sim\N{\vct{0}_{n}}{\Sigma_{2}^{[r]}}$, where 
$$\Sigma_{2}^{[r]} = \paren{\paren{\sigma\srt_2-\zeta} - \paren{\rho\sr_{2}-\zeta'}}I_{n} + \paren{\rho\sr_{2}-\zeta'}\on\on^\top,$$ 
our construction of $\Sigma_{2}^{[r]}$ guarantees that the overall data-generating model satisfies the requirements on $\V{\YY{2}\mid \bX, R=r}$.

\subsection{Parameter Selection}\label{c1:sec:sim_des:parm_sel}
Appendix \ref{c1:sec:sim_des:dgm} discussed the arguments to select in order to fully specify the data-generative model. The core parameters defining the simulation settings in Section \ref{c1:sec:sim_study} were derived via analysis of ASIC data. To obtain marginal mean/variance data-generative parameters under a given variance setting, we used the method described in the main body (along with the example marginal mean model in Equation~(\ref{c1:eq:mmm_ex}) of the main body), assuming the working variance structure in question, to model weekly delivery of CBT in the ASIC data.\footnote{Note that, for this purpose, we restricted the ASIC data to times $t_0$, $t^*$, and $t_T$ in order to provide parameters to specify a data-generative model with three time points.} 

As noted above, we also needed to specify the conditional mean/variance parameters for responding clusters - we did so by analyzing the non-responding clusters in ASIC (separately for each embedded DTR). We then averaged non-responding parameters across first-stage treatment arms and used these, along with the marginal distribution parameters, to derive responder parameters.\footnote{We did so because ASIC had a low response rate, which may lead to unstable conditional distribution parameters for responding clusters.}\footnote{Occasionally, this approach led to distributions which were incompatible with the data-generative model. In these cases, we made slight adjustments to the parameter inputs, documented on \href{https://github.com/GabrielDurham/three_lvl_cSMART_sim_study}{GitHub}.} As recommended in \citet{Pan_U}, we assumed all covariance terms were non-negative when fitting the marginal model. While this assumption is reasonable in ASIC-type scenarios and facilitates estimator stability, it is not necessarily justifiable when considering conditional covariance terms.\footnote{For example, if a non-responding cluster has two individuals and one of them was well-off at time $t^*$, then it was likely the case that the other individual was doing sufficiently poorly as to offset their counterpart and trigger cluster-wide non-response.} Therefore, we fit the conditional models using unrestricted correlation estimation.

To promote numerical stability across simulations, we chose to use a response rate of $50\%$. Additionally, we simulated two centered covariates, one individual-level with normal distribution and one cluster-level with uniform distribution, and included these covariates in the model fitting.

Lastly, as discussed in Appendix \ref{c1:sec:sim_des:dgm:dgp}, we simulated DTR assignment prior to any outcome/covariates.\footnote{Note that this simulation was independent of all other data, and therefore is equivalent to simulating $A_1, R, A_{2}$ sequentially.} We assumed equal probability of assignment across all four possible embedded DTRs, and used complete randomization for assignment.

Below, we describe the model-fitting settings to obtain the data-generative parameters for the simulations in Section~\protect\ref{c1:sec:sim_study} of the main body. The exact parameters can be found at:
\url{https://github.com/GabrielDurham/three_lvl_cSMART_sim_study}.

\subsubsection*{Section \ref{c1:sec:sim_study:validity}}
\begin{itemize}
    \item Outcome: Weekly CBT delivery.
    \item Variance Model: Heteroscedastic and heterogeneous across embedded DTR.
    \item Correlation Model: AR(1) within-person; Exchangeable between-person; Heterogeneous across embedded DTR.
    \item Number of Clusters: 20-500; Cluster Sizes: 2-3.
    \item Effect Size: $\approx 0.34$.
\end{itemize}
\subsubsection*{Section \ref{c1:sec:sim_study:wvm}}

\begin{itemize}
    \item Outcome: Aggregate CBT delivery.
    \item Variance Model: 
        \subitem Complex Data-Generative Structure: Heteroscedastic and homogeneous across embedded DTR.
        \subitem Simple Data-Generative Structure: Homoscedastic and homogeneous across embedded DTR.
    \item Correlation Model: 
        \subitem Complex Data-Generative Structure: AR(1) within-person; Unstructured \\
        between-person; Homogeneous across embedded DTR.
        \subitem Simple Data-Generative Structure: Independent.
    \item Number of Clusters: 20-150; Cluster Sizes: 2-3.
    \item Effect Sizes: $\approx 0.06$ (Complex), $\approx -0.14$ (Simple).
\end{itemize}
\subsubsection*{Section \ref{c1:sec:sim_study:power}}
\begin{itemize}
    \item Outcome: 
        \subitem High-Correlation: Aggregate CBT delivery.
        \subitem Low-Correlation: Weekly CBT delivery.
    \item Variance Model: Heteroscedastic and homogeneous across embedded DTR.
    \item Correlation Model: AR(1) within-person; Exchangeable between-person; Homogeneous across embedded DTR.
    \item Number of Clusters: 26-661; Cluster Size: 2.
    \item Effect Sizes: 0.2-1.0.
\end{itemize}

For all simulation environments other than those in Section \ref{c1:sec:sim_study:power}, we randomized cluster size to be two with probability 0.67 and three with probability 0.33 - matching the empirical distribution of ASIC cluster sizes (among schools with more than one SP). For the simulations in Section \ref{c1:sec:sim_study:power}, we used a uniform cluster size of two (as the sample size formula in \citet{NeCamp2017} requires a uniform sample size). All results in Section \ref{c1:sec:sim_study} of the main body are based on $20,000$ Monte Carlo simulations per data-generative environment.\footnote{In rare situations ($\approx 0.01\%$ of simulations), inference failed due to singular/near-singular estimate components. Given the rarity of such situations, we restricted our analysis to the cases which did not display these issues.} Given this large number of iterations, the resulting Monte Carlo confidence intervals are narrow. Full interval estimates are available at: \url{https://github.com/GabrielDurham/three_lvl_cSMART_sim_study}.

\begin{landscape}
\subsection{Supplemental Definitions}\label{c1:sec:sim_des:supp_defs}
After using Monte-Carlo simulation to obtain the following $s$ and $c$ parameters:
$$\V{\ttvec{\ee{0}}{\ee{1}}\mid \bX, R=r}
 = 
\left[\begin{array}{@{}c|c@{}}
              \paren{{s\srt_{0}} - c\sr_{0}}I_{n}
              + c\sr_{0}\on\on^\top
              & \paren{{s\sr_{0,1}} - c\sr_{0,1}}I_{n}
              + c\sr_{0,1}\on\on^\top  \\
            \hline
              \paren{{s\sr_{0,1}} - c\sr_{0,1}}I_{n}
              + c\sr_{0,1}\on\on^\top &
              \paren{{s\srt_{1}} - c\sr_{1}}I_{n}
              + c\sr_{1}\on\on^\top
            \end{array}\right],$$
we employ the following notation:\\
$\Upsilon$
$$:=
\begin{bmatrix}
s\srt_{0} & \frac{{s\srt_{0}}}{n} + \frac{n-1}{n}c\sr_{0} & P_1 s\srt_{0} + s\sr_{0,1} & \frac{{s\sr_{0,1}}}{n} + \frac{n-1}{n}c\sr_{0,1} \\
P_1s\srt_{0} + s\sr_{0,1} & \frac{P_1{s\srt_{0}}+{s\sr_{0,1}}}{n} + \frac{n-1}{n}\paren{P_1c\sr_{0}+c\sr_{0,1}} & P_1^2 s\srt_{0} + 2P_1s\sr_{0,1} + s\srt_{1} & \frac{P_1{s\sr_{0,1}}+{s\srt_{1}}}{n} + \frac{n-1}{n}\paren{P_1c\sr_{0,1}+c\sr_{1}} \\
c\sr_{0} & \frac{{s\srt_{0}}}{n} + \frac{n-1}{n}c\sr_{0} & P_1 c\sr_{0} + c\sr_{0,1} & \frac{{s\sr_{0,1}}}{n} + \frac{n-1}{n}c\sr_{0,1} \\ 
P_1c\sr_{0} + c\sr_{0,1} & \frac{P_1{s\srt_{0}}+{s\sr_{0,1}}}{n} + \frac{n-1}{n}\paren{P_1c\sr_{0}+c\sr_{0,1}} & P_1^2 c\sr_{0} + 2P_1c\sr_{0,1} + c\sr_{1} & \frac{P_1{s\sr_{0,1}}+{s\srt_{1}}}{n} + \frac{n-1}{n}\paren{P_1c\sr_{0,1}+c\sr_{1}}
\end{bmatrix},
$$
\begin{align*}
    \lalign{\zeta := \paren{\px{2}+\frac{\px{3}}{n}+\px{4}P_1}^2\ssx{0}
    + \paren{\frac{\px{3}}{n}}^2\paren{n-1}\left[\ssx{0}+\paren{n-2}\cx{0}\right]+ \paren{\px{4} + \frac{\px{5}}{n}}^2\ssx{1}}\\
    & \ralign{
     + 
    \paren{\frac{\px{5}}{n}}^2\paren{n-1}\left[\ssx{1}+\paren{n-2}\cx{1}\right]} \\
    & \ralign{ + 2\paren{\px{2}+\frac{\px{3}}{n}+\px{4}P_1}
    \paren{\frac{\px{3}}{n}\paren{n-1}\cx{0} + \paren{\px{4} + \frac{\px{5}}{n}}\sx{0,1} + \frac{\px{5}}{n}\paren{n-1}\cx{0,1}}}\\
    & \ralign{
    + 2\frac{\px{3}}{n}\paren{n-1}\paren{\paren{\px{4} + \frac{\px{5}}{n}}\cx{0,1} + \frac{\px{5}}{n}\left[\sx{0,1}+\paren{n-2}\cx{0,1}\right]}+ 2\paren{\px{4} + \frac{\px{5}}{n}}\paren{\frac{\px{5}}{n}}\paren{n-1}\cx{1},}
\end{align*}
and\newpage
\begin{align*}
    \lalign{\zeta' :=\paren{\px{2} + \frac{\px{3}}{n} + \px{4}P_1}\paren{\paren{\px{2} + \frac{\px{3}}{n}\paren{n-1} + \px{4}P_1}\cx{0} + \frac{\px{3}}{n}\ssx{0}}}\\
    & \ralign{
    + \paren{\px{2} + \frac{\px{3}}{n} + \px{4}P_1}\paren{\paren{\px{4} + \frac{\px{5}}{n}\paren{n-1}}\cx{0,1} + \frac{\px{5}}{n}\sx{0,1}}}\\
    & \ralign{+ \frac{\px{3}}{n}\paren{\px{2} + \frac{\px{3}}{n} + \px{4}P_1}\ssx{0} + \paren{\frac{\px{3}}{n}}^2\paren{n-1}\cx{0}
    + \frac{\px{3}}{n}\paren{\px{4} + \frac{\px{5}}{n}}\sx{0,1}
    }\\
    & \ralign{
     + \frac{\px{3}\px{5}}{n^2}\paren{n-1}\cx{0,1}
     + \frac{\px{3}}{n}\paren{n-2}\paren{\px{2} + 2\frac{\px{3}}{n} + \px{4}P_1}\cx{0} }\\
    & \ralign{+ \paren{\frac{\px{3}}{n}}^2\paren{n-2}\left[\ssx{0}+\paren{n-3}\cx{0}\right]
    + \frac{\px{3}}{n}\paren{n-2}\paren{\px{4} + 2\frac{\px{5}}{n}}\cx{0,1}}\\
    & \ralign{
     + \frac{\px{3}\px{5}}{n^2}\paren{n-2}\left[\sx{0,1}+\paren{n-3}\cx{0,1}\right]
     + \paren{\px{4} + \frac{\px{5}}{n}}\paren{\paren{\px{2} + \frac{\px{3}}{n}\paren{n-1} + \px{4}P_1}\cx{0,1} + \frac{\px{3}}{n}\sx{0,1}}}\\
    & \ralign{
    + \paren{\px{4} + \frac{\px{5}}{n}}\paren{\paren{\px{4} + \frac{\px{5}}{n}\paren{n-1}}\cx{1} + \frac{\px{5}}{n}\ssx{1}}}\\
    & \ralign{+ \frac{\px{5}}{n}\paren{
    \paren{\px{2} + \frac{\px{3}}{n} + \px{4}P_1}\sx{0,1} + \frac{\px{3}}{n}\paren{n-1}\cx{0,1}
    + \paren{\px{4} + \frac{\px{5}}{n}}\ssx{1} + \frac{\px{5}}{n}\paren{n-1}\cx{1}}}\\
    & \ralign{+ \frac{\px{5}}{n}\paren{n-2}\paren{\px{2} + 2\frac{\px{3}}{n} + \px{4}P_1}\cx{0,1} 
    + \paren{\frac{\px{3}\px{5}}{n^2}}\paren{n-2}\left[\sx{0,1}+\paren{n-3}\cx{0,1}\right]}\\
    & \ralign{
    + \frac{\px{5}}{n}\paren{n-2}\paren{\px{4} + 2\frac{\px{5}}{n}}\cx{1} 
    +\paren{\frac{\px{5}}{n}}^2\paren{n-2}\left[\ssx{1}+\paren{n-3}\cx{1}\right].}
\end{align*}
\end{landscape}

\section{Additional Simulations}\label{c1:sec:add_sim}
\subsection{Efficiency Gains and Model Saturation}\label{c1:sec:add_sim:oversat}
In this appendix, we draw connections between existing work on marginal mean modeling for longitudinal analyses of SMARTs (e.g., \citet{Lu2015, Seewald2019}) and recent developments in saturated longitudinal models (e.g., \citet{Wang2023}), situating both within the three-level framework introduced in the main text.

\citet{Wang2023} show that, in longitudinal settings with continuous outcomes, a saturated linear model yields asymptotic efficiency equivalent to that of a static analysis using only outcomes collected at end-of-study. To clarify the origin of the gains seen in Section~\ref{c1:sec:sim_study:power}, we compare several model specifications. As in Section~\ref{c1:sec:sim_study}, all simulations in Appendix~\ref{c1:sec:add_sim:oversat} include outcome data only at times $t_0$, $t^*$, and $t_T$. When modeling outcomes at only these three time points, Equation~(\ref{c1:eq:mmm_ex}), as a marginal mean model for a prototypical cSMART, is guaranteed to be correctly specified; there is no smoothing over time. All discussions in Appendix~\ref{c1:sec:add_sim:oversat} refer to this three time-point setting. We explore potential gains from temporal smoothing in Appendix~\ref{c1:sec:add_sim:multi_t}.

\subsubsection{Saturated and Over-Saturated Mean Modeling}\label{c1:sec:add_sim:oversat:oversat}
\citet{Wang2023} consider a model for an outcome ($Y$) in a randomized trial (with treatments $a\in\mathcal{A}$) which can freely specify $\E{Y_t^{(a)}}$ for all times $t$ and treatments $a$. In the context of longitudinal analyses of prototypical SMARTs, such an approach would produce an \quotes{over-saturated} marginal mean model (Equation~\eqref{c1:eq:mmm_over_sat} below).\footnote{We have omitted covariates in this discussion for simplicity.}

\begin{equation} \label{c1:eq:mmm_over_sat}
\begin{split}
    \mu_t(a_1, a_{2NR}; \theta) & = \Ind{t=t_0}\paren{\theta_0 + \theta_1a_1 + \theta_2a_{2NR} + \theta_3a_1a_{2NR}}\\
    & \hspace{.5cm} + \Ind{t=t^*}\paren{\theta_4 + \theta_5a_1 + \theta_6a_{2NR} + \theta_7a_1a_{2NR}}\\
    & \hspace{.5cm} + \Ind{t=t_T}\paren{\theta_8 + \theta_9a_1 + \theta_{10}a_{2NR} + \theta_{11}a_1a_{2NR}}.\\
\end{split}
\end{equation}

Although Equation~\eqref{c1:eq:mmm_over_sat} allows unrestricted specification of $\E{Y_t^{(d)}}$ for each $t \in {t_0, t^*, t_T}$ and $d \in \mathcal{D}$, this flexibility does not reflect the structural constraints imposed by the prototypical SMART design. In contrast, Equation~(\ref{c1:eq:mmm_ex}) in the main text respects these constraints. As formalized in the identification assumptions in Appendix~\ref{c1:sec:id_ass}, the collection $\set{\E{Y_t^{(d)}}}_{t \in {t_0, t^, t_T}}^{d \in \mathcal{D}}$ cannot vary freely: for instance, at baseline $(t_0)$, no treatment has been delivered, so $\E{Y_{t_0}^{(d)}}$ must be the same for all $d \in \mathcal{D}$. Such constraints induce seven independent means in $\set{\E{Y_t^{(d)}}}_{t \in {t_0, t^, t_T}}^{d \in \mathcal{D}}$, reflected in the seven parameters in Equation~(\ref{c1:eq:mmm_ex}), as opposed to the 12 parameters reflected in Equation~\ref{c1:eq:mmm_over_sat} above.

We revisit Table~\ref{c1:tab:power_comp} from the main text, modifying the longitudinal analysis to use the over-saturated mean model of Equation~\eqref{c1:eq:mmm_over_sat}.\footnote{I.e., the models compared in Table~\ref{c1:tab:eff_comp:oversat} are identical to those in Table~\ref{c1:tab:power_comp}, with the sole exception of marginal mean model for the longitudinal competitor.} Consistent with the findings in \cite{Wang2023}, this specification yields negligible differences in efficiency relative to the static analysis.

\begin{table}[ht]
\centering\begin{tabular}{ccc|cc|cc}
\hline
\multirow{2}{*}{Effect Size $(N)$}  & \multicolumn{2}{c}{RMSE} & \multicolumn{2}{c}{Coverage} & \multicolumn{2}{c}{Power} \\\cline{2-7}
      & \makecell{Long.\\(Over-Sat.)} & Static & \makecell{Long.\\(Over-Sat.)} & Static & \makecell{Long.\\(Over-Sat.)} & Static\\
  \hline
$\delta=0.2$ ($N=661$) &  7.03 &  7.03 & 0.96 & 0.95 & 0.78 & 0.81 \\ 
  $\delta=0.5$ ($N=106$) & 17.58 & 17.58 & 0.96 & 0.94 & 0.79 & 0.82 \\ 
  $\delta=0.8$ ($N=42$) & 28.15 & 28.15 & 0.94 & 0.93 & 0.81 & 0.83 \\ 
  $\delta=1$ ($N=27$) & 35.04 & 35.04 & 0.92 & 0.91 & 0.83 & 0.85 \\ 
 \hline
\end{tabular}
\caption{Static vs Over-Saturated Longitudinal Efficiency Comparison}\label{c1:tab:eff_comp:oversat}
\end{table}

\subsubsection{Variance Modeling}\label{c1:sec:add_sim:oversat:var}

To further explore the contribution of variance modeling, we next compare two longitudinal approaches: one which mirrors the longitudinal model in Table~\ref{c1:tab:power_comp} of the main body and models both within- and between- person correlation, and another which adopts an independent working correlation model with time-heterogeneous variance.\footnote{The latter model matches the independent-heteroscedastic model in Section~\ref{c1:sec:sim_study:wvm} of the main body.} Both approaches use Equation~(\ref{c1:eq:mmm_ex}) for marginal mean modeling.

\begin{table}[H]
\centering\begin{tabular}{ccc|cc|cc}
\hline
\multirow{2}{*}{Effect Size $(N)$}  & \multicolumn{2}{c}{RMSE} & \multicolumn{2}{c}{Coverage} & \multicolumn{2}{c}{Power} \\\cline{2-7}
       & Long. & \makecell{Long.\\(Ind.)} & Long. & \makecell{Long.\\(Ind.)} & Long. & \makecell{Long.\\(Ind.)}\\
  \hline
$\delta=0.2$ ($N=661$) &  6.18 &  7.03 & 0.95 & 0.95 & 0.90 & 0.81 \\ 
  $\delta=0.5$ ($N=106$) & 15.53 & 17.58 & 0.94 & 0.94 & 0.90 & 0.82 \\ 
  $\delta=0.8$ ($N=42$) & 24.77 & 28.15 & 0.93 & 0.93 & 0.91 & 0.83 \\ 
  $\delta=1$ ($N=27$) & 30.74 & 35.04 & 0.92 & 0.91 & 0.92 & 0.85 \\ 
 \hline
\end{tabular}
\caption{Longitudinal vs Longitudinal (Independent) Efficiency Comparison}\label{c1:tab:eff_comp:long_ind}
\end{table}

Table~\ref{c1:tab:eff_comp:long_ind} supports the results in Section~\ref{c1:sec:sim_study:wvm} of the main body, showing proper variance modeling can yield efficiency gains.

\subsubsection{Mean and Variance Modeling}\label{c1:sec:add_sim:oversat:mean_var}

Table~\ref{c1:tab:power_comp} of the main body shows that, through careful mean and variance modeling, incorporating repeated measures can increase efficiency for end-of-study analyses. Moreover, Appendices~\ref{c1:sec:add_sim:oversat:oversat} and \ref{c1:sec:add_sim:oversat:var} show that, in isolation, neither mean nor variance modeling alone can achieve the efficiency gained through modeling both. To summarize these findings, Table~\ref{c1:tab:eff_comp:unified} below consolidates results from Tables~\ref{c1:tab:eff_comp:oversat} and \ref{c1:tab:eff_comp:long_ind}. 

\begin{table}[H]
\centering
\setlength{\tabcolsep}{5.5pt}
\begin{tabular}{cccc|ccc|ccc}
\hline
\multirow{2}{*}{\makecell{Effect Size \\$(N)$}}  & \multicolumn{3}{c}{RMSE} & \multicolumn{3}{c}{Coverage} & \multicolumn{3}{c}{Power} \\\cline{2-10}
      & Long. & \makecell{Long.\\(Over-\\Sat.)} & \makecell{Long.\\(Ind.)} &  Long. & \makecell{Long.\\(Over-\\Sat.)} & \makecell{Long.\\(Ind.)} &  Long. & \makecell{Long.\\(Over-\\Sat.)} & \makecell{Long.\\(Ind.)}\\
  \hline
\makecell{$\delta=0.2$ \\($N=661$)} &  6.18 &  7.03 &  7.03 & 0.95 & 0.96 & 0.95 & 0.90 & 0.78 & 0.81 \\ 
  \makecell{$\delta=0.5$ \\($N=106$) }& 15.53 & 17.58 & 17.58 & 0.94 & 0.96 & 0.94 & 0.90 & 0.79 & 0.82 \\ 
  \makecell{$\delta=0.8$\\ ($N=42$)} & 24.77 & 28.15 & 28.15 & 0.93 & 0.94 & 0.93 & 0.91 & 0.81 & 0.83 \\ 
  \makecell{$\delta=1$ \\($N=27$)} & 30.74 & 35.04 & 35.04 & 0.92 & 0.92 & 0.91 & 0.92 & 0.83 & 0.85 \\ 
 \hline
\end{tabular}
\caption{Unified Efficiency Comparisons}\label{c1:tab:eff_comp:unified}
\parbox{0.9\linewidth}{\raggedright\scriptsize\textit{Long. represents longitudinal model employed in Table~\ref{c1:tab:power_comp}.\\
Long. (Over-Sat.) represents an over-saturated longitudinal model, as described in Appendix~\ref{c1:sec:add_sim:oversat:oversat}.\\
Long. (Ind.) represents the longitudinal model with independent working correlation discussed in Appendix~\ref{c1:sec:add_sim:oversat:var}}.}
\end{table}

In particular, we find that using an over-saturated marginal mean model reduces efficiency relative to a model that respects the randomization structure of the SMART. This aligns with the approach taken in \citet{Lu2015, Seewald2019}, who employ such design-aware specifications in place of over-saturated alternatives like Equation~\eqref{c1:eq:mmm_over_sat}.

\subsection{Additional Time Points}\label{c1:sec:add_sim:multi_t}
This manuscript introduces a method for modeling the marginal means and variances of outcomes under different embedded cAIs. To evaluate the method, we required a data-generative framework that could directly control these marginal parameters. At the same time, we aimed to simulate data in a way that remained realistic, and therefore chose to generate outcomes conditionally. These two goals motivated the framework described above, which produces data at three time points: $t_0$, $t^*$, and $t_T$. 

To examine the performance of the proposed method in settings with more than three time points, we adopted a simplified data-generative approach. This framework allowed us to influence, though not precisely control, the marginal means and variances. For each cluster, data were generated in three steps: $(i)$ simulate outcomes $Y_1\ddd Y_{t^*}$; $(ii)$ simulate response given $Y_1\ddd Y_{t^*}$;\footnote{For this setting, we chose the same response-generative procedure as described in Appendix~\ref{c1:sec:sim_des:dgm:dgp}, using data from time $t^*$ rather than $t=1$.} $(iii)$ simulate outcomes $Y_{t^*+1}\ddd Y_{t_T}$ conditional on response.

We present results for a setting in which five time points were modeled: $\{t_0=0, 1, t^*=2, 3, t_T=4\}$, mimicking the analysis in Section~\ref{c1:sec:data_an} of the main body. The data-generative environment was constructed to mirror the setting used in Table~\ref{c1:tab:power_comp_lc}, row 1, of the main text.\footnote{This setup was the most similar to the primary aim analysis presented in Section~\ref{c1:sec:data_an} among those assuming a homogeneous correlation structure across embedded cAIs.} As in that earlier setting, the empirical AR(1) coefficient and intraclass correlation were approximately 0.27 and 0.07, respectively. The simulation assumed a true effect size of 0.2, cluster sizes of two, and produced piecewise linear marginal means (as in Equation~(\ref{c1:eq:mmm_ex}) of the main body). Results in Table~\ref{c1:tab:Multi_t_Sim} are based on 20,000 Monte Carlo iterations.

\begin{table}[H]
\centering
\begin{tabular}{c | cc | cc | cc | cc}
\hline
Number & \multicolumn{2}{c|}{Relative Bias} & \multicolumn{2}{c|}{RMSE} & \multicolumn{2}{c|}{Coverage} & \multicolumn{2}{c}{CI Width} \\
of Clusters
& $3$ & $5$ 
& $3$ & $5$ 
& $3$ & $5$ 
& $3$ & $5$ \\
\hline
N=20 & -0.03 & -0.04 & 1.92 & 1.89 & 0.92 & 0.92 & 7.25 & 7.05 \\ 
  N=30 & -0.03 & -0.04 & 1.54 & 1.51 & 0.93 & 0.93 & 5.94 & 5.80 \\ 
  N=40 & -0.03 & -0.03 & 1.36 & 1.33 & 0.93 & 0.93 & 5.15 & 5.04 \\ 
  N=50 & -0.03 & -0.03 & 1.20 & 1.18 & 0.94 & 0.94 & 4.63 & 4.53 \\ 
  N=75 & -0.03 & -0.04 & 0.99 & 0.97 & 0.94 & 0.94 & 3.80 & 3.72 \\ 
  N=100 & -0.03 & -0.04 & 0.85 & 0.84 & 0.94 & 0.94 & 3.30 & 3.23 \\ 
\hline
\end{tabular}
\caption{Estimator Performance: Number of Time Points Modeled (Three versus Five)}\label{c1:tab:Multi_t_Sim}
\vspace{-0.8em}
\caption*{\raggedright\scriptsize\textit{Column subheadings indicate the number of time points modeled.}}
\end{table}

Table~\ref{c1:tab:Multi_t_Sim} compares the performance of the proposed estimator when modeling all five time points versus modeling only three time points: $t_0$, $t^*$, and $t_T$. The two approaches yield similar performance with respect to coverage and relative bias. However, modeling additional time points leads to improved precision, reflected in narrower $95\%$ confidence intervals.

\section{Asymptotic Results}\label{c1:sec:asy_thry}
Recall the estimating equation presented in Section \ref{c1:sec:est} of the main body:
$$
    0   = \SumiN\Sumd I_i(d) W_i D(d,\bX_i;\theta)^\top V^d(\bX_i;\alpha)\inv(\vbY_i - \vbmu(d,\bX_i;\theta)).
$$
In this appendix, we first consider the case of estimated weights,\footnote{As the result under known weights is a simple corollary of this case.} i.e., $W_i=W_i(\pihat)$.
As we show below, the solution of this estimating equation has the following asymptotic distribution:
$$\sqrt{N}\paren{\thetahat-\theta_0}\xrightarrow[N\to\infty]{\mathcal{L}}\N{0}{J\inv Q J\inv},$$
where 
\begin{align*}
    J  & = \E{\Sumd I(d) W(\pi_0)D(d,\bX;\theta)^\top \Sigma^d(\alpha_+)\inv D(d,\bX;\theta)}\text{, and }\\
    Q  & = \E{UU^\top} - \E{US_{\pi_0}^\top}\E{S_{\pi_0}S_{\pi_0}^\top}\inv\E{S_{\pi_0}U^\top},
\end{align*}
with $S$ as defined in Appendix \ref{c1:sec:wgt_est}. This appendix also presents the technical motivation for this choice of estimating equation.

We note that the derivation of this estimating equation and the demonstration of the asymptotic distribution of its solution are standard statistical arguments. However, we include detailed versions here for completeness and to assist novice readers who may be interested in the technical derivation of the estimator and its large sample properties. For the reader familiar with such methods, earlier works such as \citep{Lu2015, NeCamp2017, Seewald2019} succinctly present similar arguments. 

Throughout the appendix, we consider comparison of embedded cAIs with respect to the marginal mean of a continuous (repeatedly measured) outcome, $Y$, in a general $K$-stage clustered SMART with arbitrary clustering structure (defined by covariance $\Sigma$).

\subsection{Derivation of Estimating Equation}
Solving Equation~(\ref{c1:eq:est_eq}) in the main body provides a framework for parameter estimation and inference. This approach is similar to the generalized estimating equation of Liang and Zeger \citep{LIANG1986}. A statistician can arrive at this estimating equation from many schools of thought, one of which we discuss below.

Given data $\set{\paren{\bY_i, \bX_i}}_{i=1\ddd N}$ from a clustered SMART, a statistician attempting to model $\E{\vct{\bY}_i\d\mid \bX_i}$ by $\vbmu(d,\bX_i;\theta)$, with $\vbY_i\d-\vbmu(d,\bX_i;\theta)$ having variance $\Sigma_i^d$ (i.e., $\V{\vbY_i\d-\vbmu(d,\bX_i;\theta)\mid \bX_i}=:\Sigma_i^d$), may wish to do so by taking
\begin{equation}\label{c1:eq:loss_full}
    \thetahat = \argmin_{\theta\in\mathbb{R}^p}\SumiN\sum_{d\in\mathcal{D}}
\ltnorm{{\Sigma_i^d}^{-\frac{1}{2}}\paren{\vbY_i\d-\vbmu\paren{d, \bX_i;\theta}}}^2.
\end{equation}

However, the statistician cannot apply Equation \ref{c1:eq:loss_full}, as they will not be able to observe all potential outcomes for any given cluster; i.e., the statistician cannot observe $\vbY_i\d$ for all $d\in\mathcal{D}$. To correct for this phenomenon, the statistician can consider the weighted loss function, depending only on observed data:
\begin{equation}\label{c1:eq:loss}
    \thetahat = \argmin_{\theta\in\mathbb{R}^p}\SumiN\sum_{d\in\mathcal{D}}
I_i(d) W_i \ltnorm{{\Sigma_i^d}^{-\frac{1}{2}}\paren{\vbY_i-\vbmu\paren{d, \bX_i;\theta}}}^2,
\end{equation}
with $I_i(d)$ a consistency indicator function and $W_i$ being an estimation weight, both defined in Section \ref{c1:sec:est} of the main body. We observe the following result:
\begin{proposition}\label{c1:prop:unbiased}
    Given data from a clustered SMART satisfying all causal identifiability assumptions discussed in Appendix \ref{c1:sec:gen_smart:id_ass}, it holds that
    $$\scalebox{0.94}{$\displaystyle
    \E{\Sumd I(d) W \ltnorm{{\Sigma^d}^{-\frac{1}{2}}\paren{\vbY-\vbmu\paren{d, \bX;\theta}}}^2\mid \bX}
    = 
    \E{\Sumd \ltnorm{{\Sigma^d}^{-\frac{1}{2}}\paren{\vbY\d-\vbmu\paren{d, \bX;\theta}}}^2\mid \bX},
    $}
    $$
    with all parameters defined as above.
\end{proposition}
\begin{proof}

As shown in \citet{Murphy2001}, given the causal identifiability assumptions discussed in Appendix \ref{c1:sec:gen_smart:id_ass}, $I(d)W$ is a version of the Radon-Nikodym derivative between $P_{obs}$ and $P_{d}$, where $P_{obs}$ is the distribution of the observed data and $P_d$ is the distribution of the data in the population where all clusters follow the embedded cAI $d$. Thus, for any $d\in\mathcal{D}$,
$$
     \E{ I(d) W \ltnorm{{\Sigma^d}^{-\frac{1}{2}}\paren{\vbY-\vbmu\paren{d, \bX;\theta}}}^2\mid \bX}  =
     \E{\ltnorm{{\Sigma^d}^{-\frac{1}{2}}\paren{\vbY\d-\vbmu\paren{d, \bX;\theta}}}^2\mid \bX}.
$$
The result follows.
\end{proof}

With the above result in mind, the statistician wishing to estimate $\theta$ using Equation \ref{c1:eq:loss_full} can then turn to Equation \ref{c1:eq:loss}, which relies only on observed data. This loss function induces the estimating equation presented in the main body (Equation~(\ref{c1:eq:est_eq})), as shown below.

\begin{proposition}\label{c1:prop:min_to_ee}
    Let $\vct{\mu}$ be linear in $\theta$. Then 
    \begin{align*}
    \thetahat & = \argmin_{\theta\in\mathbb{R}^p}\SumiN\sum_{d\in\mathcal{D}}
I_i(d) W_i \ltnorm{{\Sigma_i^d}^{-\frac{1}{2}}\paren{\vbY_i-\vbmu\paren{d, \bX_i;\theta}}}^2,\\
    \shortintertext{if and only if}
    0  & = \SumiN\Sumd I_i(d) W_i D(d,\bX_i;\thetahat)^\top{\Sigma_i^d}\inv(\vbY_i - \vbmu(d,\bX_i;\thetahat)),
    \end{align*}
    where all parameters are as defined in Section \ref{c1:sec:est} of the main body and Appendix \ref{c1:sec:gen_smart}.
\end{proposition}
\begin{proof}

We wish to minimize\\
$\SumiN\sum_{d\in\mathcal{D}}
I_i(d) W_i \ltnorm{{\Sigma_i^d}^{-\frac{1}{2}}\paren{\vbY_i-\vbmu\paren{d, \bX_i;\theta}}}^2$
$$= 
\SumiN\sum_{d\in\mathcal{D}}
I_i(d) W_i \paren{\vbY_i-\vbmu\paren{d, \bX_i;\theta}}^\top{\Sigma_i^d}\inv \paren{\vbY_i-\vbmu\paren{d, \bX_i;\theta}}.$$

Let 
\begin{align*}
    \Tilde{\vbY} & := \tttvec{\vbY_1}{\vdots}{\vbY_N},\\
    \Tilde{\vbmu}\paren{d; \theta} & := \tttvec{\vbmu_1\paren{d, \bX_1; \theta}}{\vdots}{\vbmu_N\paren{d, \bX_N; \theta}},\\
    \Tilde{\boldsymbol\Sigma}^d & := 
    \begin{bmatrix}
        I_1(d)W_1{\Sigma_1^d}\inv & \dots & 0\\
        \vdots & \ddots & \vdots \\
        0 & \dots & I_N(d)W_N{\Sigma_N^d}\inv
    \end{bmatrix}.
\end{align*}
Then, we observe\\
$\SumiN\sum_{d\in\mathcal{D}}
I_i(d) W_i \ltnorm{{\Sigma_i^d}^{-\frac{1}{2}}\paren{\vbY_i-\vbmu\paren{d, \bX_i;\theta}}}^2$
\begin{align*}
     \lalign{ = 
\SumiN\sum_{d\in\mathcal{D}}
I_i(d) W_i \paren{\vbY_i - \vbmu\paren{d, \bX_i; \theta}}^\top {\Sigma_i^d}\inv \paren{\vbY_i-\vbmu\paren{d, \bX_i;\theta}}}\\
    & = \Sumd \paren{\Tilde{\vbY} - \Tilde{\vbmu}\paren{d; \theta}}^\top 
    \Tilde{\boldsymbol\Sigma}^d
    \paren{\Tilde{\vbY} - \Tilde{\vbmu}\paren{d; \theta}}.
    \shortintertext{Subsequently,}
    \shortintertext{$\partiald{\theta}
    \SumiN\sum_{d\in\mathcal{D}}
    I_i(d) W_i \ltnorm{{\Sigma_i^d}^{-\frac{1}{2}} \paren{\vbY_i-\vbmu\paren{d, \bX_i;\theta}}}^2 $}
    & = \partiald{\theta} \Sumd 
    \Tilde{\vbmu}\paren{d; \theta}^\top 
    \Tilde{\boldsymbol\Sigma}^d
    \Tilde{\vbmu}\paren{d; \theta}
    - 2 \Tilde{\vbmu}\paren{d; \theta}^\top 
    \Tilde{\boldsymbol\Sigma}^d
    \Tilde{\vbY}\\
    & = \Sumd
    2\partialD{\Tilde{\vbmu}\paren{d; \theta}}{\theta}^\top
    \Tilde{\boldsymbol\Sigma}^d 
    \Tilde{\vbmu}\paren{d; \theta}
    - 2\partialD{\Tilde{\vbmu}\paren{d; \theta}}{\theta}^\top
    \Tilde{\boldsymbol\Sigma}^d
    \Tilde{\vbY}.
\end{align*}
Therefore, 
$$\partiald{\theta}
    \SumiN\sum_{d\in\mathcal{D}}
    I_i(d) W_i \ltnorm{{\Sigma_i^d}^{-\frac{1}{2}} \paren{\vbY_i-\vbmu\paren{d, \bX_i;\theta}}}^2=0$$
if and only if
$$
    \Sumd
    \partialD{\Tilde{\vbmu}\paren{d; \theta}}{\theta}^\top
    \Tilde{\boldsymbol\Sigma}^d
    \paren{\Tilde{\vbY} - 
    \Tilde{\vbmu}\paren{d; \theta}} = 0.
$$

Recalling the construction of $\Tilde{\vbY}$, $\Tilde{\vbmu}\paren{d; \theta}$, and $\Tilde{\boldsymbol\Sigma}^d$, the above implies that
\begin{align*}
    \partiald{\theta}
    \SumiN\sum_{d\in\mathcal{D}}
    I_i(d) W_i \ltnorm{{\Sigma_i^d}^{-\frac{1}{2}} \paren{\vbY_i-\vbmu\paren{d, \bX_i;\theta}}}^2 
    & = 0
    \shortintertext{if and only if}
    \SumiN\sum_{d\in\mathcal{D}}
    I_i(d) W_i 
    \partialD{\vbmu \paren{d, \bX_i; \theta}}{\theta}^\top{\Sigma_i^d}\inv 
    \paren{\vbY_i-\vbmu\paren{d, \bX_i;\theta}} 
    & = 0.
\end{align*}

Since $\vct{\mu}$ is linear in $\theta$, $\SumiN\sum_{d\in\mathcal{D}}
    I_i(d) W_i \ltnorm{{\Sigma_i^d}^{-\frac{1}{2}} \paren{\vbY_i-\vbmu\paren{d, \bX_i;\theta}}}^2$ 
is convex and hence $\thetahat$ minimizes
$\SumiN\sum_{d\in\mathcal{D}}
    I_i(d) W_i \ltnorm{{\Sigma_i^d}^{-\frac{1}{2}} \paren{\vbY_i-\vbmu\paren{d, \bX_i;\theta}}}^2$ 
if and only if
$$\SumiN\Sumd I_i(d) W_i D(d,\bX_i;\thetahat)^\top {\Sigma_i^d}\inv(\vbY_i - \vbmu(d,\bX_i;\thetahat))=0,$$ as required.

\end{proof}

\begin{corollary}\label{c1:cor:ee_centered}
    Given data from a clustered SMART satisfying all causal identifiability assumptions discussed in Appendix \ref{c1:sec:gen_smart:id_ass}, if $\vbmu$ is correctly specified (as formalized below in Assumption \ref{c1:ass_stat:mm_spec}), it holds that
    $$\E{\Sumd I(d) W D(d,\bX;\thetahat)^\top{\Sigma^d}\inv(\vbY - \vbmu(d,\bX;\theta_0))\mid \bX} = 0,$$
    with all parameters defined as above.
\end{corollary}

\begin{proof}
    Well, by argument analogous to the proof of Proposition \ref{c1:prop:unbiased}, it holds that, for any $d\in\mathcal{D}$\\
    $\E{I(d) W D(d,\bX;\thetahat)^\top {\Sigma^d}\inv(\vbY - \vbmu(d,\bX;\theta_0))\mid \bX}$
    \begin{align*}
         \lalign{ = \E{ D(d,\bX;\thetahat)^\top{\Sigma^d}\inv(\vbY\d - \vbmu(d,\bX;\theta_0))\mid \bX}}\\
        \shortintertext{(assuming $\vbmu$ linear in $\theta$)}
        & = \E{ D(d,\bX)^\top{\Sigma^d}\inv(\vbY\d - \vbmu(d,\bX;\theta_0))\mid \bX}\\
        & = D(d,\bX)^\top{\Sigma^d}\inv\E{\vbY\d - \vbmu(d,\bX;\theta_0)\mid \bX}\\
        \shortintertext{(by Assumption \ref{c1:ass_stat:mm_spec})}
        & = 0.
    \end{align*}
\end{proof}

\subsection{Asymptotic Distribution}

As discussed in the main body of the report, we consider 
$$\underbrace{
U\paren{Z_i, \bX_i, \vbY_i; \theta, \alpha, \pi}}_{=:U_{i\mid \theta, \alpha, \pi}} =
\Sumd I_i(d) W_{i}(\pi) D(d,\bX_i;\theta)^\top V^d(\bX_i;\alpha)\inv(\vbY_i - \vbmu(d,\bX_i;\theta)),$$ 
where $Z_i:=\paren{\Bar{A}_{K,i}, \Bar{R}_{K-1,i}}$, and take $\thetahat = \thetahat_n\paren{\set{\bX_i, \vbY_i, Z_i}; \alpha, \pi}$ to be the solution to the estimating equation $0 = \SumiN U_{i\mid \theta, \alpha, \pi}$. Refer to Section \ref{c1:sec:est} of the main body for definitions of the component functions.

In the subsequent sections, we will establish several results regarding the asymptotic properties of $\thetahat$. In doing so, we rely on the assumptions discussed below.
\subsubsection{Statistical Assumptions}\label{c1:sec:asy_theory:ass}
\begin{enumerate}[label=(\roman*)]
    \item\label{c1:ass_stat:alpha_conv} There exists $\alpha_+$ such that $\sqrt{N}\paren{\alphahat-\alpha_+}=O_p(1)$.
    \item\label{c1:ass_stat:mm_spec} The marginal structural model for $Y$ is correctly specified; i.e., there exists a $\theta_0$ such that 
    $\E{\vbY\d\mid\bX}=\vbmu(d,\bX;\theta_0)$.
    \item\label{c1:ass_stat:unique_soln} If $\theta^*$ solves $\E{ \U{\theta, \alphahat, \pi_0}\mid \bX}=0$, then $\theta^*=\theta_0$.
    \item\label{c1:ass_stat:finite_u} $\E{\lnorm{U_f}{\infty}}<\infty$, where $U_f(\theta) := \U{\theta, \alphahat, \pi_0}$. 
    \item\label{c1:ass_stat:sigma_wb} For all $i=1\ddd N$, $V_i^d:=V^d(\bX_i;\alpha)$ is continuously differentiable in $\alpha$, with $\partialD{V_i^d}{\alpha}$ bounded. Furthermore, there exists a $c>0$ and a neighborhood of $\alpha_+$, $\mathcal{U}_{\alpha_+, c}$, such that, for all $\alpha\in\mathcal{U}_{\alpha_+, c}$,  $\lambda>c$ for every eigenvalue $\lambda$ of $V_i^d(\alpha)$. Further assume that $\partialD{\alphahat}{\theta}$ is bounded. We note that this is analogous to assuming invertible working variance matrix estimates.
    \item\label{c1:ass_stat:score} $\pihat$ is estimated using maximum likelihood, with continuously differentiable (in $\pi$) score function. Let $S_{\pi,i}$ denote the cluster-level score contribution for Cluster $i$, and let $S_\pi$ denote a generic cluster-level score contribution. Further assume that there exists a $c'>0$ and a neighborhood $\mathcal{U}'_{c'}$ of $(\theta_0, \alpha_0, \pi_0)$ such that $\E{U_{\theta, \alpha, \pi}S^\top_\pi}$ is finite in this neighborhood and 
    $\lambda>c'$ for any eigenvalue, $\lambda$, of $\E{S_\pi S_\pi^\top}$ for any $(\theta, \alpha, \pi)\in\mathcal{U}'$. As before, we note that this is analogous to assuming $\E{S_\pi S_\pi^\top}$ (as well as its plug-in counterpart) is positive definite.
    \item\label{c1:ass_stat:iid} $\paren{\Ui{\theta_0,\alpha_+,\pi_0},S_{\pi_0}}^N_{i=1}$ are independent with identical first and second moments.
    \item\label{c1:ass_stat:lin}
    $\vbmu$ is linear in $\theta$. I.e., $D(d, \bX; \theta) = D(d,\bX)$.
\end{enumerate}

As with $\theta$, we let $\pi_0$ denote the true parameter value for $\pi$. Note that, due to the randomized structure of a SMART, $\pi_0$ is known.\footnote{While the analyst may estimate weights to improve efficiency, the true probabilities are known, and a consistent weight estimation approach will have $\pihat\xrightarrow[N\to\infty]{\mathbb{P}}\pi_0$.}
For brevity, let $U_{\theta, \alpha, \pi}:=U\paren{Z, \bX, \vbY; \theta, \alpha, \pi}$ for a random cluster (with $U:=U_{\theta_0, \alpha_+, \pi_0}$).

\subsubsection{Lemmas and Proofs}
As stated above, these calculations follow standard statistical arguments. However, we include them for any reader interested in the technical details establishing the asymptotic distribution of the estimator presented in the main body.

\begin{lemma}
\label{c1:lemm:theta_con} 
    Let $\thetahat = \thetahat_N$ be a solution to 
    $$0 = \SumiN U_{i\mid \theta, \alpha_+, \pi_0},$$ 
    then $\theta_0$ is a zero of $\mathbf{\psi}(\theta) := \E{\frac{1}{N}\SumiN U_{i\mid \theta, \alpha_+, \pi_0}}$. 
    Moreover, if $\theta_0$ is the \textit{unique} root of $\mathbf{\psi}$, then $\thetahat$ is a consistent estimator for $\theta_0$.
\end{lemma}

\begin{proof}
Let $\mathbf{\psi}(\theta) := \E{ U_{\theta, \alpha_+, \pi_0}}$ and $\mathbf{\psi}_n(\theta) := \frac{1}{N}\SumiN U_{i\mid \theta, \alpha_+, \pi_0}.$ 

Corollary \ref{c1:cor:ee_centered} guarantees 
$\E{U_{\theta, \alpha_+, \pi_0}} = 0$.
Therefore, $\theta_0$ is indeed a zero of $\mathbf{\psi}(\theta).$
    
Via Statistical Assumption \ref{c1:ass_stat:finite_u} and the Weak Law of Large Numbers for Random Functions,\footnote{See Theorem 9.2 in \citet{Keener2010}.} we observe 
$\lnorm{\mathbf{\psi}_n(\theta) - \mathbf{\psi}(\theta)}{\infty}\xrightarrow[N\to\infty]{\mathbb{P}} 0.$

Employing classical regularity assumptions and results of $M$-estimators,\footnote{See, for example, Theorem 9.4 in \citet{Keener2010}.} we conclude
$$\thetahat_N \xrightarrow[N\to\infty]{\mathbb{P}} \theta_0.$$

\end{proof}

\begin{lemma}
\label{c1:ScoreDerivs}
Suppose we use inverse probability weights of the form discussed in Appendix \ref{c1:sec:gen_smart:est}, $W_i(\pihat) = \dfrac{1}{p_{\pihat}\paren{\Bar{a}_{K,i}; \Bar{r}_{K-1,i},\bh_{K,i}}}$, where $\pihat$ is obtained via maximum likelihood estimation, $p_{\pihat}\paren{\Bar{a}_{K,i}; \Bar{r}_{K-1,i},\bh_{K,i}}$ is differentiable with respect to $\pi$ and $S_{\pi,i}$ denotes cluster-level score contribution for Cluster $i$. Then
$$\frac{1}{N}\SumiN \partialD{\Ui{\theta_0, \alpha_+, \pi_0}}{\pi} \sqrt{N}\paren{\pihat-\pi_0} = -\frac{1}{\sqrt{N}}\SumiN\E{US_{\pi_0}^\top}\E{S_{\pi_0}S_{\pi_0}^\top}\inv S_{\pi_0,i} + o_p(1).$$
\end{lemma}

\begin{proof}
Let $p_{\pihat,i}\paren{\Bar{a}_K}:=p_{\pihat}\paren{\Bar{a}_{K}; \Bar{r}_{K-1,i},\bh_{K,i}}$. Further let $\mathcal{Z}_K:=\paren{\Bar{\mathcal{A}}_K,\Bar{\mathcal{R}}_{K-1}}$ denote the set of possible treatment-response pathways in the SMART.

We observe the following likelihood for the observed treatment assignments: 
$$L\paren{\pi; \set{\paren{\Bar{z}_{K,i}, \bh_{K,i}}}_{i=1}^N} = \ProdiN\displaystyle\prod_{z_K\in\mathcal{Z}_K} \paren{p_{\pi,i}\paren{\Bar{a}_K}}^{\Ind{Z_{K,i}=z_K}}.$$

This induces the cluster-level score contribution
$$
S_{\pi,i}:= \sum_{z_K\in\mathcal{Z}_K} \frac{\Ind{Z_{K,i}=z_K}}{p_{\pi,i}\paren{\Bar{a}_K}}\paren{\partialD{p_{\pi,i}\paren{\Bar{a}_K}}{\pi}}^\top
$$
Let
\begin{align*}
\mathcal S_{N,\pi}(\pi)
:=
\left(\frac{\partial l(\pi)}{\partial \pi}\right)^\top
=
\sum_{i=1}^N S_{\pi,i}(\pi)
\end{align*}
denote the total sample score. The maximum likelihood estimator $\hat\pi$ is chosen such that
$\mathcal{S}_{N,\pi}(\hat{\pi})=0$.

We recall that 
\begin{align*}
    \Ui{\theta_0, \alpha_+, \pi_0} & := \Sumd I_i(d) W_{i}(\pi_0) D(d,\bX_i;\theta_0)^\top{V_i^d}(\alpha_+)\inv(\vbY_i - \vbmu(d,\bX_i;\theta_0))\\
    & = \Sumd I_i(d) {D(d,\bX_i;\theta_0)^\top{V_i^d}(\alpha_+)\inv(\vbY_i - \vbmu(d,\bX_i;\theta_0))}\dfrac{1}{p_{\pi_0,i}\paren{\Bar{a}_{K,i}}}.\\
    \shortintertext{Thus,}
    \partialD{\Ui{\theta_0, \alpha_+, \pi_0}}{\pi}  & = \partialD{\Ui{\theta_0, \alpha_+, \pi_0}}{W_i} \partialD{W_i(\pi_0)}{\pi}\\
    & = \paren{\Sumd I_i(d)D(d,\bX_i;\theta_0)^\top{V_i^d}(\alpha_+)\inv(\vbY_i - \vbmu(d,\bX_i;\theta_0))}\\
    & \hspace{2cm}\times\paren{-\dfrac{1}{p_{\pi_0,i}\paren{\Bar{a}_{K,i}}^2}}\paren{\partialD{p_{\pi_0,i}\paren{\Bar{a}_{K,i}}}{\pi}}\\
    & = -\Ui{\theta_0, \alpha_+, \pi_0}\paren{\dfrac{1}{p_{\pi_0,i}\paren{\Bar{a}_{K,i}}}\partialD{p_{\pi_0,i}\paren{\Bar{a}_{K,i}}}{\pi}}\\
    & = -\Ui{\theta_0, \alpha_+, \pi_0}
    \displaystyle\sum_{z_{K}\in\mathcal{Z}_K}\dfrac{\Ind{Z_{K,i}=z_{K}}}{p_{\pi_0,i}\paren{\Bar{a}_{K,i}}}\partialD{p_{\pi_0,i}\paren{\Bar{a}_{K,i}}}{\pi}\\
    & = - \Ui{\theta_0, \alpha_+, \pi_0}S_{\pi_0, i}^\top.
\end{align*}
Thus,
\begin{align*}
    \frac{1}{N}\SumiN \partialD{\Ui{\theta_0, \alpha_+, \pi_0}}{\pi} \sqrt{N}\paren{\pihat-\pi_0} & = \paren{\frac{1}{N}\SumiN - \Ui{\theta_0, \alpha_+, \pi_0}S_{\pi_0, i}^\top}\sqrt{N}\paren{\pihat-\pi_0}\\
    & = \paren{\E{- US_{\pi_0}^\top}+ o_p(1)}\sqrt{N}\paren{\pihat-\pi_0} \\
    & = -\E{US_{\pi_0}^\top}\sqrt{N}\paren{\pihat-\pi_0} + o_p(1).
\end{align*}

We proceed via Taylor expansion of the score function:
\begin{align*}
    \frac{1}{N}\underbrace{\SumiN S_{\pihat, i}}_{=:\mathcal{S}_{N,\pi}(\hat{\pi})=0} & =  \frac{1}{N}\SumiN S_{\pi_0, i} + \frac{1}{N}\SumiN\paren{\partiald{\pi}S_{\Tilde{\pi}, i}}(\pihat-\pi_0),\\
    \shortintertext{for some $\Tilde{\pi}\in\set{p\pihat + (1-p)\pi_0\mid p\in[0,1]}$. Thus (by Assumption \ref{c1:ass_stat:score}),}
    - (\pihat-\pi_0) & = 
     \frac{1}{N}\SumiN \paren{\frac{1}{N}\displaystyle\sum_{j=1}^N\partiald{\pi}S_{\Tilde{\pi}, j}}\inv S_{\pi_0, i}\\
      & = 
     \frac{1}{N}\SumiN \paren{\E{\partiald{\pi}S_{\pi_0}}\inv + \Opsn}S_{\pi_0, i}\\
     & = \frac{1}{N}\SumiN \E{\partiald{\pi}S_{\pi_0}}\inv S_{\pi_0, i} + \opsn.
     \shortintertext{With the last equality holding since $\E{S_{\pi_0}}=0$. Noting that $S_{\pi_0}$ represents the derivative of log-likelihood of treatment assignment, we substitute the outer product as below:}
    & = -\frac{1}{N}\SumiN \E{S_{\pi_0}S_{\pi_0}^\top}\inv S_{\pi_0, i} + \opsn.\\
    \shortintertext{Consequently,}
    \sqrt{N}(\pihat-\pi_0) & = \paren{\frac{1}{\sqrt{N}}\SumiN \E{S_{\pi_0}S_{\pi_0}^\top}\inv S_{\pi_0, i}} + o_p(1).
\end{align*}

Putting the above together, we observe:
\begin{align*}
     \frac{1}{N}\SumiN \partialD{\Ui{\theta_0, \alpha_+, \pi_0}}{\pi} \sqrt{N}\paren{\pihat-\pi_0} & = -\E{US_{\pi_0}^\top}\sqrt{N}\paren{\pihat-\pi_0} + o_p(1)\\
     & = -\frac{1}{\sqrt{N}}\SumiN \E{US_{\pi_0}^\top}\E{S_{\pi_0}S_{\pi_0}^\top}\inv S_{\pi_0, i} + o_p(1).
\end{align*}

\end{proof}


\begin{lemma}
\label{c1:TaylorExp} Under the assumptions above,
    $$
    \scalebox{0.95}{$\displaystyle
    \sqrt{N}\paren{\thetahat-\theta_0} = -\E{\partialD{U}{\theta}}\inv\paren{\frac{1}{\sqrt{N}}\SumiN \Ui{\theta_0, \alpha_+, \pi_0}-\E{US_{\pi_0}^\top}\E{S_{\pi_0}S_{\pi_0}^\top}\inv S_{\pi_0,i}} + o_p(1).
    $}
    $$
\end{lemma}

\begin{proof}
We begin via Taylor expansion. We see that
\begin{align*}
    \frac{1}{N}\underbrace{\SumiN \Ui{\thetahat, \alphahat, \pihat}}_{=0} & =
    \frac{1}{N}\SumiN \Ui{\theta_0, \alphahat, \pihat} + 
    \frac{1}{N}\SumiN \paren{\partialD{\Ui{\Tilde{\theta}, \alphahat, \pihat}}{\theta} + \partialD{\Ui{\Tilde{\theta}, \alphahat, \pihat}}{\alphahat} \partialD{\alphahat}{\theta}
    }\paren{\thetahat-\theta_0},\\
    \frac{1}{N}\SumiN \Ui{\theta_0, \alphahat, \pihat} & = \frac{1}{N}\SumiN \Ui{\theta_0, \alpha_+, \pihat} + 
    \frac{1}{N}\SumiN \partialD{\Ui{\theta_0, \Tilde{\alpha}, \pihat}}{\alpha} \paren{\alphahat-\alpha_+},\\
    \frac{1}{N}\SumiN \Ui{\theta_0, \alpha_+, \pihat} &
    =  \frac{1}{N}\SumiN \Ui{\theta_0, \alpha_+, \pi_0} + \frac{1}{N}\SumiN  \partiald{\pi} \Ui{\theta_0, \alpha_+, \Tilde{\pi}} \paren{\pihat-\pi_0},
\end{align*}
for some $\Tilde{\theta}, \Tilde{\alpha}, \Tilde{\pi}$ in $\{p\thetahat + (1-p)\theta_0\mid p\in[0,1]\}$, $\{p\alphahat + (1-p)\alpha_+\mid p\in[0,1]\}$, $\{p\pihat + (1-p)\pi_0\mid p\in[0,1]\}$ (respectively).

Note that, by Assumption \ref{c1:ass_stat:alpha_conv}, we have
\begin{align*}
    \frac{1}{N}\SumiN \partialD{\Ui{\theta_0, \Tilde{\alpha}, \pihat}}{\alpha} \paren{\alphahat-\alpha_+} & = \frac{1}{N}\SumiN \paren{\partialD{\Ui{\theta_0, \alpha_+, \pihat}}{\alpha} + \Opsn} \paren{\alphahat-\alpha_+}\\
    & = \frac{1}{N}\SumiN \partialD{\Ui{\theta_0, \alpha_+, \pihat}}{\alpha} \paren{\alphahat-\alpha_+} + \opsn.
\end{align*}
Similar logic shows 
$$\dfrac{1}{N}\SumiN  \partialD{\Ui{\theta_0, \alpha_+, \Tilde{\pi}}}{\pi}  \paren{\pihat-\pi_0} = \dfrac{1}{N}\SumiN  \partialD{\Ui{\theta_0, \alpha_+, \pi_0}}{\pi}  \paren{\pihat-\pi_0} + \opsn.$$

Therefore, \\
$-\frac{1}{N}\SumiN \paren{\partialD{\Ui{\Tilde{\theta}, \alphahat, \pihat}}{\theta}+
    \partialD{\Ui{\Tilde{\theta}, \alphahat, \pihat}}{\alphahat}\partialD{\alphahat}{\theta}}\paren{\thetahat-\theta_0}$
\begin{align*}
    \lalign{ = \frac{1}{N}\SumiN \Ui{\theta_0, \alpha_+, \pi_0}
    + \frac{1}{N}\SumiN \partialD{\Ui{\theta_0, \alpha_+, \pihat}}{\alpha}  \paren{\alphahat-\alpha_+}
    } \\
    &  \ralign{ 
     + \frac{1}{N}\SumiN \partialD{\Ui{\theta_0, \alpha_+, \pi_0}}{\pi}  \paren{\pihat-\pi_0} + \opsn.}
\end{align*}

We first note that 
$$\frac{1}{N}\SumiN \partialD{\Ui{\theta_0, \alpha_+, \pihat}}{\alpha}\paren{\alphahat-\alpha_+} = \paren{\E{\partialD{U}{\alpha}}+\Opsn}\paren{\alphahat-\alpha_+}.$$
Applying arguments analogous to those in the proof of Proposition~\ref{c1:prop:unbiased}, we observe that Assumption~\ref{c1:ass_stat:mm_spec} implies $\E{\partialD{U}{\alpha}}=0$.
Therefore, Assumption \ref{c1:ass_stat:alpha_conv} gives us
$$\frac{1}{N}\SumiN \partialD{\Ui{\theta_0, \alpha_+, \pihat}}{\alpha}\paren{\alphahat-\alpha_+} = \opsn.$$

By similar logic, we also note that $\frac{1}{N}\SumiN\partialD{\Ui{\theta, \alphahat, \pihat}}{\alphahat}=o_p(1)$. Combining this with Assumption \ref{c1:ass_stat:sigma_wb} gives $\frac{1}{N}\SumiN\partialD{\Ui{\Tilde{\theta}, \alphahat, \pihat}}{\alphahat}\partialD{\alphahat}{\theta} = o_p(1)$.

Moving forward with our Taylor expansion of $\frac{1}{N}\SumiN \Ui{\thetahat, \alphahat, \pihat},$ we see\\
$-\paren{\frac{1}{N}\SumiN \partialD{\Ui{\Tilde{\theta}, \alphahat, \pihat}}{\theta}+
    o_p(1)}\paren{\thetahat-\theta_0}$
$$ = \frac{1}{N}\SumiN \Ui{\theta_0, \alpha_+, \pi_0}  + \frac{1}{N}\SumiN\partialD{\Ui{\theta_0, \alpha_+, \pi_0}}{\pi} \paren{\pihat-\pi_0} + \opsn.$$
Noting the continuity of $\partialD{U}{\theta}$ and the convergence of $\thetahat$, $\alphahat,$ and $\pihat$, we observe
\begin{align*}
    \lalign{\sqrt{N}\paren{\thetahat-\theta_0} = -\paren{ \frac{1}{N} \SumiN \partialD{\Ui{\theta_0, \alpha_+, \pi_0}}{\theta} + o_p(1)}\inv\Bigg(\frac{1}{\sqrt{N}}\SumiN \Ui{\theta_0, \alpha_+, \pi_0} }\\
    & \ralign{+ \frac{1}{N}\SumiN\partialD{\Ui{\theta_0, \alpha_+, \pi_0}}{\pi}\sqrt{N}\paren{\pihat-\pi_0}\Bigg) 
    + o_p(1).}
\end{align*} 

Using Lemma \ref{c1:ScoreDerivs}, we conclude $-\sqrt{N}\paren{\thetahat-\theta_0}$
$$ = \paren{\E{\partialD{U}{\theta}}+o_p(1)}\inv
\paren{\frac{1}{\sqrt{N}}\SumiN \Ui{\theta_0, \alpha_+, \pi_0} -\E{US_{\pi_0}^\top}\E{S_{\pi_0}S_{\pi_0}^\top}\inv S_{\pi_0,i}} + o_p(1).$$

\end{proof}

\subsubsection{Core Results}

\begin{theorem}\label{c1:thm:asydist_wgt}
Under the assumptions listed in Appendix \ref{c1:sec:asy_theory:ass} above, then $$\sqrt{N}\paren{\thetahat-\theta_0}\xrightarrow[N\to\infty]{\mathcal{L}}\N{0}{J\inv Q J\inv},$$
where
\begin{itemize}
    \item $J = \E{\Sumd I(d) W(\pi_0)D(d,\bX)^\top V^d(\alpha_+)\inv D(d,\bX)}$,
    \item $Q = \E{UU^\top} - \E{US_{\pi_0}^\top}\E{S_{\pi_0}S_{\pi_0}^\top}\inv\E{S_{\pi_0}U^\top}.$
\end{itemize}
\end{theorem}

\begin{proof}
We first note that\\
$\E{U - \E{US_{\pi_0}^\top}\E{S_{\pi_0}S_{\pi_0}^\top}\inv S_{\pi_0}}$
\begin{align*}
     \lalign{ = \E{U} - \E{US_{\pi_0}^\top}\E{S_{\pi_0}S_{\pi_0}^\top}\inv \E{S_{\pi_0}}}\\
    & = 0
    \text{, and}
    \shortintertext{$\V{U - \E{US_{\pi_0}^\top}\E{S_{\pi_0}S_{\pi_0}^\top}\inv S_{\pi_0}}$} 
    & = \E{\paren{U - \E{US_{\pi_0}^\top}\E{S_{\pi_0}S_{\pi_0}^\top}\inv S_{\pi_0}}\paren{U  - \E{US_{\pi_0}^\top}\E{S_{\pi_0}S_{\pi_0}^\top}\inv S_{\pi_0}}^\top} \\
    & = \E{UU^\top} - \E{US_{\pi_0}^\top}\E{S_{\pi_0}S_{\pi_0}^\top}\inv\E{S_{\pi_0}U^\top} \\
    & \hspace{1cm} - 
    \E{US_{\pi_0}^\top}\E{S_{\pi_0}S_{\pi_0}^\top}\inv \E{S_{\pi_0}U^\top}\\
    & \hspace{1cm} + \E{US_{\pi_0}^\top}\E{S_{\pi_0}S_{\pi_0}^\top}\inv \E{S_{\pi_0}S_{\pi_0}^\top}\E{S_{\pi_0}S_{\pi_0}^\top}\inv\E{S_{\pi_0}U^\top}\\
    & = \E{UU^\top} - \E{US_{\pi_0}^\top}\E{S_{\pi_0}S_{\pi_0}^\top}\inv\E{S_{\pi_0}U^\top}\\
    & = Q.
\end{align*}

We employ the Central Limit Theorem to conclude
$$\dfrac{1}{\sqrt{N}} \SumiN \left[\Ui{\theta_0, \alpha_+, \pi_0} 
 - \E{US_{\pi_0}^\top}\E{S_{\pi_0}S_{\pi_0}^\top}\inv S_{\pi_0, i} \right]
\xrightarrow[N\to\infty]{\mathcal{L}}\N{0}{Q}.$$

We combine the above result, Lemma \ref{c1:TaylorExp}, and Slutsky's Lemma to prove 
$$\sqrt{N}\paren{\thetahat-\theta_0}\xrightarrow[N\to\infty]{\mathcal{L}} \paren{\E{\partialD{U}{\theta}}}\inv C,$$
where $C\sim\N{0}{Q}.$

Under Assumptions \ref{c1:ass_stat:mm_spec} and \ref{c1:ass_stat:lin},
\begin{align*}
    - \E{\partialD{U}{\theta}} & = \E{\Sumd I(d) W(\pi_0)D(d,\bX)^\top V^d(\alpha_+)\inv D(d,\bX)}\\
    & =: J.
\end{align*}

Therefore, applying the Delta Method gives us our desired result
$$\sqrt{N}\paren{\thetahat-\theta_0}\xrightarrow[N\to\infty]{\mathcal{L}}\N{0}{J\inv Q J\inv}.$$

\end{proof}

We note that the asymptotic distribution under known weights (presented in Section \ref{c1:sec:est:asy_dist} of the main body) is an immediate corollary of Theorem \ref{c1:thm:asydist_wgt}.

\addtocontents{toc}{\protect\setcounter{tocdepth}{2}}


\bibliographystyle{apalike} 
\bibliography{Bibliography.bib}

\end{document}

%% file: Packages.tex
\usepackage{hyperref, amsmath, amssymb, mathtools, amsthm}
\usepackage{bbm, dsfont}
\usepackage{verbatim}
\usepackage{graphicx}
\usepackage[utf8]{inputenc}
\usepackage{shellesc}
\usepackage{pdfpages}
\usepackage{float}
\usepackage{titling}
\usepackage{makecell}

\usepackage{enumitem}
\usepackage{ifthen}
\usepackage{caption, algorithm, algpseudocode}
\usepackage{adjustbox}
\usepackage{multirow}
\usepackage{threeparttable}
\usepackage{array}
\usepackage{booktabs}
\usepackage{natbib}[usenatbib]
\usepackage{pdflscape}

\usepackage{subcaption}

\usepackage{lipsum}
\makeatletter
\newenvironment{breakablealgorithm}
  {
   \begin{center}
     \refstepcounter{algorithm}
     \hrule height.8pt depth0pt \kern2pt
     \renewcommand{\caption}[2][\relax]{
       {\raggedright\textbf{\fname@algorithm~\thealgorithm} ##2\par}%
       \ifx\relax##1\relax 
         \addcontentsline{loa}{algorithm}{\protect\numberline{\thealgorithm}##2}%
       \else 
         \addcontentsline{loa}{algorithm}{\protect\numberline{\thealgorithm}##1}%
       \fi
       \kern2pt\hrule\kern2pt
     }
  }{
     \kern2pt\hrule\relax
   \end{center}
  }
\makeatother

%% file: Commands_and_Functions.tex
\DeclareMathOperator*{\argmin}{arg\,min}

\makeatletter
\newcommand{\AlignFootnote}[1]{%
  \ifmeasuring@
  \else
    \iffirstchoice@
      \footnote{#1}%
    \fi
  \fi}
\makeatother

\newtheorem{theorem}{Theorem}
\newtheorem{lemma}{Lemma}
\newtheorem{corollary}{Corollary}
\newtheorem{proposition}{Proposition}

\newcommand{\vct}[1]{#1}
\newcommand\ddfrac[2]{\frac{\displaystyle #1}{\displaystyle #2}}

\newcommand\Prb[1]{\mathbb{P}\left[#1\right]}
\newcommand\E[1]{\mathbb{E}\left[#1\right]}

\newcommand\V[1]{\mathbb{V}\left[#1\right]}
\newcommand\paren[1]{\left( #1 \right)}

\newcommand\partiald[1]{\frac{\partial}{\partial #1}}
\newcommand\partialD[2]{\dfrac{\partial #1}{\partial #2}}
\newcommand\Sum[2]{\displaystyle\sum_{#1=1}^{#2}}
\newcommand\Sumz[2]{\displaystyle\sum_{#1=0}^{#2}}

\newcommand\quotes[1]{``#1''}

\newcommand{\indep}{\perp \!\!\! \perp}

\newcommand\Cov[1]{\text{Cov}\left(#1\right)}

\newcommand\ltnorm[1]{\lVert #1 \rVert_2}

\newcommand\lnorm[2]{\lVert #1 \rVert_{#2}}
\newcommand\Ind[1]{\mathbbm{1}_{#1}}

\newcommand\ttttvec[4]{\begin{bmatrix} #1 \\ #2 \\ #3 \\ #4 \end{bmatrix}}
\newcommand\tttvec[3]{\begin{bmatrix} #1 \\ #2 \\ #3 \end{bmatrix}}

\newcommand{\tttvechTinline}[3]{%
    \bigl[#1 #2 #3\bigr]^\top%
}

\newcommand\ttvec[2]{\begin{bmatrix} #1 \\ #2  \end{bmatrix}}

\newcommand\N[2]{\mathcal{N}\left(#1,\s #2\right)}

\newcommand\set[1]{\left\{#1\right\}}

\global\long\def\sigmahat{\hat{\sigma}}
\global\long\def\rhohat{\hat{\rho}}

\global\long\def\thetahat{\hat{\theta}}
\global\long\def\pihat{\hat{\pi}}
\global\long\def\alphahat{\hat{\alpha}}

\global\long\def\oni{\vct{\mathbf{1}}_{n_i}}
\global\long\def\on{\vct{\mathbf{1}}_{n}}
\global\long\def\Rho{\mathrm{P}}

\global\long\def\d{^{(d)}}
\global\long\def\dpr{^{(d')}}

\global\long\def\s{\text{ }}

\global\long\def\SumiN{\displaystyle\sum_{i=1}^N}
\global\long\def\ProdiN{\displaystyle\prod_{i=1}^N}

\global\long\def\Sumd{\displaystyle\sum_{d\in\mathcal{D}}}

\global\long\def\d{^{(d)}}

\global\long\def\ddd{,\ldots,}%
\global\long\def\poh{^{\frac{1}{2}}}%
\global\long\def\bX{\mathbf{X}}%
\global\long\def\bS{\mathbf{S}}%
\global\long\def\bs{\mathbf{s}}%
\global\long\def\bx{\mathbf{x}}%
\global\long\def\vbY{\vct{\mathbf{Y}}}
\global\long\def\vbmu{\vct{\boldsymbol\mu}}
\global\long\def\bY{\mathbf{Y}}%
\global\long\def\bh{\mathbf{h}}%

\global\long\def\muhat{\hat{\mu}}
\global\long\def\d{{^{(d)}}}

\global\long\def\inv{^{-1}}

\global\long\def\s{\text{ }}

\newcommand{\YY}[2][NULL]{\ifthenelse{\equal{#1}{NULL}}{\vct{Y}_{#2}}{\paren{\vct{Y}_{#2}}_{#1}}}
\newcommand{\ee}[2][NULL]{\ifthenelse{\equal{#1}{NULL}}{\vct{\varepsilon}_{#2}}{\paren{\vct{\varepsilon}_{#2}}_{#1}}}

\newcommand{\eex}[2][NULL]{\ifthenelse{\equal{#1}{NULL}}{\vct{\varepsilon}_{#2}^{(r)}}{\paren{\vct{\varepsilon}_{#2}^{(r)}}_{#1}}}

\newcommand{\px}[2][NULL]{\ifthenelse{\equal{#1}{NULL}}{P_{#2}^{[r]}}{P_{#2}^{[#1]}}}
\newcommand{\deltax}[2][NULL]{\ifthenelse{\equal{#1}{NULL}}{\delta_{#2}^{(r)}}{\delta_{#2}^{(#1)}}}
\newcommand{\epsilonx}[2][NULL]{\ifthenelse{\equal{#1}{NULL}}{\epsilon_{#2}^{<r>}}{\epsilon_{#2}^{<#1>}}}

\newcommand{\cx}[1]{c^{<r>}_{#1}}
\newcommand{\sx}[1]{s^{<r>}_{#1}}
\newcommand{\ssx}[1]{s^{2<r>}_{#1}}

\global\long\def\sr{^{<r>}}
\global\long\def\srt{^{2<r>}}
\global\long\def\sR{^{<1>}}
\global\long\def\sRt{^{2<1>}}
\global\long\def\sNR{^{<0>}}
\global\long\def\sNRt{^{2<0>}}
\global\long\def\di{^{(d_i)}}

\global\long\def\sss{n}
\global\long\def\on{\vct{\mathbf{1}}_{\sss}}

\global\long\def\Spi{S_\pi}%
\newcommand{\U}[1]{U_{#1}}
\newcommand{\Ui}[1]{U_{i \mid #1}}
\global\long\def\opsn{o_p\paren{\sqrt{N}^{-1}}}
\global\long\def\Opsn{O_p\paren{\sqrt{N}^{-1}}}

\NewDocumentCommand{\lalign}{ m O{1cm}}{%
    \hspace{#2} &  \mathmakebox[\displaywidth-#2-.35cm][l]{\hspace{.1cm} #1
    \hfill}
}
\newcommand{\ralign}[1]{
  &\mathllap{\hfill #1}%
}

%% file: Bibliography.bib
@article{NeCamp2017,
  doi = {10.1177/0962280217708654},
  url = {https://doi.org/10.1177/0962280217708654},
  year = {2017},
  month = "Jun",
  publisher = {{SAGE} Publications},
  volume = {26},
  number = {4},
  pages = {1572--1589},
  author = {NeCamp, Timothy and  Kilbourne, Amy and Almirall, Daniel},
  title = {Comparing Cluster-Level Dynamic Treatment Regimens using Sequential, Multiple Assignment, Randomized Trials: Regression Estimation and Sample Size Considerations},
  journal = {Statistical Methods in Medical Research}
}

@article{Seewald2019,
 author = {Seewald, Nicholas J and Kidwell, Kelley M and Nahum-Shani, Inbal and Wu, Tianshuang and McKay, James R and Almirall, Daniel},
 copyright = {All rights reserved},
 doi = {10.1177/0962280219877520},
 ids = {seewaldSampleSizeConsiderations},
 issn = {0962-2802},
 journal = {Statistical Methods in Medical Research},
 number = {7},
 pages = {1891--1912},
 title = {Sample Size Considerations for Comparing Dynamic Treatment Regimens in a Sequential Multiple-Assignment Randomized Trial with a Continuous Longitudinal Outcome},
 volume = {29},
 year = {2020}
}

@article{Lu2015,
  doi = {10.1002/sim.6819},
  url = {https://doi.org/10.1002/sim.6819},
  year = {2015},
  month = "Dec",
  publisher = {Wiley},
  volume = {35},
  number = {10},
  pages = {1595--1615},
  author = {Lu, Xi and Nahum-Shani, Inbal and Kasari, Connie and Lynch, Kevin G. and Oslin, David W. and Pelham, William E. and Fabiano, Gregory and Almirall, Daniel},
  title = {Comparing Dynamic Treatment Regimes Using Repeated-Measures Outcomes: Modeling Considerations in {SMART} Studies},
  journal = {Statistics in Medicine}
}

@book{Keener2010,
  title = {Theoretical Statistics: Topics for a Core Course},
  ISBN = {9780387938394},
  ISSN = {1431-875X},
  DOI = {10.1007/978-0-387-93839-4},
  journal = {Springer Texts in Statistics},
  publisher = {Springer New York},
  author = {Keener,  Robert W.},
  year = {2010}
}

@article{Murphy2001,
  title = {Marginal Mean Models for Dynamic Regimes},
  volume = {96},
  ISSN = {1537-274X},
  url = {http://dx.doi.org/10.1198/016214501753382327},
  DOI = {10.1198/016214501753382327},
  number = {456},
  journal = {Journal of the American Statistical Association},
  publisher = {Informa UK Limited},
  author = {Murphy,  S A and van der Laan,  M J and Robins,  J M},
  year = {2001},
  month = "Dec",
  pages = {1410–1423}
}

@article{Collins2004,
  title = {A Conceptual Framework for Adaptive Preventive Interventions},
  volume = {5},
  ISSN = {1389-4986},
  url = {http://dx.doi.org/10.1023/b:prev.0000037641.26017.00},
  DOI = {10.1023/b:prev.0000037641.26017.00},
  number = {3},
  journal = {Prevention Science},
  publisher = {Springer Science and Business Media LLC},
  author = {Collins,  Linda M. and Murphy,  Susan A. and Bierman,  Karen L.},
  year = {2004},
  month = "Sep",
  pages = {185–196}
}

@article{Murphy2003,
  title = {Optimal Dynamic Treatment Regimes},
  volume = {65},
  ISSN = {1467-9868},
  url = {http://dx.doi.org/10.1111/1467-9868.00389},
  DOI = {10.1111/1467-9868.00389},
  number = {2},
  journal = {Journal of the Royal Statistical Society Series B: Statistical Methodology},
  publisher = {Oxford University Press (OUP)},
  author = {Murphy, S. A.},
  year = {2003},
  month = "Apr",
  pages = {331–355}
}

@article{LIANG1986,
  title = {Longitudinal Data Analysis Using Generalized Linear Models},
  volume = {73},
  ISSN = {1464-3510},
  url = {http://dx.doi.org/10.1093/biomet/73.1.13},
  DOI = {10.1093/biomet/73.1.13},
  number = {1},
  journal = {Biometrika},
  publisher = {Oxford University Press (OUP)},
  author = {Liang, Kung-Yee and Zeger, Scott L.},
  year = {1986},
  pages = {13–22}
}

@article{Kilbourne2018,
 title = {Adaptive School-based Implementation of {CBT} ({ASIC}): Clustered-{SMART} for Building an Optimized Adaptive Implementation Intervention to Improve Uptake of Mental Health Interventions in Schools},
 volume = {13},
 ISSN = {1748-5908},
 url = {http://dx.doi.org/10.1186/s13012-018-0808-8},
 DOI = {10.1186/s13012-018-0808-8},
 number = {1},
 journal = {Implementation Science},
 publisher = {Springer Science and Business Media LLC},
 author = {Kilbourne, Amy M. and Smith, Shawna N. and Choi, Seo Youn and Koschmann, Elizabeth and Liebrecht, Celeste and Rusch, Amy and Abelson, James L. and Eisenberg, Daniel and Himle, Joseph A. and Fitzgerald, Kate and Almirall, Daniel},
 year = {2018},
 month = "Sep" 
}

@article{Quanbeck2020,
  title = {The Balanced Opioid Initiative: Protocol for a Clustered, Sequential, Multiple-Assignment Randomized Trial to Construct an Adaptive Implementation Strategy to Improve Guideline-Concordant Opioid Prescribing in Primary Care},
  volume = {15},
  ISSN = {1748-5908},
  url = {http://dx.doi.org/10.1186/s13012-020-00990-4},
  DOI = {10.1186/s13012-020-00990-4},
  number = {1},
  journal = {Implementation Science},
  publisher = {Springer Science and Business Media LLC},
  author = {Quanbeck, Andrew and Almirall, Daniel and Jacobson, Nora and Brown, Randall T. and Landeck, Jillian K. and Madden, Lynn and Cohen, Andrew and Deyo, Brienna M. F. and Robinson, James and Johnson, Roberta A. and Schumacher, Nicholas},
  year = {2020},
  month = "Apr" 
}

@article{Kilbourne2014,
  title = {Protocol: Adaptive Implementation of Effective Programs Trial ({ADEPT}): Cluster Randomized {SMART} Trial Comparing a Standard Versus Enhanced Implementation Strategy to Improve Outcomes of a Mood Disorders Program},
  volume = {9},
  ISSN = {1748-5908},
  url = {http://dx.doi.org/10.1186/s13012-014-0132-x},
  DOI = {10.1186/s13012-014-0132-x},
  number = {1},
  journal = {Implementation Science},
  publisher = {Springer Science and Business Media LLC},
  author = {Kilbourne, Amy M and Almirall, Daniel and Eisenberg, Daniel and Waxmonsky, Jeanette and Goodrich, David E and Fortney, John C and Kirchner, JoAnn E and Solberg, Leif I and Main, Deborah and Bauer, Mark S and Kyle, Julia and Murphy, Susan A and Nord, Kristina M and Thomas, Marshall R},
  year = {2014},
  month = "Sep" 
}

@article{NahumShani2012,
  title = {Experimental Design and Primary Data Analysis Methods for Comparing Adaptive Interventions.},
  volume = {17},
  ISSN = {1082-989X},
  url = {http://dx.doi.org/10.1037/a0029372},
  DOI = {10.1037/a0029372},
  number = {4},
  journal = {Psychological Methods},
  publisher = {American Psychological Association (APA)},
  author = {Nahum-Shani, Inbal and Qian, Min and Almirall, Daniel and Pelham, William E. and Gnagy, Beth and Fabiano, Gregory A. and Waxmonsky, James G. and Yu, Jihnhee and Murphy, Susan A.},
  year = {2012},
  month = "Dec",
  pages = {457–477}
}

@inbook{Seewald2021,
  title = {Sequential, Multiple Assignment, Randomized Trials ({SMART})},
  ISBN = {9783319526775},
  url = {http://dx.doi.org/10.1007/978-3-319-52677-5_280-1},
  DOI = {10.1007/978-3-319-52677-5_280-1},
  booktitle = {Principles and Practice of Clinical Trials},
  publisher = {Springer International Publishing},
  author = {Seewald, Nicholas J. and Hackworth, Olivia and Almirall, Daniel},
  year = {2021},
  pages = {1–19}
}

@article{Duong2020,
  title = {Rates of Mental Health Service Utilization by Children and Adolescents in Schools and Other Common Service Settings: A Systematic Review and Meta-Analysis},
  volume = {48},
  ISSN = {1573-3289},
  url = {http://dx.doi.org/10.1007/s10488-020-01080-9},
  DOI = {10.1007/s10488-020-01080-9},
  number = {3},
  journal = {Administration and Policy in Mental Health and Mental Health Services Research},
  publisher = {Springer Science and Business Media LLC},
  author = {Duong, Mylien T. and Bruns, Eric J. and Lee, Kristine and Cox, Shanon and Coifman, Jessica and Mayworm, Ashley and Lyon, Aaron R.},
  year = {2020},
  month = "Sep",
  pages = {420–439}
}

@article{Pan_U,
  author  = {Pan, Wenchu and Almirall, Daniel and Kilbourne, Amy M. and Quanbeck, Andrew and Wang, Lu},
  title   = {Finite-Sample Adjustments for Comparing Clustered Adaptive Interventions Using Data from a Clustered {SMART}},
  journal = {arXiv Preprint},
  year    = {2024},
  url     = {https://arxiv.org/abs/2405.00185},
  doi     = {10.48550/arXiv.2405.00185}
}

@article{Fernandez2020,
  title = {Quit{SMART} {U}tah: {A}n Implementation Study Protocol for a Cluster-Randomized, Multi-Level Sequential Multiple Assignment Randomized Trial to Increase Reach and Impact of Tobacco Cessation Treatment in Community Health Centers},
  volume = {15},
  ISSN = {1748-5908},
  url = {http://dx.doi.org/10.1186/s13012-020-0967-2},
  DOI = {10.1186/s13012-020-0967-2},
  number = {1},
  journal = {Implementation Science},
  publisher = {Springer Science and Business Media LLC},
  author = {Fernandez, Maria E. and Schlechter, Chelsey R. and Del Fiol, Guilherme and Gibson, Bryan and Kawamoto, Kensaku and Siaperas, Tracey and Pruhs, Alan and Greene, Tom and Nahum-Shani, Inbal and Schulthies, Sandra and Nelson, Marci and Bohner, Claudia and Kramer, Heidi and Borbolla, Damian and Austin, Sharon and Weir, Charlene and Walker, Timothy W. and Lam, Cho Y. and Wetter, David W.},
  year = {2020},
  month = "Jan" 
}

@article{Zhou2020,
  title = {Adaptive Interventions for Optimizing Malaria Control: an Implementation Study Protocol for a Block-Cluster Randomized, Sequential Multiple Assignment Trial},
  volume = {21},
  ISSN = {1745-6215},
  url = {http://dx.doi.org/10.1186/s13063-020-04573-y},
  DOI = {10.1186/s13063-020-04573-y},
  number = {1},
  journal = {Trials},
  publisher = {Springer Science and Business Media LLC},
  author = {Zhou, Guofa and Lee, Ming-chieh and Atieli, Harrysone E. and Githure, John I. and Githeko, Andrew K. and Kazura, James W. and Yan, Guiyun},
  year = {2020},
  month = "Jul" 
}

@article{NahumShani2020,
  title = {{SMART} Longitudinal Analysis: A Tutorial for Using Repeated Outcome Measures from {SMART} Studies to Compare Adaptive Interventions},
  volume = {25},
  ISSN = {1082-989X},
  url = {http://dx.doi.org/10.1037/met0000219},
  DOI = {10.1037/met0000219},
  number = {1},
  journal = {Psychological Methods},
  publisher = {American Psychological Association (APA)},
  author = {Nahum-Shani, Inbal and Almirall, Daniel and Yap, Jamie R. T. and McKay, James R. and Lynch, Kevin G. and Freiheit, Elizabeth A. and Dziak, John J.},
  year = {2020},
  month = "Feb",
  pages = {1–29}
}

@ARTICLE{Orellana2010-gv,
  title     = "Dynamic Regime Marginal Structural Mean Models for Estimation of Optimal Dynamic Treatment Regimes, Part I: Main Content",
  author    = "Orellana, Liliana and Rotnitzky, Andrea and Robins, James M",
  DOI = {10.2202/1557-4679.1200},
  journal   = "The International Journal of Biostatistics",
  publisher = "Walter de Gruyter GmbH",
  volume    =  6,
  number    =  2,
  pages     = "Article 8",
  year      =  2010,
  language  = "en"
}

@ARTICLE{Orellana2010-rc,
  title     = "Dynamic Regime Marginal Structural Mean Models for Estimation of Optimal Dynamic Treatment Regimes, Part {II}: Proofs of Results",
  author    = "Orellana, Liliana and Rotnitzky, Andrea and Robins, James M",
  journal   = "The International Journal of Biostatistics",
  DOI = {10.2202/1557-4679.1242},
  publisher = "Walter de Gruyter GmbH",
  volume    =  6,
  number    =  2,
  pages     = "Article 9",
  month     =  "Mar",
  year      =  2010,
  language  = "en"
}

@ARTICLE{Robins2008-ua,
  title     = {Estimation and Extrapolation of Optimal Treatment and Testing Strategies},
  author    = "Robins, James and Orellana, Liliana and Rotnitzky, Andrea",
  journal   = "Statistics in Medicine",
  DOI = {10.1002/sim.3301},
  publisher = "Wiley",
  volume    =  27,
  number    =  23,
  pages     = "4678--4721",
  month     =  "Oct",
  year      =  2008,
  language  = "en"
}

@article{Artman2024,
  title = {A Marginal Structural Model for Partial Compliance in {SMART}s},
  volume = {18},
  ISSN = {1932-6157},
  url = {http://dx.doi.org/10.1214/21-AOAS1586},
  DOI = {10.1214/21-aoas1586},
  number = {2},
  journal = {The Annals of Applied Statistics},
  publisher = {Institute of Mathematical Statistics},
  author = {Artman, William J. and Bhattacharya, Indrabati and Ertefaie, Ashkan and Lynch, Kevin G. and McKay, James R. and Johnson, Brent A.},
  year = {2024},
  month = "Jun" 
}

@article{Wank2024,
  title = {A Partially Randomized Patient Preference, Sequential, Multiple‐Assignment, Randomized Trial Design Analyzed via Weighted and Replicated Frequentist and {B}ayesian Methods},
  volume = {43},
  ISSN = {1097-0258},
  url = {http://dx.doi.org/10.1002/sim.10276},
  DOI = {10.1002/sim.10276},
  number = {30},
  journal = {Statistics in Medicine},
  publisher = {Wiley},
  author = {Wank, Marianthie and Medley, Sarah and Tamura, Roy N. and Braun, Thomas M. and Kidwell, Kelley M.},
  year = {2024},
  month = "Nov",
  pages = {5777–5790}
}

@article{Li2016,
  title = {Comparison of Adaptive Treatment Strategies Based on Longitudinal Outcomes in Sequential Multiple Assignment Randomized Trials},
  volume = {36},
  ISSN = {1097-0258},
  url = {http://dx.doi.org/10.1002/sim.7136},
  DOI = {10.1002/sim.7136},
  number = {3},
  journal = {Statistics in Medicine},
  publisher = {Wiley},
  author = {Li, Zhiguo},
  year = {2016},
  month = "Sep",
  pages = {403–415}
}

@article{Luers2019,
  author  = {Luers, Brook and Qian, Min and Nahum-Shani, Inbal and Kasari, Connie and Almirall, Daniel},
  title = {Linear Mixed Models for Comparing Dynamic Treatment Regimens on a Longitudinal Outcome in Sequentially Randomized Trials},
  journal = {arXiv Preprint},
  year    = {2019},
  url = {https://arxiv.org/abs/1910.10078},
  doi = {10.48550/ARXIV.1910.10078},
}

@article{Dziak2019,
  title = {A Data Analysis Method for Using Longitudinal Binary Outcome Data from a {SMART} to Compare Adaptive Interventions},
  volume = {54},
  ISSN = {1532-7906},
  url = {http://dx.doi.org/10.1080/00273171.2018.1558042},
  DOI = {10.1080/00273171.2018.1558042},
  number = {5},
  journal = {Multivariate Behavioral Research},
  publisher = {Informa UK Limited},
  author = {Dziak, John J. and Yap, Jamie R. T. and Almirall, Daniel and McKay, James R. and Lynch, Kevin G. and Nahum-Shani, Inbal},
  year = {2019},
  month = "Jan",
  pages = {613–636}
}

@article{Rubin2005,
  title = {Causal Inference Using Potential Outcomes: Design, Modeling, Decisions},
  volume = {100},
  ISSN = {1537-274X},
  url = {http://dx.doi.org/10.1198/016214504000001880},
  DOI = {10.1198/016214504000001880},
  number = {469},
  journal = {Journal of the American Statistical Association},
  publisher = {Informa UK Limited},
  author = {Rubin, Donald B},
  year = {2005},
  month = "Mar",
  pages = {322–331}
}

@article{Rubin1980,
  title = {Discussion of ``{R}andomization Analysis of Experimental Data: The {F}isher Randomization Test'' by {D}. {B}asu},
  volume = {75},
  ISSN = {0162-1459},
  url = {http://dx.doi.org/10.2307/2287653},
  DOI = {10.2307/2287653},
  number = {371},
  journal = {Journal of the American Statistical Association},
  publisher = {JSTOR},
  author = {Rubin,  Donald B.},
  year = {1980},
  month = "Sep",
  pages = {591}
}

@article{Robins2000,
  title = {Marginal Structural Models and Causal Inference in Epidemiology},
  volume = {11},
  ISSN = {1044-3983},
  url = {http://dx.doi.org/10.1097/00001648-200009000-00011},
  DOI = {10.1097/00001648-200009000-00011},
  number = {5},
  journal = {Epidemiology},
  publisher = {Ovid Technologies (Wolters Kluwer Health)},
  author = {Robins,  James M. and Hernán,  Miguel \'Angel and Brumback,  Babette},
  year = {2000},
  month = "Sep",
  pages = {550–560}
}

@article{Bauer2020,
  title = {Implementation Science: What is it and Why Should {I} Care?},
  volume = {283},
  ISSN = {0165-1781},
  url = {http://dx.doi.org/10.1016/j.psychres.2019.04.025},
  DOI = {10.1016/j.psychres.2019.04.025},
  journal = {Psychiatry Research},
  publisher = {Elsevier BV},
  author = {Bauer,  Mark S. and Kirchner,  JoAnn},
  year = {2020},
  month = "Jan",
  pages = {112376}
}

@inbook{Ghosh2015,
  title = {Chapter 5: Sample Size Calculations for Clustered {SMART} Designs},
  ISBN = {9781611974188},
  url = {http://dx.doi.org/10.1137/1.9781611974188.ch5},
  DOI = {10.1137/1.9781611974188.ch5},
  booktitle = {Adaptive Treatment Strategies in Practice},
  publisher = {Society for Industrial and Applied Mathematics},
  author = {Ghosh,  Palash and Cheung,  Ying Kuen and Chakraborty,  Babas},
  year = {2015},
  month = "Dec",
  pages = {55–70}
}

@article{Chakraborty2014,
  title = {Dynamic Treatment Regimes},
  volume = {1},
  ISSN = {2326-831X},
  url = {http://dx.doi.org/10.1146/annurev-statistics-022513-115553},
  DOI = {10.1146/annurev-statistics-022513-115553},
  number = {1},
  journal = {Annual Review of Statistics and Its Application},
  publisher = {Annual Reviews},
  author = {Chakraborty,  Bibhas and Murphy,  Susan A.},
  year = {2014},
  month = "Jan",
  pages = {447–464}
}

@article{Xu2019,
  title = {{SMART}p: A {SMART} Design for Nonsurgical Treatments of Chronic Periodontitis with Spatially Referenced and Nonrandomly Missing Skewed Outcomes},
  volume = {62},
  ISSN = {1521-4036},
  url = {http://dx.doi.org/10.1002/bimj.201900027},
  DOI = {10.1002/bimj.201900027},
  number = {2},
  journal = {Biometrical Journal},
  publisher = {Wiley},
  author = {Xu,  Jing and Bandyopadhyay,  Dipankar and Salehabadi,  Sedigheh Mirzaei and Michalowicz,  Bryan and Chakraborty,  Bibhas},
  year = {2019},
  month = "Sep",
  pages = {282–310}
}

@article{Rutterford2015,
  title = {Methods for Sample Size Determination in Cluster Randomized Trials},
  volume = {44},
  ISSN = {1464-3685},
  url = {http://dx.doi.org/10.1093/ije/dyv113},
  DOI = {10.1093/ije/dyv113},
  number = {3},
  journal = {International Journal of Epidemiology},
  publisher = {Oxford University Press (OUP)},
  author = {Rutterford,  Clare and Copas,  Andrew and Eldridge,  Sandra},
  year = {2015},
  month = "Jun",
  pages = {1051–1067}
}

@article{Smith2022,
  title = {Primary Aim Results of a Clustered {SMART} for Developing a School-Level, Adaptive Implementation Strategy to Support {CBT} Delivery at High Schools in {M}ichigan},
  volume = {17},
  ISSN = {1748-5908},
  url = {http://dx.doi.org/10.1186/s13012-022-01211-w},
  DOI = {10.1186/s13012-022-01211-w},
  number = {1},
  journal = {Implementation Science},
  publisher = {Springer Science and Business Media LLC},
  author = {Smith, Shawna N. and Almirall, Daniel and Choi, Seo Youn and Koschmann,  Elizabeth and Rusch,  Amy and Bilek,  Emily and Lane,  Annalise and Abelson,  James L. and Eisenberg,  Daniel and Himle,  Joseph A. and Fitzgerald,  Kate D. and Liebrecht,  Celeste and Kilbourne,  Amy M.},
  year = {2022},
  month = "Jul"
}

@article{Laber2014,
  title = {Dynamic Treatment Regimes: Technical Challenges and Applications},
  volume = {8},
  ISSN = {1935-7524},
  url = {http://dx.doi.org/10.1214/14-EJS920},
  DOI = {10.1214/14-ejs920},
  number = {1},
  journal = {Electronic Journal of Statistics},
  publisher = {Institute of Mathematical Statistics},
  author = {Laber,  Eric B. and Lizotte,  Daniel J. and Qian,  Min and Pelham,  William E. and Murphy,  Susan A.},
  year = {2014},
  month = "Jan" 
}

@incollection{Oetting2010,
  author      = "Oetting, Alena I. and Levy, Janet A. and Weiss, Roger D. and Murphy, Susan A.",
  title       = "Statistical Methodology for a {SMART} Design in the Development of Adaptive Treatment Strategies",
  editor      = "Shrout, Patrick and Keyes, Katherine and Ornstein, Katherine",
  booktitle   = "Causality and Psychopathology: Finding the Determinants of Disorders and their Cures",
  publisher   = "Oxford University Press",
  pages       = "179-205",
  year        = 2010
}

@article{Sun2003,
 ISSN = {10170405, 19968507},
 URL = {http://www.jstor.org/stable/24307113},
 author = {Yanqing Sun and Hulin Wu},
 journal = {Statistica Sinica},
 number = {3},
 pages = {593--612},
 publisher = {Institute of Statistical Science, Academia Sinica},
 title = {{AUC}-Based Tests for Nonparametric Functions with Longitudinal Data},
 volume = {13},
 year = {2003}
}

@article{Raudenbush1997,
  title = {Statistical Analysis and Optimal Design for Cluster Randomized Trials},
  volume = {2},
  ISSN = {1082-989X},
  url = {http://dx.doi.org/10.1037/1082-989X.2.2.173},
  DOI = {10.1037/1082-989x.2.2.173},
  number = {2},
  journal = {Psychological Methods},
  publisher = {American Psychological Association (APA)},
  author = {Raudenbush,  Stephen W.},
  year = {1997},
  month = "Jun",
  pages = {173–185}
}

@misc{SMART_IES,
  author = "Nahum-Shani, Inbal and Almirall, Daniel",
  title = "An Introduction to Adaptive Interventions and {SMART} Designs in Education",
  note=  "U.S. Department of Education. Washington, DC: National Center for Special
Education Research",
  howpublished = "NCSER 2020-001",
  year = "2019",
  month = "Nov",
  url = "https://files.eric.ed.gov/fulltext/ED600470.pdf"
}

@article{Offorha2023,
  title = {Analysing Cluster Randomised Controlled Trials Using GLMM,  GEE1,  GEE2,  and QIF: Results from Four Case Studies},
  volume = {23},
  ISSN = {1471-2288},
  url = {http://dx.doi.org/10.1186/s12874-023-02107-z},
  DOI = {10.1186/s12874-023-02107-z},
  number = {1},
  journal = {BMC Medical Research Methodology},
  publisher = {Springer Science and Business Media LLC},
  author = {Offorha,  Bright C. and Walters,  Stephen J. and Jacques,  Richard M.},
  year = {2023},
  month = "Dec"
}

@article{Wang2014,
  title = {Generalized Estimating Equations in Longitudinal Data Analysis: A Review and Recent Developments},
  volume = {2014},
  ISSN = {2314-8314},
  url = {http://dx.doi.org/10.1155/2014/303728},
  DOI = {10.1155/2014/303728},
  journal = {Advances in Statistics},
  publisher = {Wiley},
  author = {Wang,  Ming},
  year = {2014},
  month = "Dec",
  pages = {1–11}
}

@article{Teerenstra2010,
  title = {Sample Size Considerations for {GEE} Analyses of Three‐Level Cluster Randomized Trials},
  volume = {66},
  ISSN = {1541-0420},
  url = {http://dx.doi.org/10.1111/j.1541-0420.2009.01374.x},
  DOI = {10.1111/j.1541-0420.2009.01374.x},
  number = {4},
  journal = {Biometrics},
  publisher = {Oxford University Press (OUP)},
  author = {Teerenstra,  Steven and Lu,  Bing and Preisser,  John S. and van Achterberg,  Theo and Borm,  George F.},
  year = {2010},
  month = "Dec",
  pages = {1230–1237}
}

@article{Huang2021,
  title = {Analyzing Cross-Sectionally Clustered Data Using Generalized Estimating Equations},
  volume = {47},
  ISSN = {1935-1054},
  url = {http://dx.doi.org/10.3102/10769986211017480},
  DOI = {10.3102/10769986211017480},
  number = {1},
  journal = {Journal of Educational and Behavioral Statistics},
  publisher = {American Educational Research Association (AERA)},
  author = {Huang,  Francis L.},
  year = {2021},
  month = "Jun",
  pages = {101–125}
}

@article{Borghi2005,
  title = {Construction of the World Health Organization Child Growth Standards: Selection of Methods for Attained Growth Curves},
  volume = {25},
  ISSN = {1097-0258},
  url = {http://dx.doi.org/10.1002/sim.2227},
  DOI = {10.1002/sim.2227},
  number = {2},
  journal = {Statistics in Medicine},
  publisher = {Wiley},
  author = {Borghi,  E. and de Onis,  M. and Garza,  C. and Van den Broeck,  J. and Frongillo,  E. A. and Grummer-Strawn,  L. and Van Buuren,  S. and Pan,  H. and Molinari,  L. and Martorell,  R. and Onyango,  A. W. and Martines,  J. C.},
  year = {2005},
  pages = {247–265}
}

@article{Azur2011,
  title = {Multiple Imputation by Chained Equations: What is it and How does it Work?},
  volume = {20},
  ISSN = {1557-0657},
  url = {http://dx.doi.org/10.1002/mpr.329},
  DOI = {10.1002/mpr.329},
  number = {1},
  journal = {International Journal of Methods in Psychiatric Research},
  publisher = {Wiley},
  author = {Azur,  Melissa J. and Stuart,  Elizabeth A. and Frangakis,  Constantine and Leaf,  Philip J.},
  year = {2011},
  month = "Feb",
  pages = {40–49}
}

@book{Rubin1987,
  title = {Multiple Imputation for Nonresponse in Surveys},
  ISBN = {9780470316696},
  ISSN = {1940-6347},
  url = {http://dx.doi.org/10.1002/9780470316696},
  DOI = {10.1002/9780470316696},
  journal = {Wiley Series in Probability and Statistics},
  publisher = {Wiley},
  author = {Rubin,  Donald B.},
  year = {1987},
  month = "Jun" 
}

@misc{ExTENd,
  author       = "Oslin, David",
  title        = {Managing Alcoholism in People Who Do Not Respond to Naltrexone ({EXTEND})},
  howpublished = "National Institutes of Health",
  year         = "2005",
  note         = "ClinicalTrials.gov ID NCT00115037"
}

@book{Fitzmaurice2011,
  title = {Applied Longitudinal Analysis},
  ISBN = {9781119513469},
  ISSN = {1940-6347},
  url = {http://dx.doi.org/10.1002/9781119513469},
  DOI = {10.1002/9781119513469},
  journal = {Wiley Series in Probability and Statistics},
  publisher = {Wiley},
  author = {Fitzmaurice,  Garrett M. and Laird,  Nan M. and Ware,  James H.},
  year = {2011},
  month = "Aug" 
}

@article{Wang2023,
  title = {Improving the mixed model for repeated measures to robustly increase precision in randomized trials},
  volume = {20},
  ISSN = {1557-4679},
  url = {http://dx.doi.org/10.1515/ijb-2022-0101},
  DOI = {10.1515/ijb-2022-0101},
  number = {2},
  journal = {The International Journal of Biostatistics},
  publisher = {Walter de Gruyter GmbH},
  author = {Wang,  Bingkai and Du,  Yu},
  year = {2023},
  month = "Nov",
  pages = {585–598}
}

@BOOK{Cohen1988,
  title     = "Statistical Power Analysis for the Behavioral Sciences",
  author    = "Cohen, Jacob",
  publisher = "Routledge Member of the Taylor and Francis Group",
  edition   =  2,
  month     =  "Aug",
  year      =  1988,
  address   = "New York, NY"
}

@article{Raudenbush2001,
  title = {Effects of Study Duration, Frequency of Observation, and Sample Size on Power in Studies of Group Differences in Polynomial Change},
  volume = {6},
  DOI = {10.1037/1082-989x.6.4.387},
  number = {4},
  journal = {Psychological Methods},
  publisher = {American Psychological Association (APA)},
  author = {Raudenbush,  Stephen W. and Liu,  Xiao-Feng},
  year = {2001},
  pages = {387–401}
}
